\documentclass[12pt,oneside,english]{amsart}

\usepackage{courier}

\usepackage[T1]{fontenc}
\usepackage[latin9]{inputenc}
\usepackage{geometry}
\usepackage{babel}
\usepackage{calc}
\usepackage{amsbsy}
\usepackage{amstext}
\usepackage{amsthm}
\usepackage{amssymb}
\usepackage{graphicx}
\usepackage{wasysym}
\usepackage[unicode=true,pdfusetitle,
 bookmarks=true,bookmarksnumbered=false,bookmarksopen=false,
 breaklinks=false,pdfborder={0 0 0},pdfborderstyle={},backref=false,colorlinks=false]
 {hyperref}

\makeatletter

\newcommand{\lyxdot}{.}

\numberwithin{equation}{section}
\numberwithin{figure}{section}
\theoremstyle{plain}
\newtheorem{thm}{\protect\theoremname}[section]
  \theoremstyle{plain}
  \newtheorem{lem}[thm]{\protect\lemmaname}

\usepackage{textcomp}

\makeatother

  \providecommand{\lemmaname}{Lemma}
\providecommand{\theoremname}{Theorem}

\begin{document}

\title{A Guide to the Bott Index and Localizer Index}

\author{Terry A. Loring}

\address{Department of Mathematics and Statistics, University of New Mexico,
Albuquerque, New Mexico 87131, USA}

\email{loring@math.unm.edu}

\begin{abstract}
The Bott index is inherently global. The pseudospectal index is inherently
local, and so now the preferred name is the localizer index. We look
at these on a rather standard model for a Chern insulator, with an
emphasis how to program these effectively. We also discuss how to
tune the localizer index so it behaves like a global index.
\end{abstract}

\maketitle
\tableofcontents{}

\section{Introduction}

The most elegant real-space topological invariant to describe a Chern
insulator is computed on an infinite model, the index of a Fredholm
operator:
\begin{equation}
\mathrm{ind}\left(P\frac{X+iY}{\left|X+iY\right|}P+I-P\right).\label{eq:infinite-area-index}
\end{equation}
Here $X$ and $Y$ represent position operators and $P$ is the Fermi
projection derived via functional calculus from the Hamiltonian $H$.
This does not require translation invariance, just a gap in the spectrum
of $H$ containing the Fermi level. By the work of Bellissard, van
Elst, and Schulz-Baldes \cite{BellissardEtAl_NCG_QuantHall} we know
this equals the Chern number in the case of translation invariant
systems.

For numerical studies we must work with finite matrices. This can
be done in several ways. Some infinite area formulas can be approximated
by finite area calculations. For example, Prodan explains in \cite[\S 6]{ProdanTopInsNCG}
how a finite area approximation leads to a number that, while not
an integer, converges rapidly to the noncommutative Chern number.
Kitaev \cite[\S C.3]{kitaev2006anyons} defined a generalized notion
of the Chern number with a formula to be computed over an infinite
area. Several authors have noted that in many instances a finite-area
approximation converges well to Kitaev's generalized Chern number
\cite{bianco2011mapping,MitchellNashAmorphous}. This method has the
advantage of giving a local index. Thus two lattice models can be
melded on a linear interface and a local (non-integral) real-space
marker can be defined that seems to sensibly mark topological and
trivial areas \cite[Figure 4.d]{MitchellNashAmorphous}.

An alternate approach is to define, via $K$-theory, an integer associated
to a finite model, and attempt to prove that this integer coincides
with the infinite-area index in (\ref{eq:infinite-area-index}). The
$C^{*}$-algebra in this case is typically very elementary, just the
$n$-by-$n$ complex matrices $\boldsymbol{M}_{n}(\mathbb{C})$. In
proving things about such an index, one tends to need complicated
$C^{*}$-algebras defined by generators and relations \cite{Blackadar-shape-theory,EilersLoringContingenciesStableRelations}
and tools like $E$-theory \cite{ConnesHigsonAsymptoticMorphisms}.
Keeping the definitions simple should, in general, lead to fast algorithms.

The essence of this approach is to look to the old theory of what
$K$-theory can tell us about almost commuting matrices \cite{Loring-K-thryAsymCommMatrices,EilersLoringContingenciesStableRelations}.
For example of what was known by the 1990s, consider two unitary matrices
that almost commute. These are sometimes small perturbations of commuting
unitary matrices, sometimes not. It can be shown \cite[Theorem 6.14]{ELP-pushBusby}
that being close or far from commuting unitary matrices can be inferred
from the value of a $K$-theoretic invariant. There are ways to compute
this invariant that do not really look like $K$-theory \cite{ExelLoringInvariats},
and other formulas that are straight-forward computations involving
formulas for idempotents. Indeed, there are many, many formulas we
might choose \cite{EilersLoringContingenciesStableRelations}. 

Finite models force upon us a choice of boundary conditions. Actually,
first one selects a shape, with square being the \emph{de facto }choice.
A frequent choice in physics is to use periodic boundary conditions,
which implicitly turns the square into a torus. This is great mathematically,
being an enabler of the Fourier transform. However, this choice eliminates
edge modes, arguably the most attractive feature of a Chern insulator.
The other choice is working with open boundaries, where on a lattice
model the hopping terms in the infinite-area Hamiltonian are dropped
if they involve sites outside the finite patch. Mathematically, this
is just the obvious compression of the infinite-area Hamiltonian.
This leads to edge modes in the finite model, and is more realistic
than a torus model. This is actually an annoyance in some instances,
as the edge modes can make it harder to see bulk phenomenon in a small
system \cite{loring2018bulk}.

\section{The Clifford spectrum}

The Clifford spectrum gives a nice way to see the difference between
the Bott index and the localizer index. In both cases there is a topological
space associated to a finite system, but the spaces are very different.
All sorts of spaces can emerge as the Clifford spectrum of a finite
collection of Hermitian matrices. This is well established in string
theory, for example in \cite{berenstein2012matrix,sykora2016fuzzy}.
Kisil \cite{kisil1996mobius} observed earlier that the joint spectrum
of a few Hermitian matrices could be a surface.

For a periodic system, square with sides of length $L$, one cannot
easily work with the position operators $X$ and $Y$ as these will
have commutator norm with the Hamiltonian $H$ on the order of $L$.
Thus one considers ``periodic observables''
\[
U=e^{\frac{2\pi i}{L}X},\quad V=e^{\frac{2\pi i}{L}Y}
\]
which really correspond to four observables (Hermtian matrices)
\[
M_{1}=\cos\left(\frac{2\pi i}{L}X\right),\quad M_{2}=\sin\left(\frac{2\pi i}{L}X\right),\quad M_{3}=\cos\left(\frac{2\pi i}{L}Y\right),\quad M_{5}=\sin\left(\frac{2\pi i}{L}Y\right).
\]
One has then a fifth matrix 
\[
M_{5}=P
\]
where $P$ is the spectral projection for $H$ associated to energy
levels below the Fermi gap.

The Clifford spectrum of Hermitian matrices $M_{1},\dots.M_{d}$ (all
of the same size) is a subset $\Lambda(M_{1},\dots,M_{d})$ of $\mathbb{R}^{d}$,
defined as follows. One chooses nontrivial $\Gamma_{1},\dots,\Gamma_{d}$
that form a Clifford representation, here meaning 
\[
\Gamma_{j}^{\dagger}=\Gamma_{j},\quad\Gamma_{j}^{2}=I
\]
for all $j$ and 
\[
\Gamma_{j}\Gamma_{k}=-\Gamma_{k}\Gamma_{j}
\]
 whenever $j\neq k$. Then one forms the \emph{localizer} 
\[
L_{\boldsymbol{\text{\ensuremath{\lambda}}}}(M_{1},\dots,M_{d})=\sum\left(M_{j}-\lambda_{j}\right)\otimes\Gamma_{j}
\]
(called sometimes the localized Dirac operator \cite{sykora2016fuzzy}).
Then
\[
\boldsymbol{\lambda}\in\Lambda(M_{1},\dots,M_{d})\iff L_{\boldsymbol{\text{\ensuremath{\lambda}}}}(M_{1},\dots,M_{d})\text{ is singular}
\]
 defines the Clifford spectrum of $\left(M_{1},\dots.M_{d}\right)$
, and this is independent of the choice of $\Gamma_{j}$.

For reasonable $H$ (one needs some locality, boundedness and a reasonable
gap) when $L$ is large enough, we have $U$, $V$ and $P$ that satisfy
the following:\vspace*{\medskipamount}
\\
\hspace*{0.15\columnwidth}%
\fbox{\begin{minipage}[c]{0.7\columnwidth}%
\[
U^{\dagger}U\approx I,\quad UU^{\dagger}\approx I,\quad V^{\dagger}V\approx I,\quad VV^{\dagger}\approx I,
\]
\begin{equation}
UV=VU,\quad UP\approx PU,\quad VP\approx PV,\label{eq:periodic_relations}
\end{equation}
\[
P^{2}=P,\quad P^{\dagger}=P.
\]
\end{minipage}}\vspace*{\medskipamount}
\\
All this can be made ever so rigorous, as in \cite{LoringCstarRelations}.
What is critical is that if these relations were exact, they would
become equations in $\mathbb{R}^{5}$ that define a copy of the space
\[
\mathbb{T}\sqcup\mathbb{T},
\]
meaning two disjoint copies of the torus. 

What we have then is a fuzzy version of two copies of the torus, and
the Clifford spectrum will be a subset of $\mathbb{R}^{5}$ that is
close to the union of two tori. It is the topology of $\mathbb{T}\sqcup\mathbb{T}$
that is used in defining the Bott index. Perhaps, when the Bott index
is nonzero, and in a clean system on a regular lattice, the Clifford
spectrum of $M_{1},\dots,M_{5}$ is homeomorphic to $\mathbb{T}\sqcup\mathbb{T}$.
Proving this is out of reach, for now.

The localizer index is based on different topology, and it starts
with open boundaries. In this case, $H$ will almost commute with
$X_{0}=\tfrac{2}{L}X$ and $Y_{0}=\tfrac{2}{L}Y$. Notice \emph{we
cannot expect a gap in the spectrum of $H$}. One can spectrally flatten
the infinite-area Hamiltonian and compress that, but the result will
still be ungapped due to edge modes.

Let us assume the system is centered at $0$. The Clifford spectrum
of $(X_{0},Y_{0},H)$ will be some compact subset of $\mathbb{R}^{3}$.
The key relation here is
\[
L_{\boldsymbol{0}}(X_{0},Y_{0},P)\geq c,
\]
where $c>0$ is roughly that spectral gap in the infinite area Hamiltonian.
See \cite{L_S-B_finite_vol} for details on what we know about $c$.
We also have relations\vspace*{\medskipamount}
\\
\hspace*{0.15\columnwidth}%
\fbox{\begin{minipage}[c]{0.7\columnwidth}%
\[
-1\leq X_{0}\leq1,\quad-1\leq Y_{0}\leq1,
\]
\begin{equation}
X_{0}Y_{0}=Y_{0}X_{0},\quad X_{0}H\approx HX_{0},\quad Y_{0}H\approx HY_{0},\label{eq:open_relations}
\end{equation}
\[
-C\leq H\leq C
\]
\end{minipage}}\vspace*{\medskipamount}
\\
where $C$ is the norm of the infinite-area Hamiltonian. We have here
a fuzzy version of 
\[
S(C,c)=\left([-1,1]\times[-1,1]\times[-C,C]\right)\setminus B(c)
\]
where $B(c)$ is the open ball of radius $c$ at the origin. The pseudospectrum
$\Lambda(X_{0},Y_{0},H)$ is expected to be some manner of a surface
within $S(C,c)$. While $S(C,c)$ is homotopic to a sphere, the space
$\Lambda(X_{0},Y_{0},H)$ will be more complicated, and is very difficult
to compute. Still, it will have a $K$-theory class coming ultimately
from the $K$-theory of the sphere which will lead to the localizer
index.

If one uses the spectrally flattened infinite-area Hamiltonian and
compresses that to the finite square system, one has then a Hamiltonian
$H_{1}$ that has ``spectrum localized near $(0,0)$ equal to a fuzzy
$\{0,1\}$'' in the sense put forward in \cite{loring2018bulk}. A
naive expectation is that the Clifford spectrum of $(X_{0},Y_{0},H_{1})$
will be roughly two copies of the square. In the case of trivial insulators,
this is what one should to find. It is the edge states, when the system
is topological, that work to glue together the top and bottom square
to create more of a cube. This is all heurisitic, until we find betters
ways to compute the Clifford spectrum. In the case of the Bott index,
the periodic boundary conditions turn the top and bottom square each
into a torus. In the case of the localizer index, when the system
is topologically non-trivial, the edge states add to the Clifford
spectrum, forming walls that join the top and bottom squares into
something like a sphere.

\begin{figure}
\includegraphics[viewport=25bp 25bp 510bp 380bp,clip,scale=0.7]{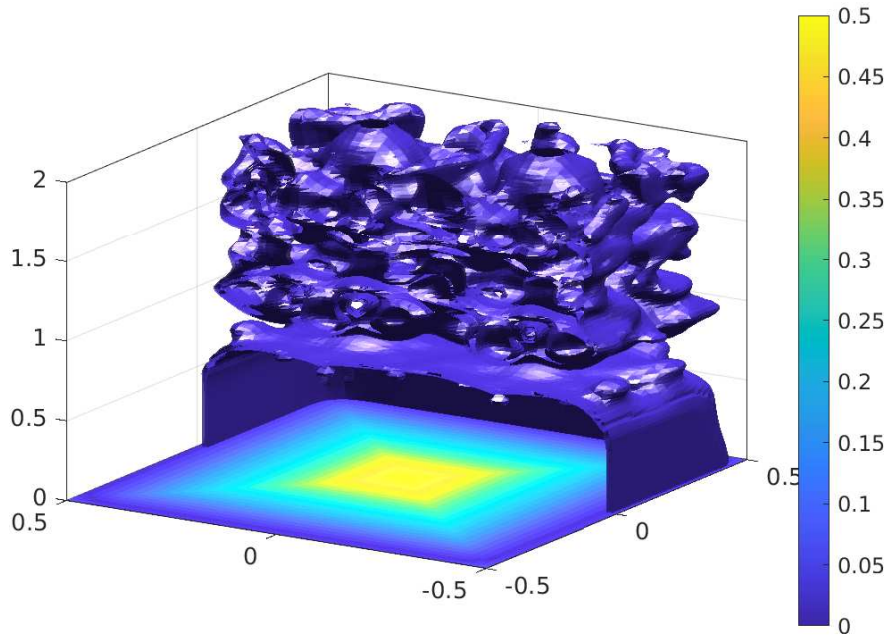}

\caption{Some of the pseuodospectrum of a round Chern insulator. Along the
x-y plane the gap $\left\Vert L_{\boldsymbol{\text{\ensuremath{\lambda}}}}(X_{0},Y_{0},H)^{-1}\right\Vert ^{-1}$
is indicated by color. Also shown are portions of some contour surfaces
for a few gap sizes very close to zero. \label{fig:Pseudospectrum_qc}}
\end{figure}

We can, with lots of computer time, compute the Clifford pseudospectrum
of $(X_{0},Y_{0},H)$. This is essentially a blurred out version of
the Clifford spectrum, and formally is the scalar-valued map on $3$-space
\[
\boldsymbol{\lambda}\mapsto\left\Vert L_{\boldsymbol{\text{\ensuremath{\lambda}}}}(X_{0},Y_{0},H)^{-1}\right\Vert ^{-1}.
\]
Where this is zero is the Clifford spectrum. We can more easily compute
where this is small, and get an estimate on where lives the Clifford
spectrum. 

Now that we are working with just three observables, we can specify
a choice for the $\Gamma_{j}$, specifically $\Gamma_{1}=\sigma_{x}$,
$\Gamma_{2}=\sigma_{y}$, $\Gamma_{3}=\sigma_{z}$. Then
\begin{equation}
L_{\boldsymbol{\text{\ensuremath{\lambda}}}}(M_{1},M_{2},M_{3})=\left(M_{1}-\lambda_{1}\right)\otimes\sigma_{x}+\left(M_{2}-\lambda_{2}\right)\otimes\sigma_{y}+\left(M_{3}-\lambda_{3}\right)\otimes\sigma_{z}\label{eq:Localizer_3Matrices}
\end{equation}
so 
\[
L_{\boldsymbol{\text{\ensuremath{\lambda}}}}(X_{0},Y_{0},H)=\left[\begin{array}{cc}
H-\lambda_{3} & \left(X_{0}-\lambda_{1}\right)-i\left(Y_{0}-\lambda_{2}\right)\\
\left(X_{0}-\lambda_{1}\right)+i\left(Y_{0}-\lambda_{2}\right) & -H+\lambda_{3}
\end{array}\right].
\]

Ideally, the reader would see this computed for the Hamiltonian considered
in the rest of the paper. As the author has consumed already considerable
resources from the Center for Advanced Research Computing at the university
of New Mexico, it seems prudent to offer the reader a different example,
which was computed earlier. Here the Hilbert space is built from a
square sample of a quasilattice, and the Hamiltonian is a variation
on a ``$p_{x}+ip_{y}$'' tight binding model, as in \cite{fulga2016aperiodic,loring2018bulk}.
A portion of the pseodospectrum is indicated in Figure~\ref{fig:Pseudospectrum_qc}.
The Hamiltonian is the compressed original Hamiltonian, not a compression
of the spectrally flattened Hamiltonian. The pseuodospectrum extends
up to about $\lambda_{3}=6$ and down to about $\lambda_{3}=-6$.

The localizer can have a very large gap even when the Hamiltonian
does not. This happens because, by design, the localizer can focus
on ``bulk states'' while ignoring edge states. This can happen even
with strong disorder, as is illustrated by the disorder-averaged study
in Section~\ref{sec:Localizer-as-global} and by the examination
of the localizer spectrum in \cite{V_S_S-B_HalfSignature}.

\begin{figure}
\includegraphics[viewport=25bp 25bp 510bp 380bp,clip,scale=0.7]{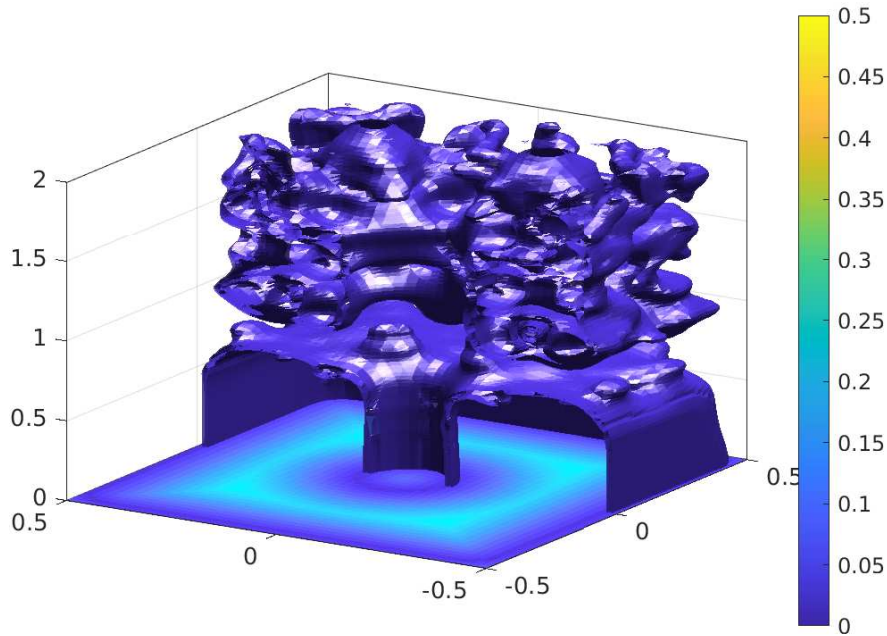}

\caption{Some of the pseuodospectrum of a round Chern insulator with a hole
in it. \label{fig:Pseudospectrum_qc-hole}}
\end{figure}

Unlike the Bott index, the localizer index is flexible in what geometry
is used. For example, one can use a square with a disk cut out of
the middle. In this case, there are additional edge modes and the
result is more like a torus than a sphere. See Figure~\ref{fig:Pseudospectrum_qc-hole}.

\section{Formulas for the Bott index}

The first step in finding an invariant for matrices with the relations
as in Equation~\ref{eq:periodic_relations} is to eliminate the second
torus (the one at high energy). This is the step of considering band-compressed
angular position operators
\[
PUP:P(\mathbb{C}^{n})\rightarrow P(\mathbb{C}^{n})
\]
and
\[
PVP:P(\mathbb{C}^{n})\rightarrow P(\mathbb{C}^{n}).
\]
A more precise notation is $\left.(PUP)\right|_{P\mathbb{C}{}^{n}}$
and $\left.(PVP)\right|_{P\mathbb{C}{}^{n}}$. These are operators,
not matrices. To obtain matrices we must select a basis for the range
$P(\mathbb{C}^{n})$ of $P$. Assuming one obtains the Fermi projector
$P$ by doing a full diagonalization of the Hamiltonian $H\in\boldsymbol{M}_{n}(\mathbb{C})$
such a basis is then available at no further cost.

This is a sticking point. Taking a full diagonalization of $H$ forces
us to work with dense matrices. As we are dealing with approximate
relations, we could work with an approximation $P_{0}$ to $P$ that
is sparse. This is a standard procedure in quantum chemistry \cite{benzi2013decay}.
However, there are other parts to the Bott index formula where a sparse
approach is not at all clear. Also, we need sparse approximate spectral
flattening given only a mobility gap, which will be harder than the
case of a true gap. For now, the Bott index is a dense matrix formula,
which is its major limitation.

Suppose then that you have used an eigensolver on $H$ to find for
$\mathbb{C}^{n}$ a basis $\boldsymbol{q}_{1},\dots,\boldsymbol{q}_{n}$
of eigenvalues, with $\boldsymbol{q}_{1},\dots,\boldsymbol{q}_{m}$
in the subspace below the Fermi level and $\boldsymbol{q}_{m+1},\dots,\boldsymbol{q}_{n}$
above. (In a disordered system, there may be no apparent gap, but
one still is selecting a Fermi level.) Then $\boldsymbol{q}_{1},\dots,\boldsymbol{q}_{m}$
is a basis for $P(\mathbb{C}^{n})$. The matrices, with respect to
this basis, for the operators $\left.(PUP)\right|_{P\mathbb{C}{}^{n}}$
and $\left.(PVP)\right|_{P\mathbb{C}{}^{n}}$ are
\begin{equation}
U_{1}=W^{\dagger}UW,\quad V_{1}=W^{\dagger}VW,\label{eq:smaller_matrices}
\end{equation}
where
\[
W=\left[\begin{array}{ccc}
\brokenvert &  & \brokenvert\\
\brokenvert &  & \brokenvert\\
\boldsymbol{q}_{1} & \cdots & \boldsymbol{q}_{m}\\
\brokenvert &  & \brokenvert\\
\brokenvert &  & \brokenvert
\end{array}\right].
\]
Notice $W$ is a partial isometry. 

The smaller matrices $U_{1}$ and $V_{1}$ satisfy the more familiar
relations\vspace*{\medskipamount}
\\
\hspace*{0.15\columnwidth}%
\fbox{\begin{minipage}[c]{0.7\columnwidth}%
\[
U_{1}^{\dagger}U_{1}\approx I,\quad U_{1}U_{1}^{\dagger}\approx I,\quad V_{1}^{\dagger}V_{1}\approx I,\quad V_{1}V_{1}^{\dagger}\approx I,
\]
\begin{equation}
U_{1}V_{1}\approx V_{1}U_{1}\label{eq:periodic_smaller}
\end{equation}
\end{minipage}}\vspace*{\medskipamount}
\\
There is a theory on how to create a $K$-theoretic index for such
approximate relations \cite{EilersLoringContingenciesStableRelations}.
Such an index is never unique. All we can say is that for very small
commutators two such indices will agree. We can't usually say how
small the commutator must be for two such invariants to agree. Some
of these choices for invariants will make more sense physically than
others. Some will be faster to compute than others.

The general method here starts with a formula for a projector-valued
function on the torus, and 2-by-2 matrices suffice. For example 
\begin{equation}
(w,z)\mapsto\left[\begin{array}{cc}
f(z) & g(z)+ih(z)w\\
g(z)-ih(z)w & 1-f(z)
\end{array}\right]\label{eq:map_to_projection}
\end{equation}
for $|v|=|w|=1$. In fact, this need not take values only in projectors:
approximate projectors ($p^{2}\approx p)$ will be fine. Thus we assume
$fg\approx0$ and $f^{2}+g^{2}+h^{2}\approx f$. The formulas used
must make sense when applied to noncommuting, nonunitary matrices.
For example, one can take $f$, $g$ and $h$ to be Laurant polynomials
in $z$ (so trig polynomials in the argument of $z$). Then the $K$-theoretic
index is the number of eigenvalues above $\tfrac{1}{2}$, minus the
number below, of the matrix 
\[
\left[\begin{array}{cc}
f(V_{1}) & g(V_{1})+ih(V_{1})U_{1}\\
g(V_{1})-iU_{1}^{\dagger}h(V_{1}) & 1-f(V_{1})
\end{array}\right].
\]
This is a simplified version of what was called the trig method in
\cite{LorHastHgTe}.

There are many variation possible on this. Moreover one can work directly
with $U$, $V$ and a sparse approximation $P$ and use the theory
of \cite{EilersLoringContingenciesStableRelations} for the space
$\mathbb{T}\sqcup\mathbb{T}$. This might lead to faster algorithm
than is a available to the Bott index, in theory. 

An issue with all invariants using this method is that they do not
treat all parts of a square lattice equally. This is because the function
in (\ref{eq:map_to_projection}) is a continuous function from the
torus to the Bloch sphere. Every such map has at least one point on
the sphere where some nonzero area of the torus is smashed together
sent to a region of zero area on the sphere.

In the special case of the approximate relations in Equation~\ref{eq:periodic_smaller},
there is an elegant alternative to produce indices proven equivalent
(equal for small commutators) to an index as above. For the case of
almost commuting unitary matrices, this was proven a while back \cite{ExelLoringInvariats}.
The formula was modified to work for almost unitary matrices that
almost commuting more recently \cite{LorHastHgTe}. It also treats
$U_{1}$ and $V_{1}$ symmetrically, and treats all areas in the lattice
equally (if the lattice is treated as on a torus). This alternate
formula is what physicists now call the Bott index.

The Bott index uses the smaller matrices from Equation~\ref{eq:smaller_matrices}
and is
\begin{equation}
\mathrm{Bott}(U,V,P)=\Re\left(\frac{1}{2\pi i}\text{Tr}\left(\log\left(U_{1}V_{1}U_{1}^{\dagger}V_{1}^{\dagger}\right)\right)\right)\label{eq:Bott_Index_Formula}
\end{equation}
where $\Re$ indicates real part. It is guaranteed, in theory, to
makes sense, and be in integer, when $(\infty,0]$ is not in the spectrum
of $U_{1}V_{1}U_{1}^{\dagger}V_{1}^{\dagger}$. In practice eigenvalues
of $U_{1}V_{1}U_{1}^{\dagger}V_{1}^{\dagger}$ very near $(\infty,0]$
also cause trouble. This formula uses the matrix logarithm, and \emph{you
don't want to compute a matrix logarithm} as that takes a long time.
\begin{lem}
If $X$ is a matrix with no spectrum lying in the set $(-\infty,0)$,
then
\[
\Re\left(\frac{1}{2\pi i}\text{Tr}\left(\log\left(X\right)\right)\right)=\sum\Re\left(\frac{1}{2\pi i}\log\left(\frac{\lambda}{|\lambda|}\right)\right)
\]
where $\lambda_{1},\dots,\lambda_{n}$ are the eigenvalues of $X$
listed according to multiplicity. In all instances the usual branch
of log is assumed.
\end{lem}

\begin{proof}
This follows easily from the spectral mapping theorem for analytic
functions.
\end{proof}
Notice we have no need for any eigenvectors of $DCD^{\dagger}C^{\dagger}$,
just the list of eigenvalues. After multiplying matrices to form the
commutator, you the calculate all the eigenvalues, take the scalar
logarithm of each, and add those. The heart of the algorithm is as
follows, when implemented in Matlab. 
\begin{quotation}
\noindent \texttt{U = W\textquotesingle {*}exp\_x{*}W;}

\noindent \texttt{V = W\textquotesingle {*}exp\_y{*}W; }

\noindent \texttt{T = eig(commutator); \% Just a list of eigenvalues}

\noindent \texttt{index = -sum(imag(log(T)))/(2{*}pi); }

\noindent \texttt{index = round(index);}
\end{quotation}
Here \texttt{W} is the partial matrix of eigenvectors and \texttt{exp\_x}
and \texttt{exp\_y} are the normalize exponentiation position matrices
$U$ and $V$. Also and \texttt{U} and \texttt{V} are what have been
denoted $U_{1}$ and $V_{1}$. See the file \texttt{oldBottNoGap.m}
in the supplementary files \cite{supplement}.

It could be that someone has computed the projector $P$ onto the
subspace of states below the gap, but not a orthonormal basis for
that space. This will happen if you use the methods of quantum chemistry
\cite{benzi2013decay} to approximately compute the Fermi projector
$P$. In that case, one can use the almost commuting matrices 
\begin{equation}
A=PUP+(I-P)\text{ and }B=PVP+(I-P)\label{eq:larger_matrices}
\end{equation}
and it can be shown that
\[
\mathrm{Bott}(U,V,P)=\Re\left(\frac{1}{2\pi i}\text{Tr}\left(\log\left(BAB^{\dagger}A^{\dagger}\right)\right)\right).
\]
Since $A$ and $B$ are larger than $U_{1}$ and $V_{1}$ then this
may end up as the slower method, but perhaps someone knows how to
compute the trace of the log of a sparse matrix quickly. Without a
method to deal with that trace of a matrix log, one is advised to
work only with $U_{1}$ and $V_{1}$.

\section{Extending an old study \label{sec:Extending-old_study}}

\begin{figure}
\includegraphics[viewport=20bp 0bp 410bp 275bp,clip,scale=0.6]{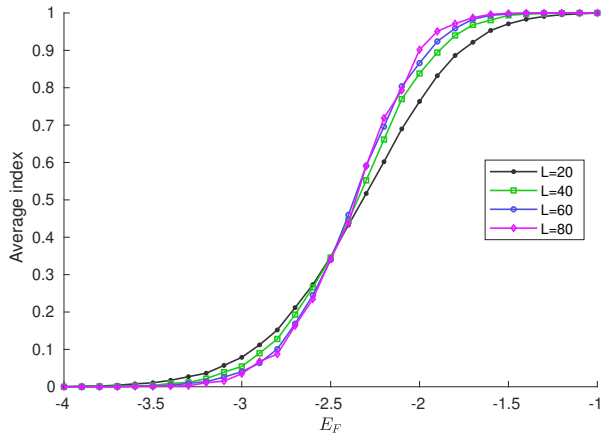}

\caption{Bott index averaged over disorder as a function of the Fermi level:
$L=20$ (12,664 samples); $L=40$ (7,681 samples); $L=60$ (3,013
samples); $L=80$ (1,121 samples).\label{fig:Bott-index-averaged}}

\end{figure}

For the purposes of evaluating the algorithm of the Bott index, we
replicate and extend the numerical study in \cite{LorHastHgTe} (see
also \cite{LoringPseudospectra}). The reason to select a known model
Hamiltonian and disorder-induced transition is to keep the focus here
on the methods for computing the Bott index.

We start with a regular square lattice, with periodic boundary conditions,
with two basis elements at each site. Upon this we consider a real-space
tight binding model, associating a Hamiltonian term to each site of
the tiling and a hopping matrix to each link between closest neighboring
sites. The on-site element corresponding to site $j$ is 
\[
H_{j}=(M-4B+W_{j})I_{2}
\]
where $M$ and $B$ are constants, and $W_{j}$ varies by site, each
$W_{j}$ independently taken from a uniform distribution in $[-4,4]$.
The full hopping term between sites $j$ and $k$, for $j$ to the
left of $k$ is
\[
H_{jk}=B\sigma_{z}+A\sigma_{x}
\]
for $j$ below $k$ is
\[
H_{jk}=B\sigma_{z}+A\sigma_{y}
\]
The constants we will use are
\[
A=1,\ B=-1,\ C=0,\ M=-2
\]
which is as was done in producing \cite[Fig.~1]{LorHastHgTe}. 

In other parts of \cite{LorHastHgTe} $M=-1$ was sometimes used.
The meaning here of the constant $A$ is off by $\tfrac{1}{2}$ how
how it is used in \cite{jiang2009numerical}. We continue to use the
convention from \cite{LorHastHgTe} as listed above, and in any case
the reader is encouraged to look at the file \texttt{ChernSystem.m
}in the supplementary files \cite{supplement} to see how the Hamiltonian
and lattice are constructed.

\begin{figure}
\includegraphics[viewport=15bp 5bp 270bp 210bp,clip,scale=0.8]{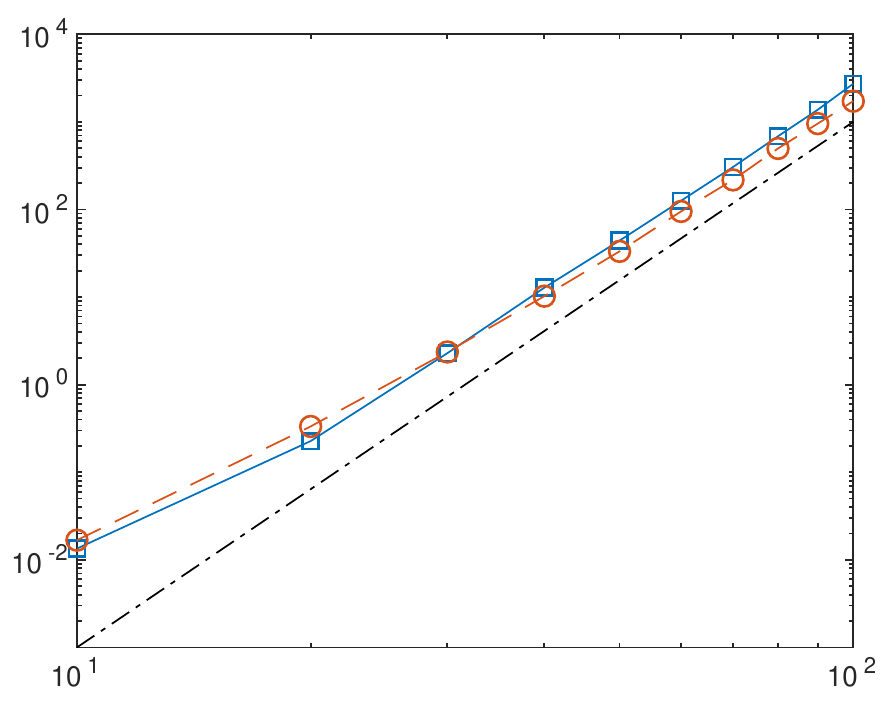}
\includegraphics[viewport=15bp 5bp 270bp 210bp,clip,scale=0.8]{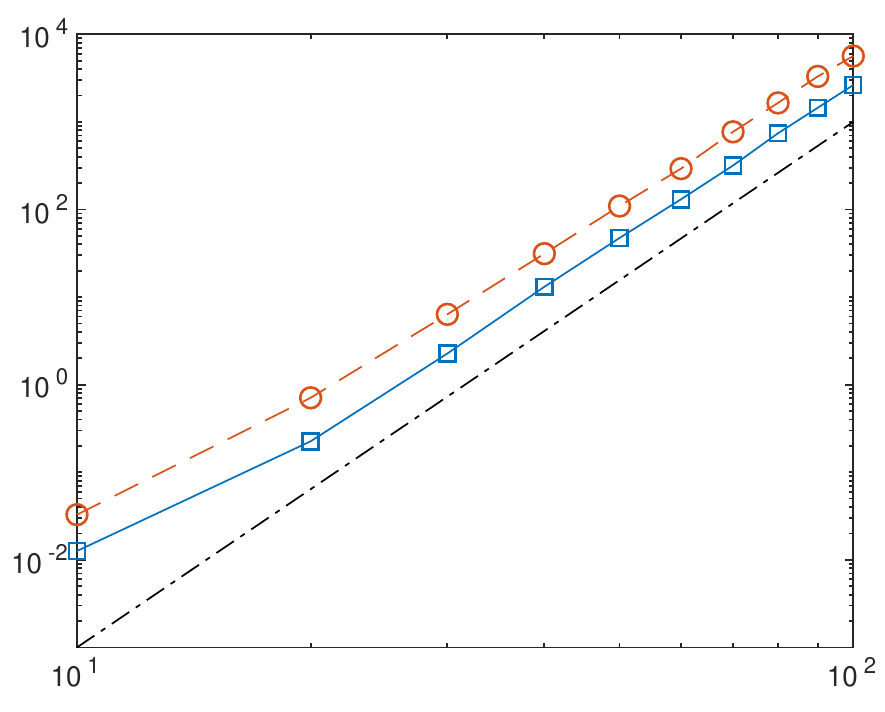}

\caption{For $L=10$ to $L=100$ the average time. Left panel is the preferred
method using the matrices as in equation~\ref{eq:periodic_smaller},
right panel the method using the larger matrices described in Equation~\ref{eq:larger_matrices}.
The solid line with squares indicates the time in seconds needed for
a full diagonalization of the Hamiltonian. The dashed line indicates
the time in seconds for the rest of the calculation, with the energy
cutoff a bit below the Fermi level. The testing was performed on a
8-core computer with each core rated at 2.67GHz. For reference, the
dotted line plots the curve $y=\left(10^{-9}\right)L^{6}$. These
plots are based on computing one Bott index for each of 10 disordered
Hamiltonians for each value of $L$. \label{fig:Time_Bott_index}}
\end{figure}

As in \cite{LorHastHgTe} we vary the Fermi level for many Hamiltonians
with different disorder terms. The index averaged over disorder that
resulted is shown in Figure~\ref{fig:Bott-index-averaged}. The improvments
in computing hardware over the past twelve years gave a moderate improvement,
as now it is feasable to work up to $L=80$ instead of only $L=60$,
as in \cite[Fig.~1]{LorHastHgTe}. Also, more samples were computed,
so the resulting plots are clearer.

As we are computing all eigenvalues of some dense matrices, the time
needed to compute the Bott index is expected to grow a the cube of
system size, so $\mathcal{O}(L^{6})$. Multiple values of the Bott
index, for multiple Fermi levels, can be computed from the same diagonalization
of the Hamiltonian, so in timing this algorithm it makes sense to
time the Hamiltonian diagonalization separately from the time to compute
a single Bott index given that diagonalization. As one can see from
the left panel of Figure~\ref{fig:Time_Bott_index}, both tasks do
take time that grows in what looks like $\mathcal{O}(L^{6})$. 

\begin{figure}
\includegraphics[viewport=10bp 0bp 400bp 303bp,clip,scale=0.6]{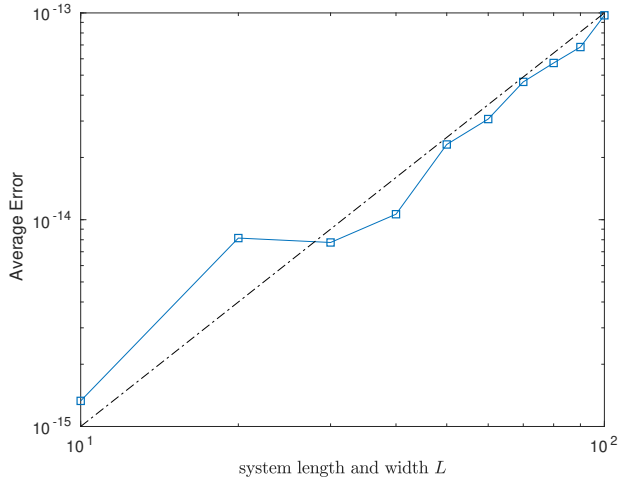}

\caption{For $L=10$ to to $L=100$ the average error, defined as the distance
to the closest integer. The solid line is the average error while,
for reference, the dotted line plots the curve $y=\left(0.4\times10^{-17}\right)L^{2}$.
This plot is based on 32 samples for each value of $L$. \label{fig:Error_Bott_index}}
\end{figure}

As a comparison, the timing is reported (right panel of Figure~\ref{fig:Time_Bott_index})
as well for the algorithm that uses the larger matrices of Equation~\ref{eq:larger_matrices}.
There is accumulated numerical error that keeps the result from being
exactly an integer. This is generally ignored; one just rounds to
the nearest integer. In Figure~\ref{fig:Error_Bott_index} is shown
the distance to the nearest integer. This seems like it will be negligible
even for systems so large that computing a single Bott index will
be impractical.

It is worth pointing out that we do not, it seems, have a theory that
tells us what the correct answer is here for the disordered averaged
Bott index. However, the Bott index is known to correlate, in the
gapped case, to the Kubo conductance \cite[Lemma~5.6]{HastLorTheoryPractice}.
It has been used and evaluated in many settings, such as \cite{AlessioRigolFloquetChern,bandres2016topological,toniolo2018time}.
What Figure~\ref{fig:Error_Bott_index} is showing is that the computed
Bott index come out very close to an integer, the error coming from
the cumulative errors in calculating eigenvalues numerically.

The plot of the time needed makes it seem that computing a Bott index
for a 100-by-100 system is reasonable. If computing an average over
disorder, and at multiple Fermi levels, this is not the case, as the
total time needed is excessive. On the other hand, since the number
of samples needed decreases with system size, the growth in that case
is more like $\mathcal{O}(L^{5})$. That is still a problem, so if
one wants to push this study further, one needs to be doing sparse
matrix computations.

\section{The localizer index \label{sec:Localizer_local_version}}

The localizer index works only with open boundaries. Unlike the Bott
index, it can easily be modified for different symmetry classes and
physical dimensions \cite{LoringPseudospectra}. Here we focus on
2D system in class AI. Several papers \cite{fulga2016aperiodic} referred
to this index as the pseudospectrum index, but that is misleading
as the index is defined where the pseudospectrum isn't.

We have three observables $(X,Y,H)$ with $X$ and $Y$ position and
$H$ the Hamiltonian. We could change units in $X$ and $Y$ and would
still have three observables, now $(\kappa X,\kappa Y,H)$ for some
nonzero, finite scaling constant $\kappa$. If $\boldsymbol{\lambda}$
is not in $\Lambda(\kappa X,\kappa Y,H)$ then the \emph{localizer
index at} $\boldsymbol{\lambda}$ \emph{for} $(\kappa X,\kappa Y,H)$
is
\[
\frac{1}{2}\mathrm{Sig}\left(L_{\boldsymbol{\lambda}}(\kappa X,\kappa Y,H)\right),
\]
where $L_{\boldsymbol{\lambda}}$ is the localizer of Equation~\ref{eq:Localizer_3Matrices}
and $\mathrm{Sig}(M)$ refers to the number of positive eigenvalues
minus the number of negative eigenvalues of the nonsingular Hermitian
matrix $M$.

In \cite{LoringPseudospectra} the localizer was referred to as the
Bott element since it does give a representative of the Bott element
when applied to standard position coordinates on the sphere. Now that
the Bott index is rather firmly entrenced in physics, it seems prudent
to call it the localizer. This reflects the fact that one can use
the localizer to differentiate edge from bulk states \cite{loring2018bulk}.

It is important to not compute all the positive or all the negative
eigenvalues of the localizer when computing the index. That would
defeat the advantage we gain from the fact that $L_{\boldsymbol{\lambda}}(\kappa X,\kappa Y,H)$
should be a sparse matrix. One might need to zero-out small hopping
terms between distant sites to ensure $H$ is sparse. In general $\kappa X$
are $\kappa Y$ are diagonal, and so very sparse.

If $M$ is nonsingular and Hermitian, one computes an LDLT decomposition
of $M$ and uses Sylevester's law of inertial to compute the signature
of $M$. This means, in practice, on finds a factorization
\[
M=S^{-1}PLDL^{\dagger}P^{\dagger}S^{-1}
\]
where $S$ is diagonal with positive values on the diagonal, $P$
is a permutation matrix, $L$ is lower triangular, and $D$ is block-diagonal,
with both 1-by-1 and 2-by-2 blocks. Since
\[
M=(S^{-1}PL)D(S^{-1}PL)^{\dagger}
\]
we have 
\[
\mathrm{Sig}(M)=\mathrm{Sig}(D)
\]
and the signature of $D$ is the sum of the signatures of the blocks
of $D$. 

\begin{figure}
\includegraphics[viewport=5bp 0bp 270bp 205bp,clip,scale=0.8]{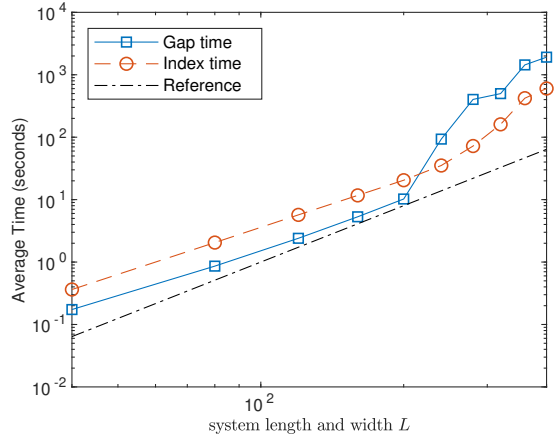}

\caption{For $L=40$ to $L=400$ the average time for the localizer method,
time to compute the gap and to compute the index plotted separately.
The testing was performed on a 8-core computer with each core rated
at 2.67GHz. For reference, the dotted line plots the curve $y=\left(5\times10^{-9}\right)L^{4}$.
These plots are based one Fermi level of of 10 disordered Hamiltonians
for each value of $L$. \label{fig:Timing_localizer}}
\end{figure}

Unfortunately, we require LDLT for a sparse complex Hermitian matrices,
while Matlab \cite{Matlab_LDL} and MUMPS \cite{MUMPS_guide} only
allow for sparse real symmetric matrices. A workaround is based on
the usual embedding
\[
a+ib\mapsto\left[\begin{array}{cc}
a & b\\
-b & a
\end{array}\right]
\]
of the reals into complex matrices. As such, one computer the signature
of
\[
\left[\begin{array}{cc}
\Re(M) & \Im(M)\\
-\Im(M) & \Re(M)
\end{array}\right]
\]
and divides by two. See \texttt{signature.m }in the supplementary
files \cite{supplement}.

\begin{figure}
\noindent {\footnotesize{}\makebox[10cm][l]{\raisebox{-0.08cm}[0.0cm][0.1cm]{$\kappa=5$, $E=0$\hspace*{4.4cm}$\kappa=2$, $E=0$}}}\\
\includegraphics[viewport=65bp 25bp 490bp 400bp,clip,scale=0.38]{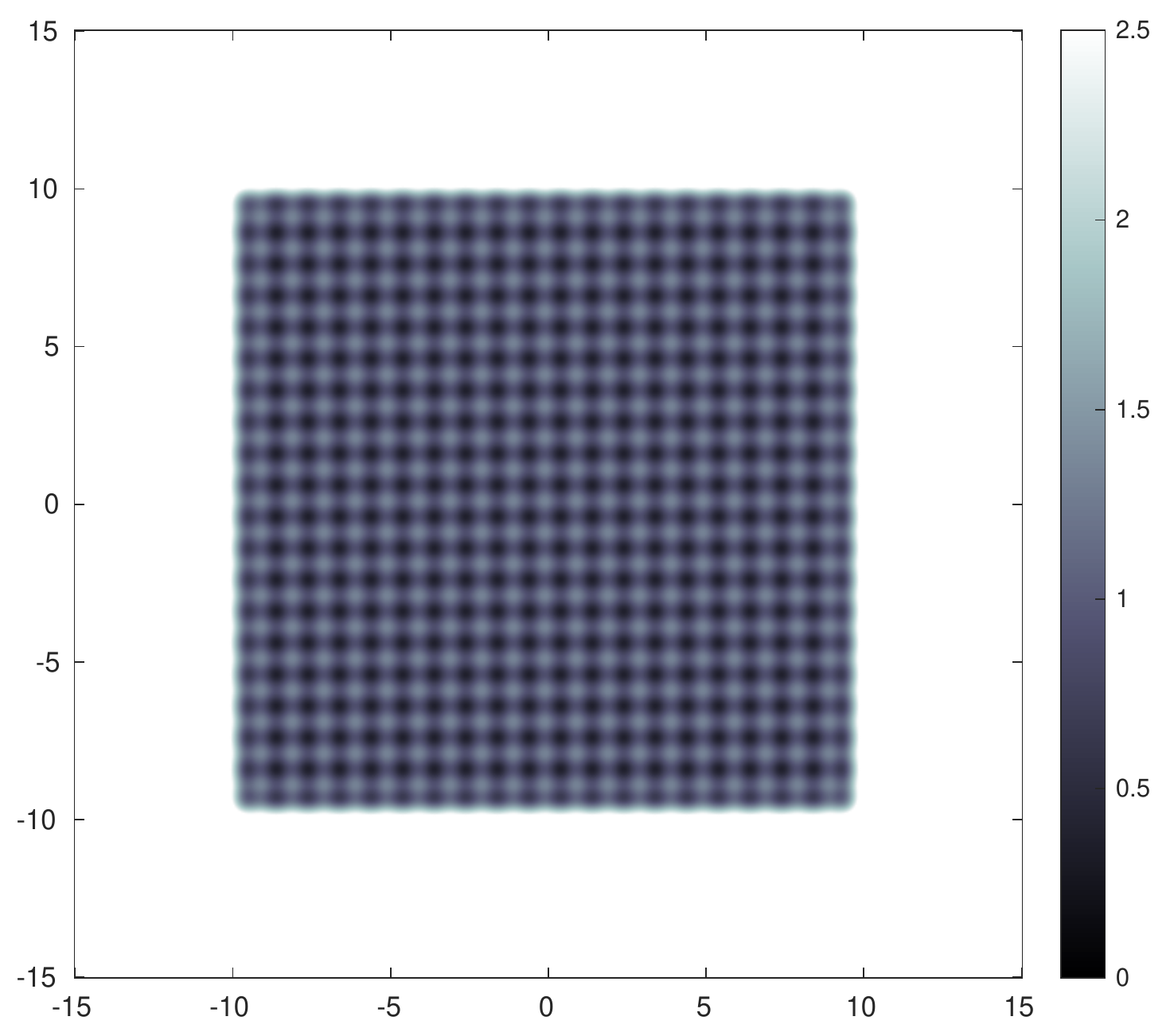}\quad{}\includegraphics[viewport=65bp 25bp 490bp 400bp,clip,scale=0.38]{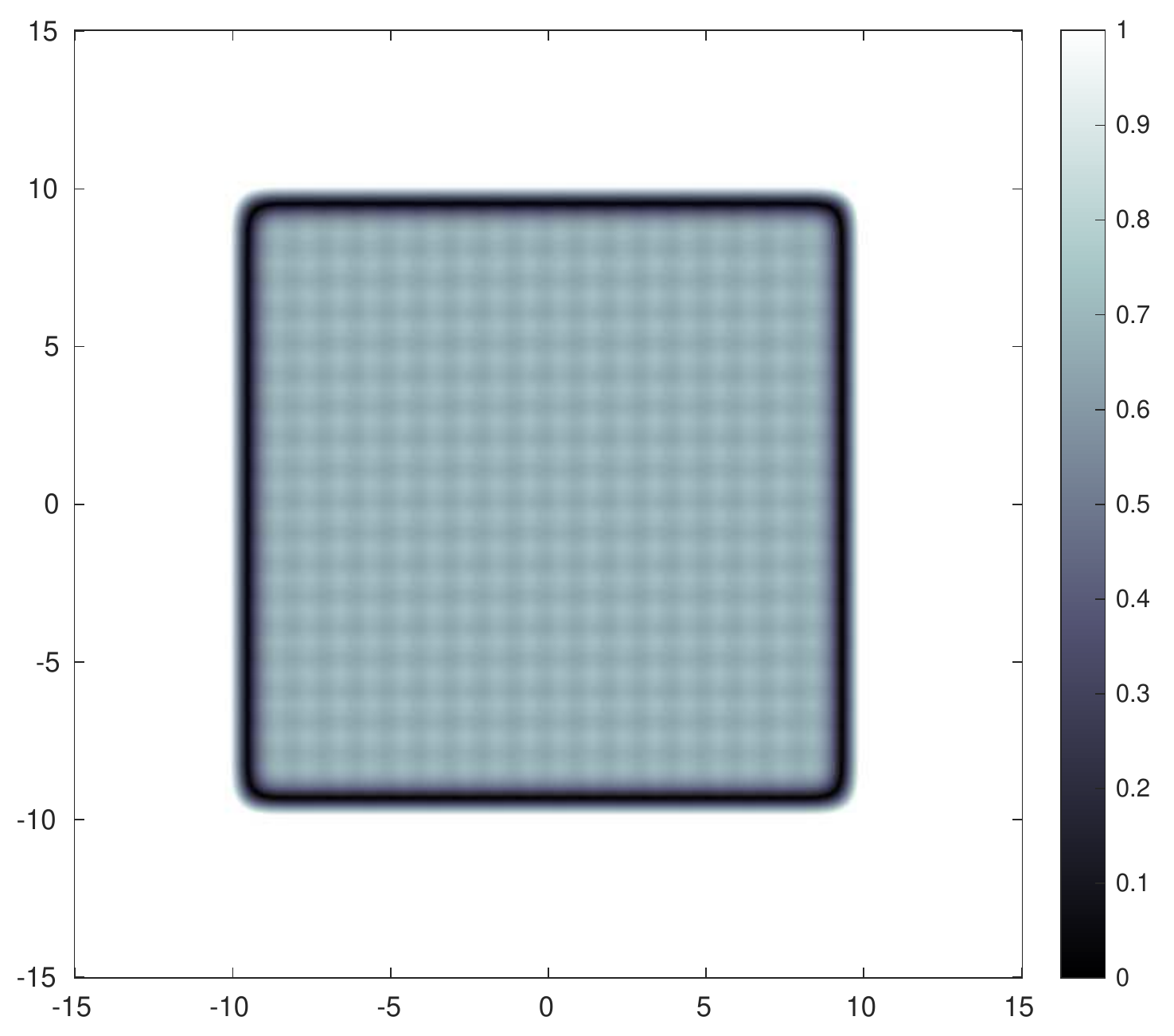}\vspace*{-0.3cm}
\\
{\footnotesize{}\makebox[10cm][l]{\raisebox{-0.08cm}[0.0cm][0.1cm]{$\kappa=1$, $E=0$\hspace*{4.4cm}$\kappa=0.5$, $E=0$}}}\\
\includegraphics[viewport=65bp 25bp 490bp 400bp,clip,scale=0.38]{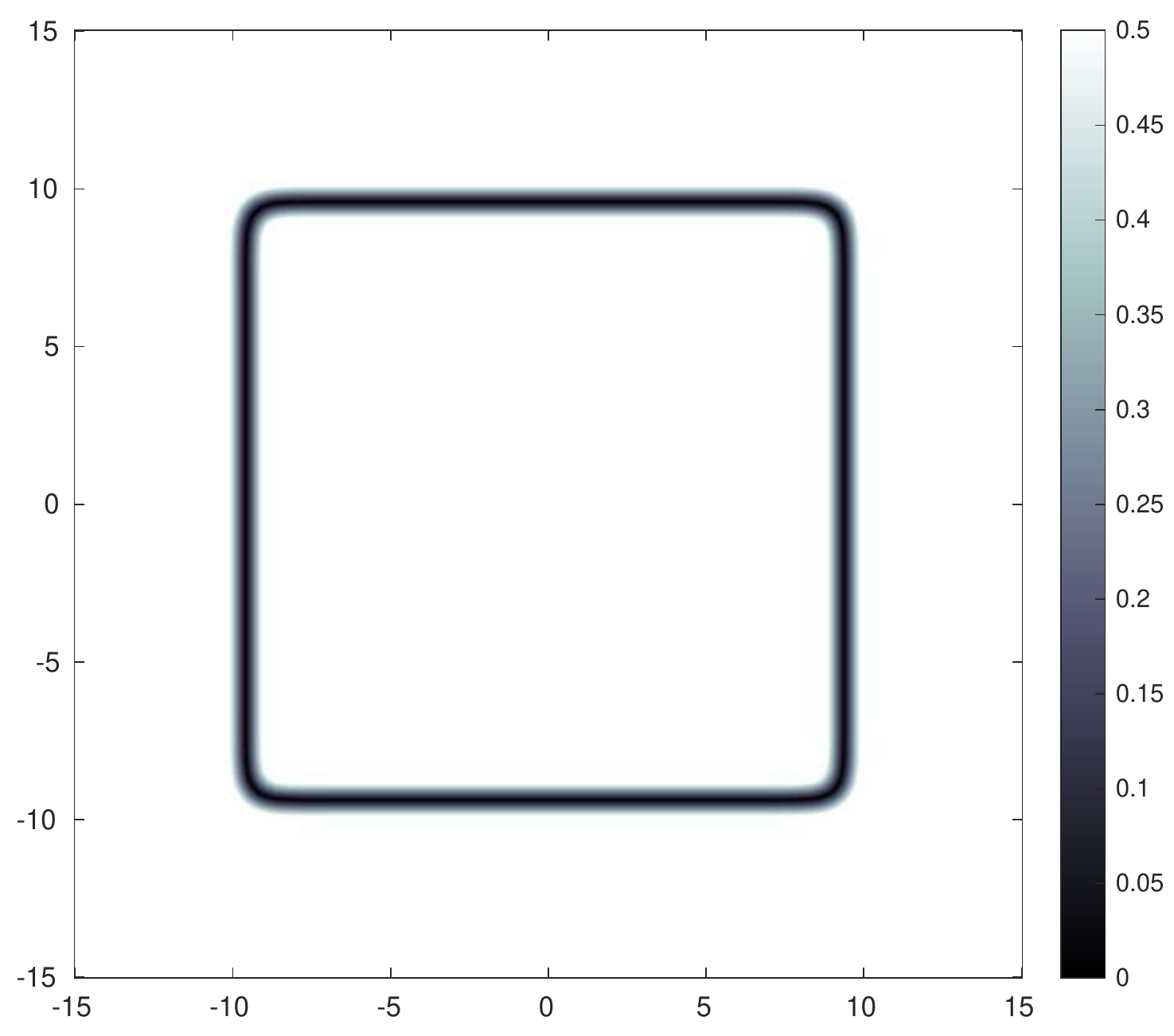}\quad{}\includegraphics[viewport=65bp 25bp 490bp 400bp,clip,scale=0.38]{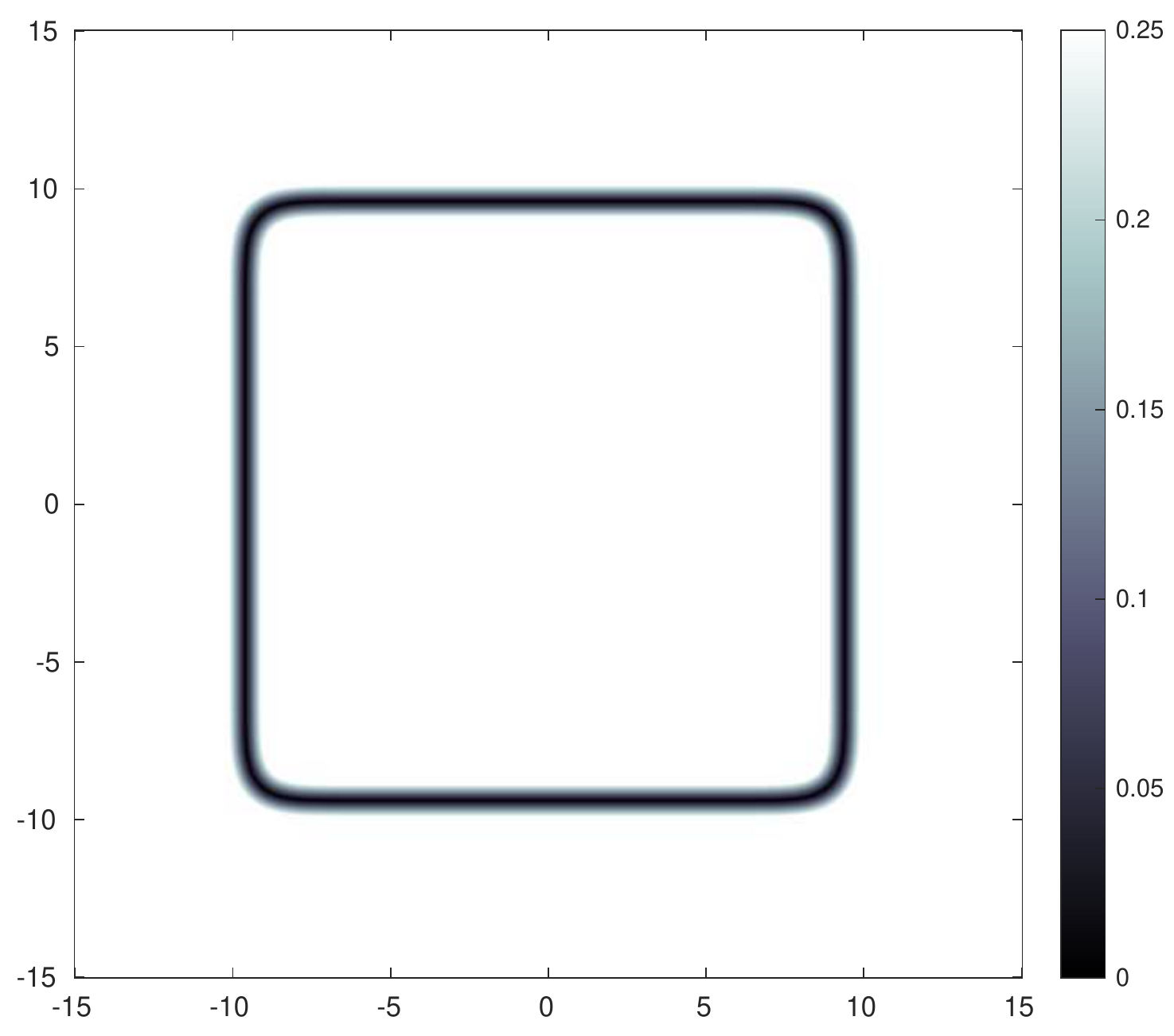}\vspace*{-0.3cm}
\\
{\footnotesize{}\makebox[10cm][l]{\raisebox{-0.08cm}[0.0cm][0.1cm]{$\kappa=0.2$, $E=0$\hspace*{4.4cm}$\kappa=0.1$, $E=0$}}}\\
\includegraphics[viewport=65bp 25bp 490bp 400bp,clip,scale=0.38]{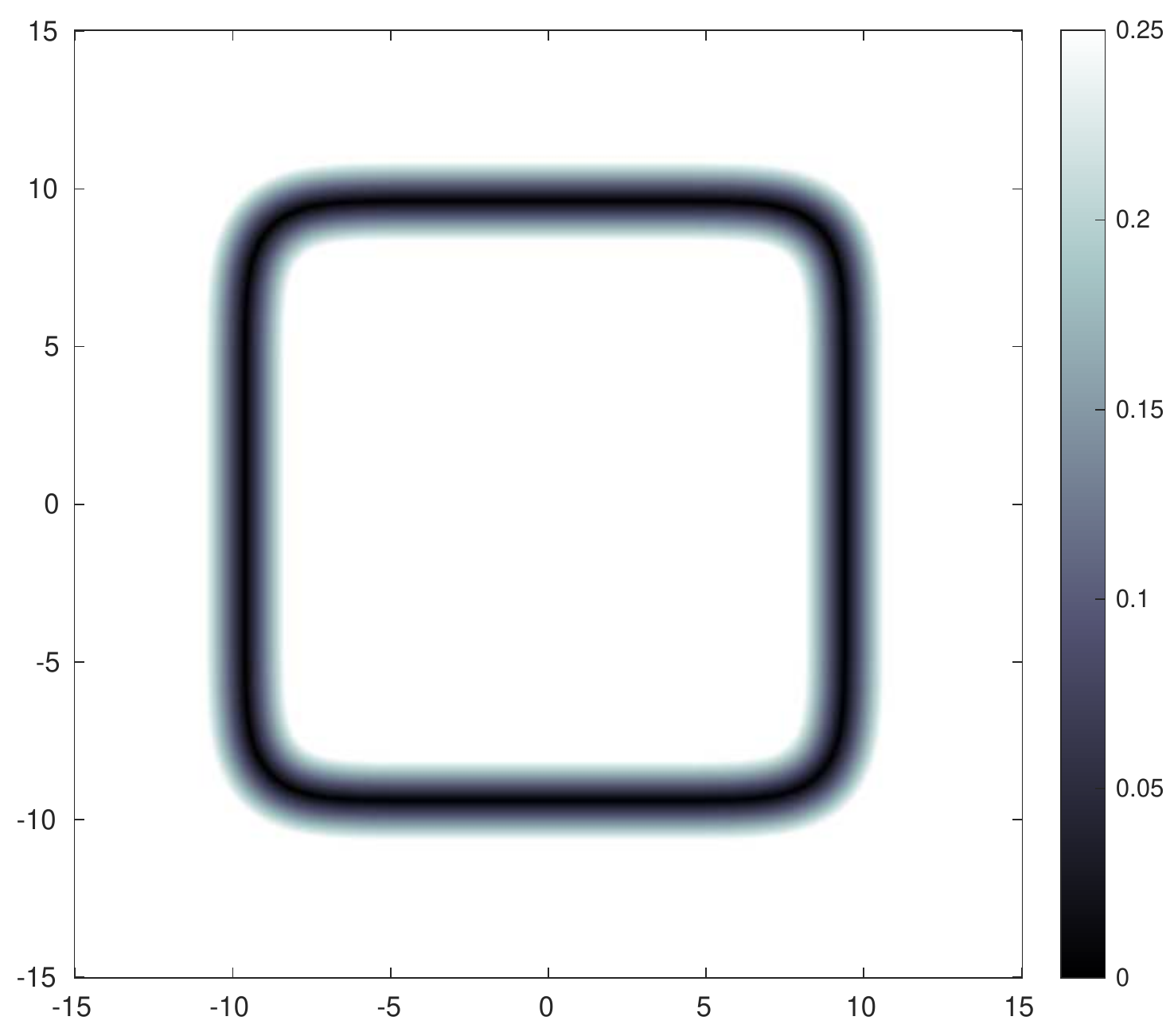}\quad{}\includegraphics[viewport=65bp 25bp 490bp 400bp,clip,scale=0.38]{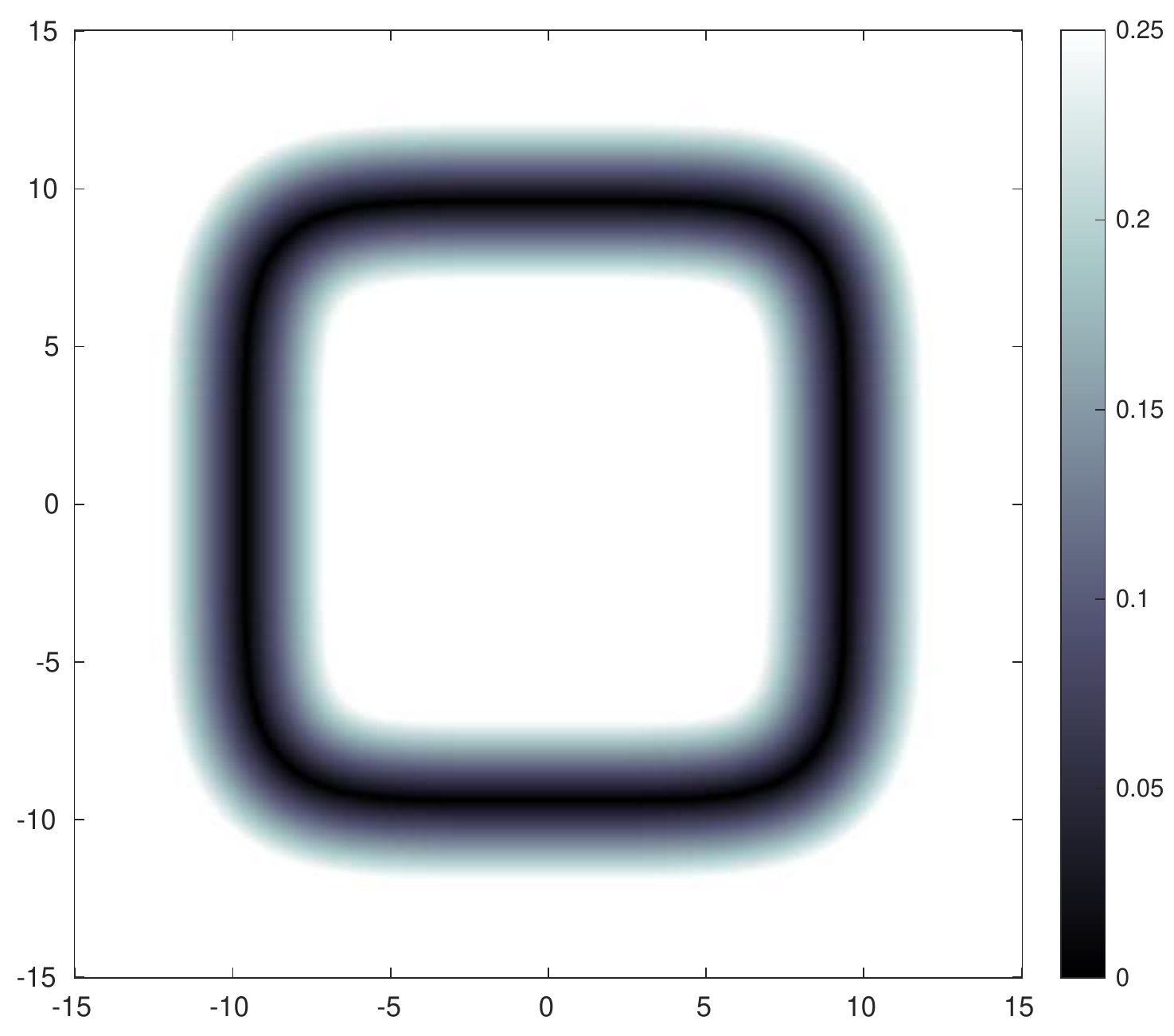}\\
{\footnotesize{}\makebox[10cm][l]{\raisebox{-0.08cm}[0.0cm][0.1cm]{$\kappa=0.05$, $E=0$\hspace*{4.4cm}$\kappa=0.0.02$, $E=0$}}}\\
\includegraphics[viewport=65bp 25bp 490bp 400bp,clip,scale=0.38]{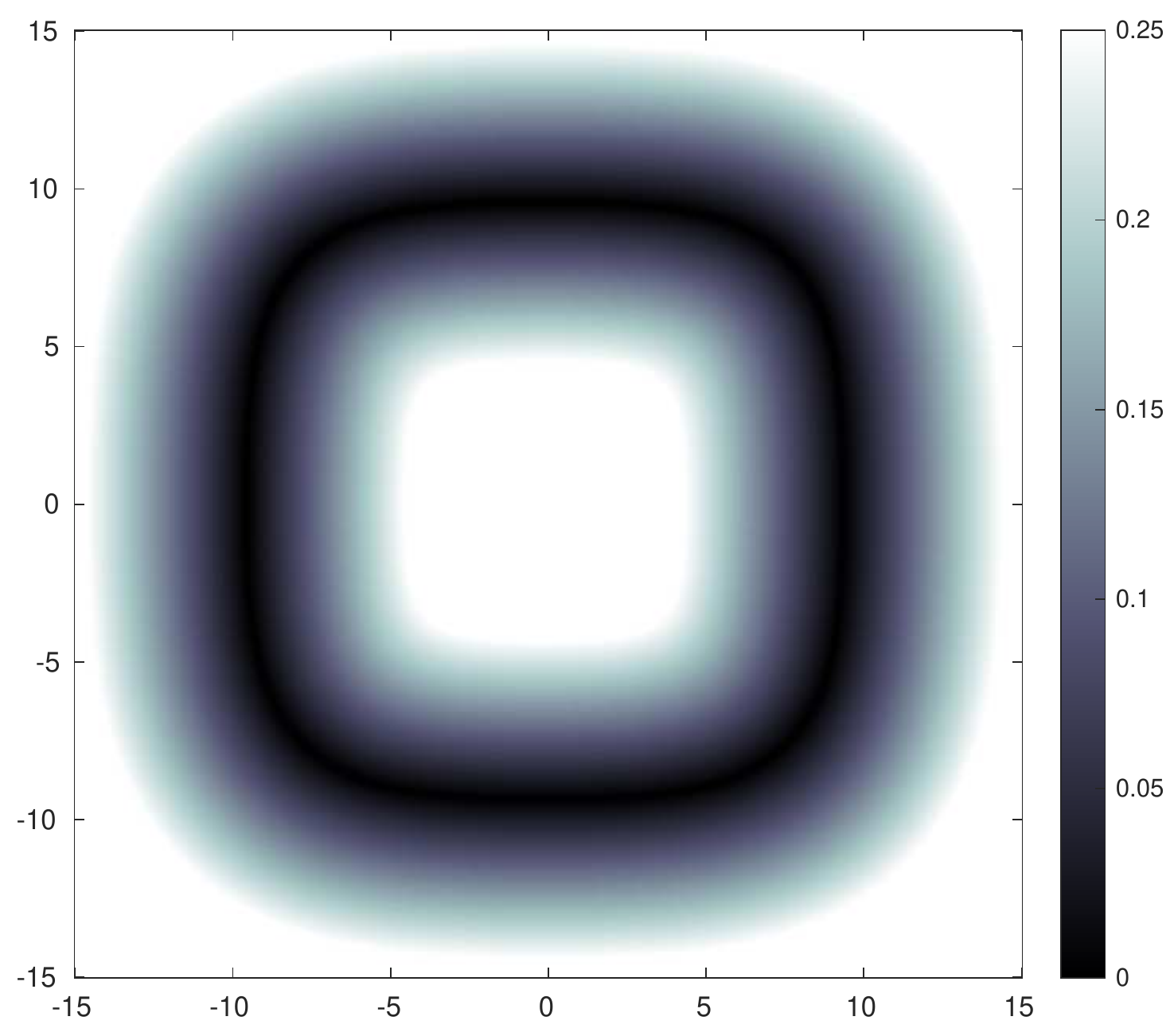}\quad{}\includegraphics[viewport=65bp 25bp 490bp 400bp,clip,scale=0.38]{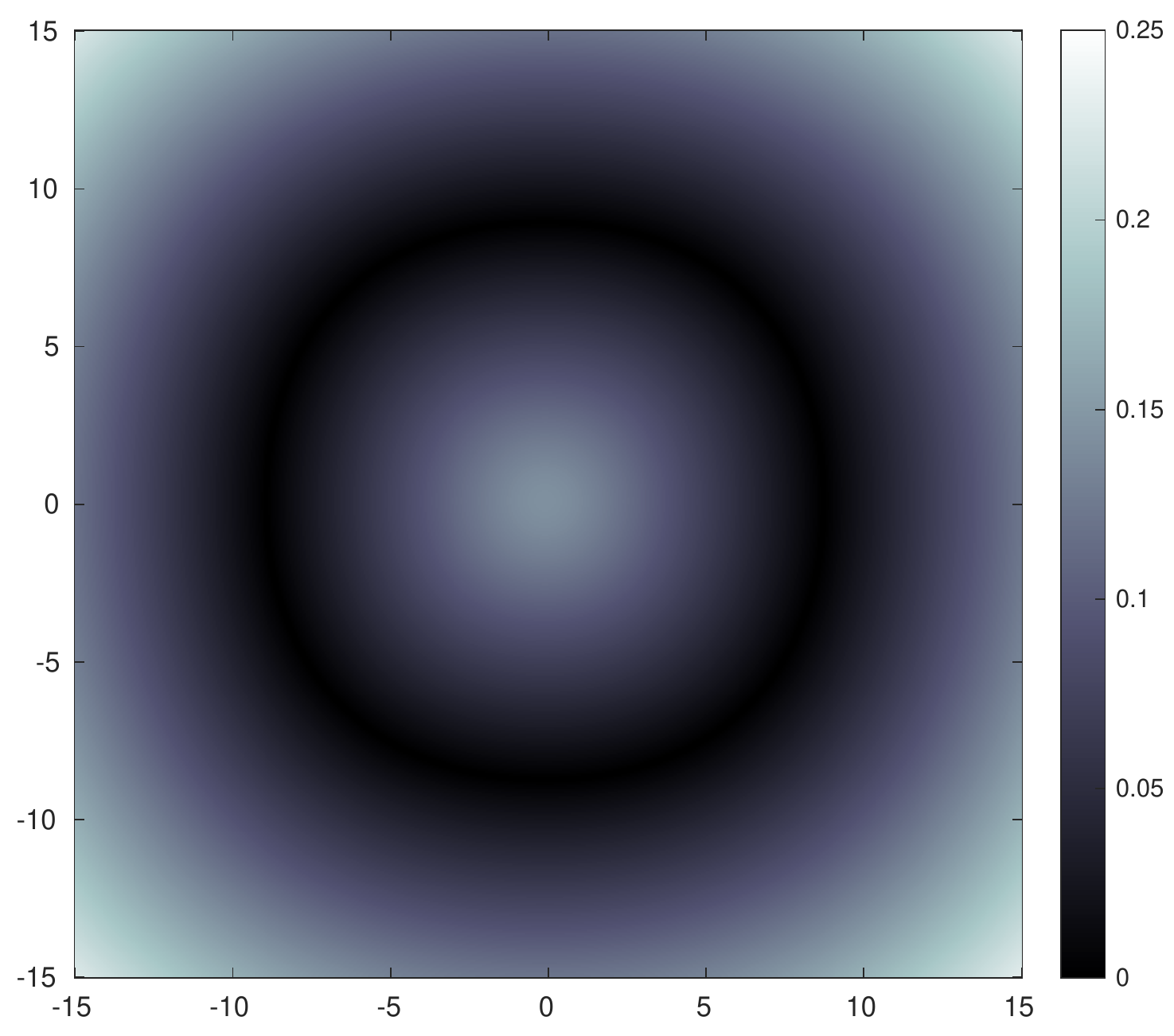}\caption{A slice of the pseudospectrum at $\lambda_{3}=E_{F}=0$. No disorder,
various $\kappa$. \label{fig:slice_zero_clean}}
\end{figure}
\begin{figure}
\noindent {\footnotesize{}\makebox[10cm][l]{\raisebox{-0.08cm}[0.0cm][0.1cm]{$\kappa=5$, $E=0$\hspace*{4.4cm}$\kappa=2$, $E=0$}}}\\
\includegraphics[viewport=65bp 25bp 490bp 400bp,clip,scale=0.38]{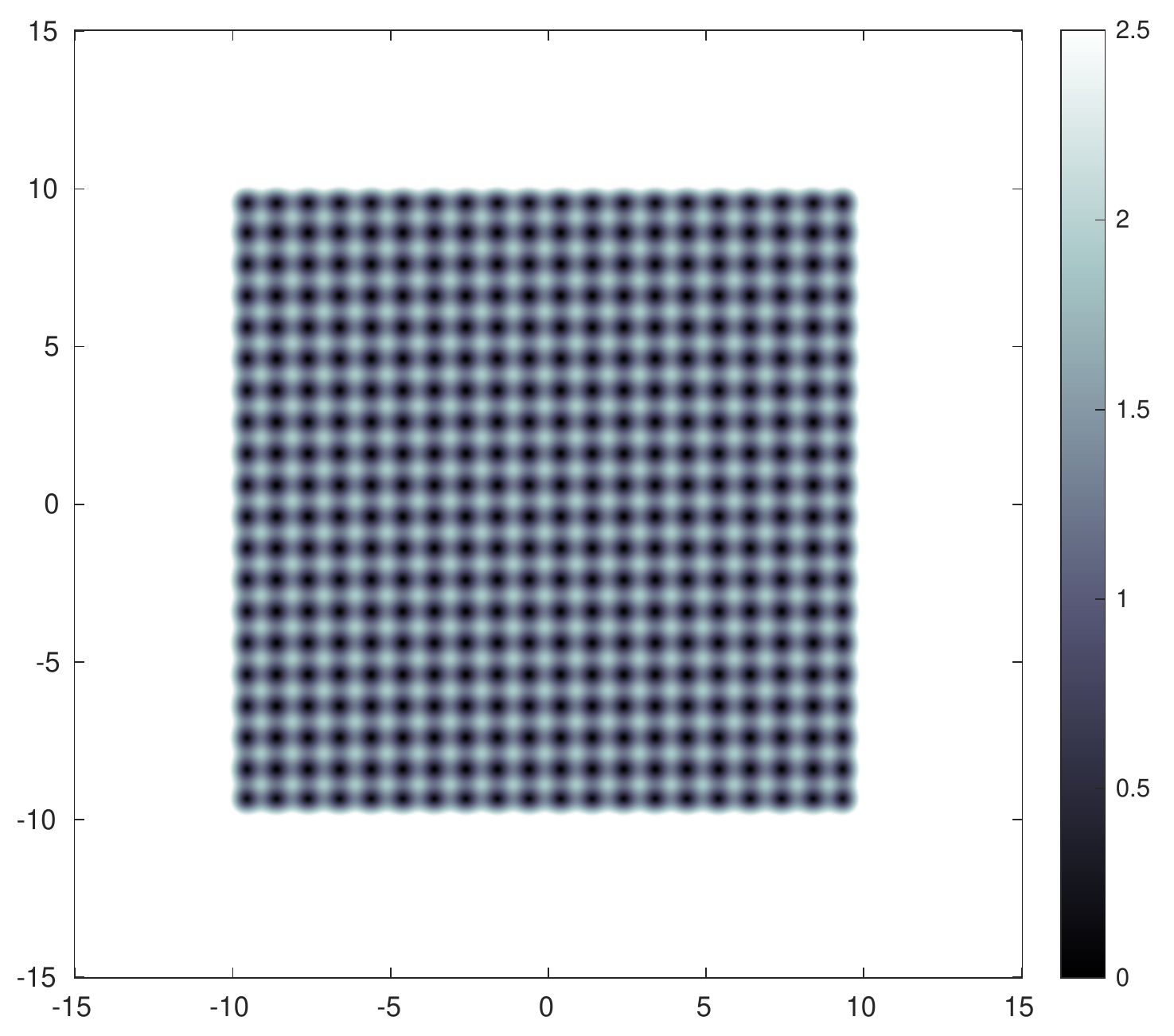}\quad{}\includegraphics[viewport=65bp 25bp 490bp 400bp,clip,scale=0.38]{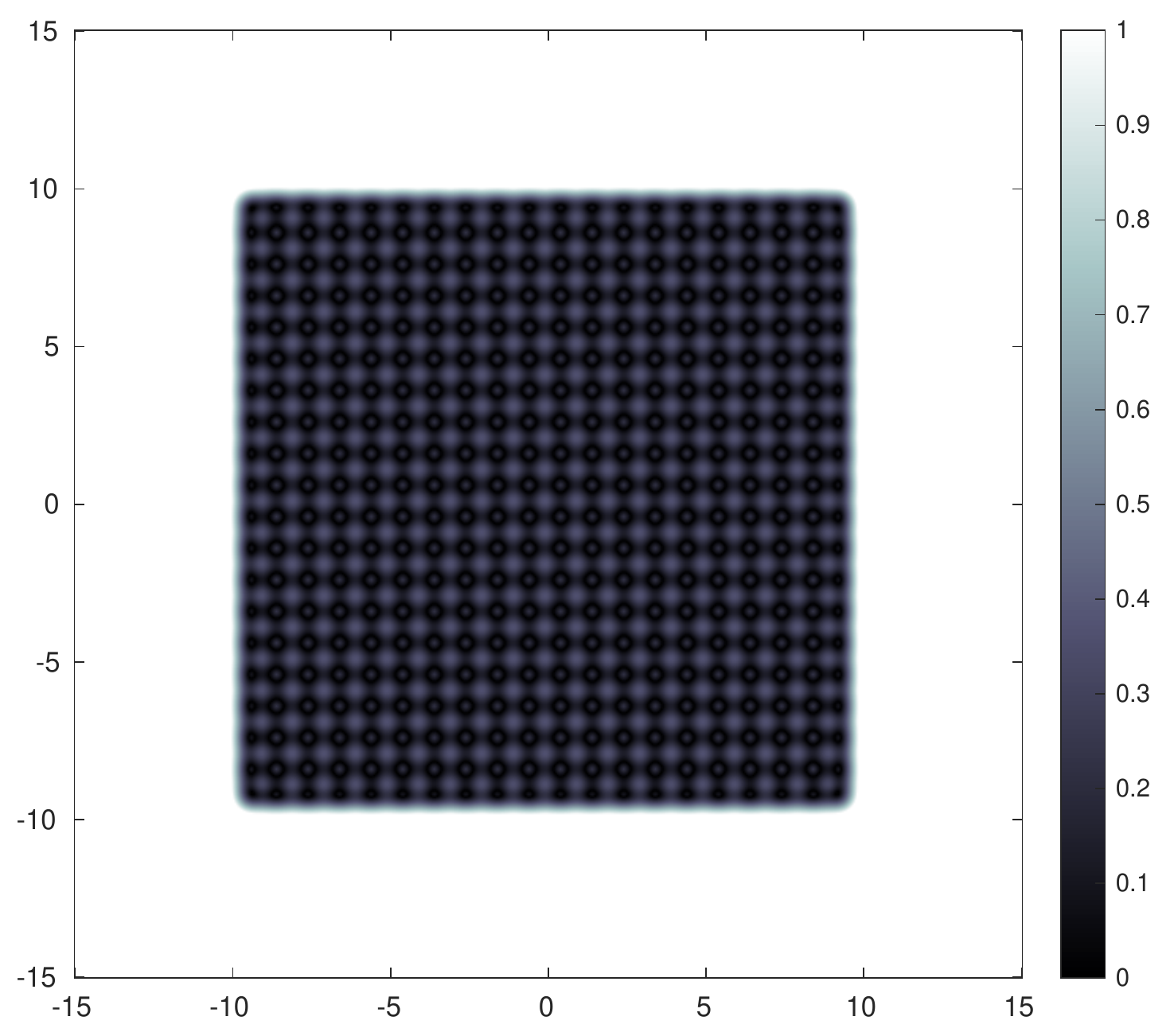}\vspace*{-0.3cm}
\\
{\footnotesize{}\makebox[10cm][l]{\raisebox{-0.08cm}[0.0cm][0.1cm]{$\kappa=1$, $E=0$\hspace*{4.4cm}$\kappa=0.5$, $E=0$}}}\\
\includegraphics[viewport=65bp 25bp 490bp 400bp,clip,scale=0.38]{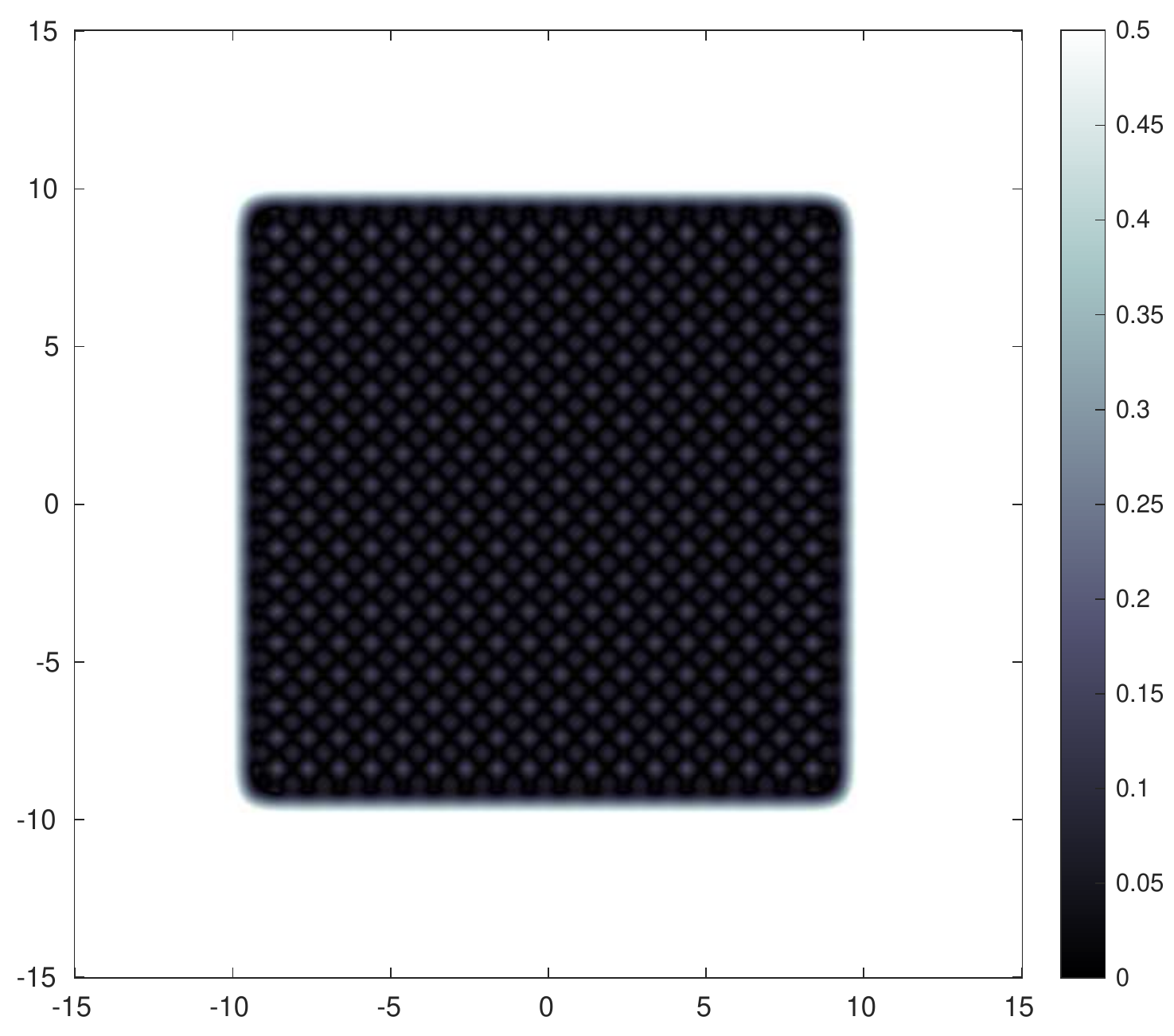}\quad{}\includegraphics[viewport=65bp 25bp 490bp 400bp,clip,scale=0.38]{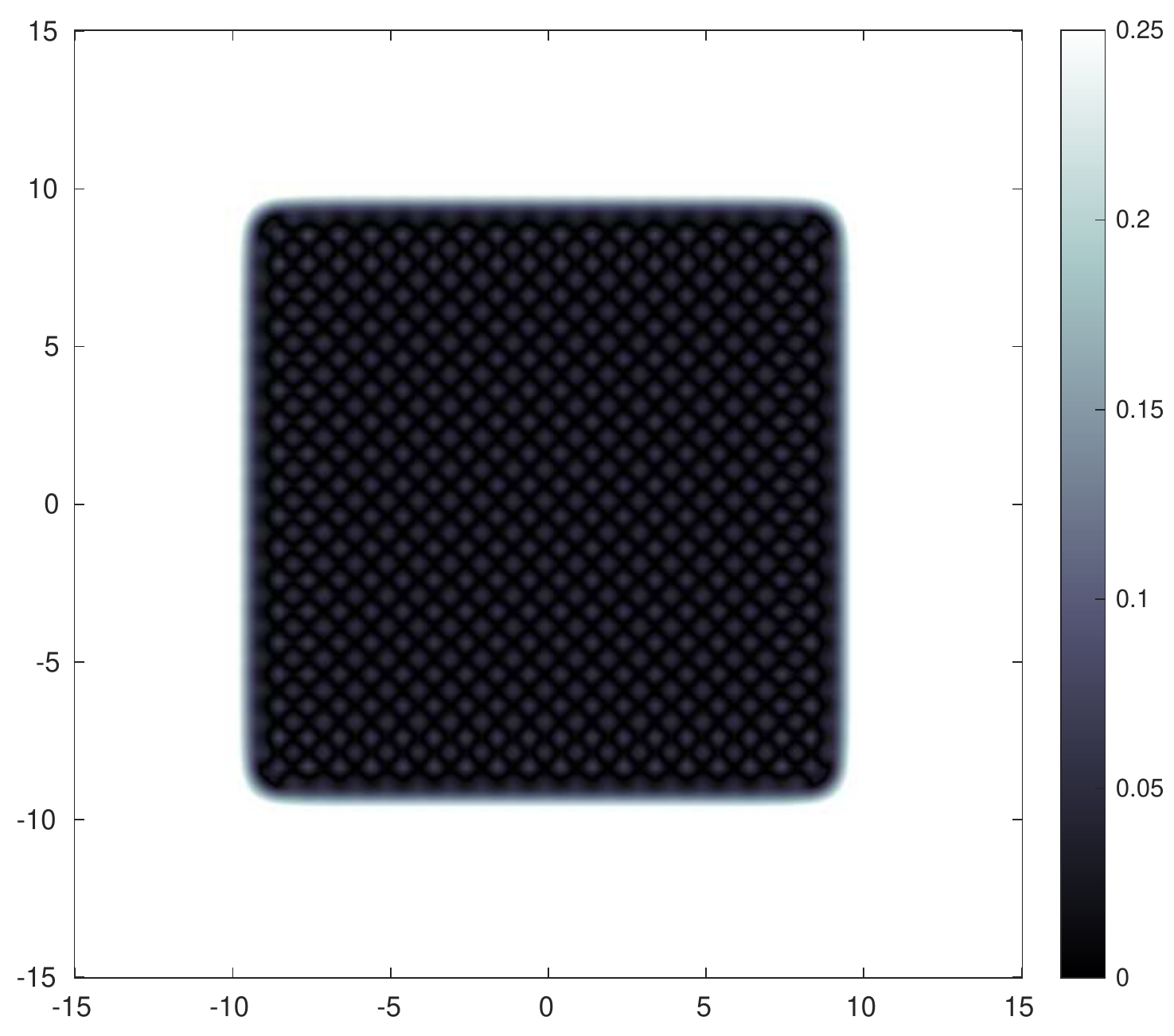}\vspace*{-0.3cm}
\\
{\footnotesize{}\makebox[10cm][l]{\raisebox{-0.08cm}[0.0cm][0.1cm]{$\kappa=0.2$, $E=0$\hspace*{4.4cm}$\kappa=0.1$, $E=0$}}}\\
\includegraphics[viewport=65bp 25bp 490bp 400bp,clip,scale=0.38]{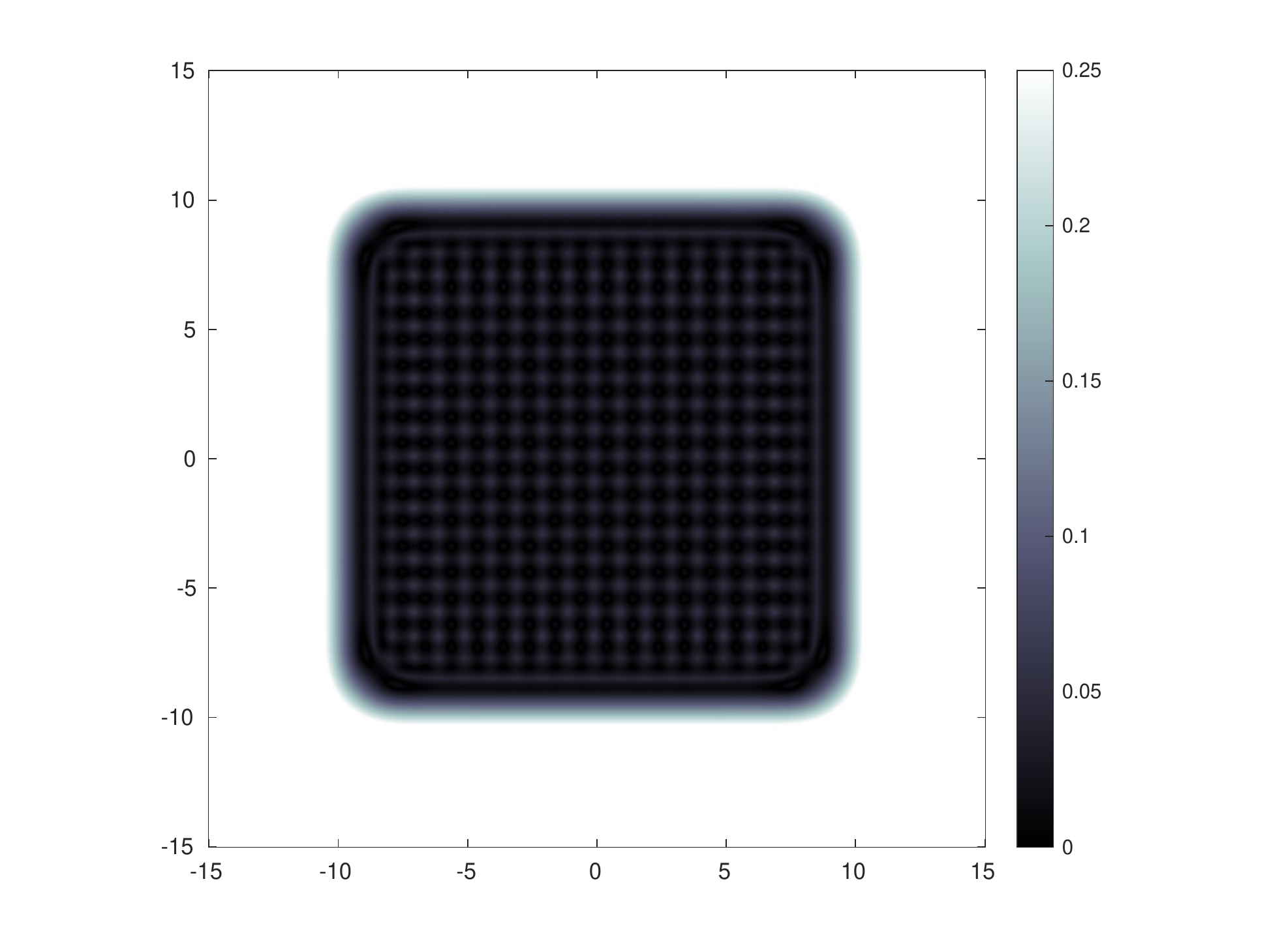}\quad{}\includegraphics[viewport=65bp 25bp 490bp 400bp,clip,scale=0.38]{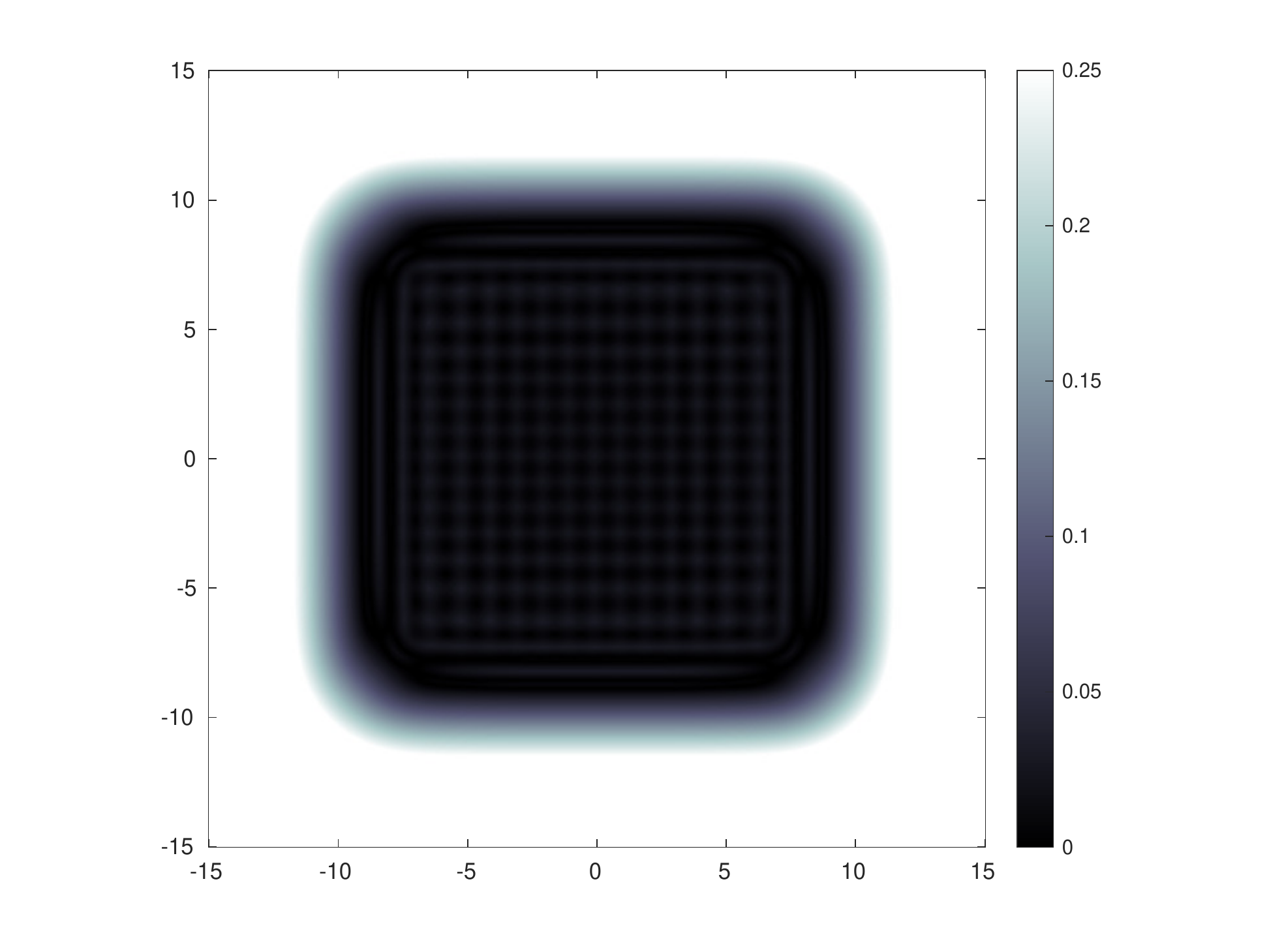}\vspace*{-0.3cm}
\\
{\footnotesize{}\makebox[10cm][l]{\raisebox{-0.08cm}[0.0cm][0.1cm]{$\kappa=0.05$, $E=0$\hspace*{4.4cm}$\kappa=0.0.02$, $E=0$}}}\\
\includegraphics[viewport=65bp 25bp 490bp 400bp,clip,scale=0.38]{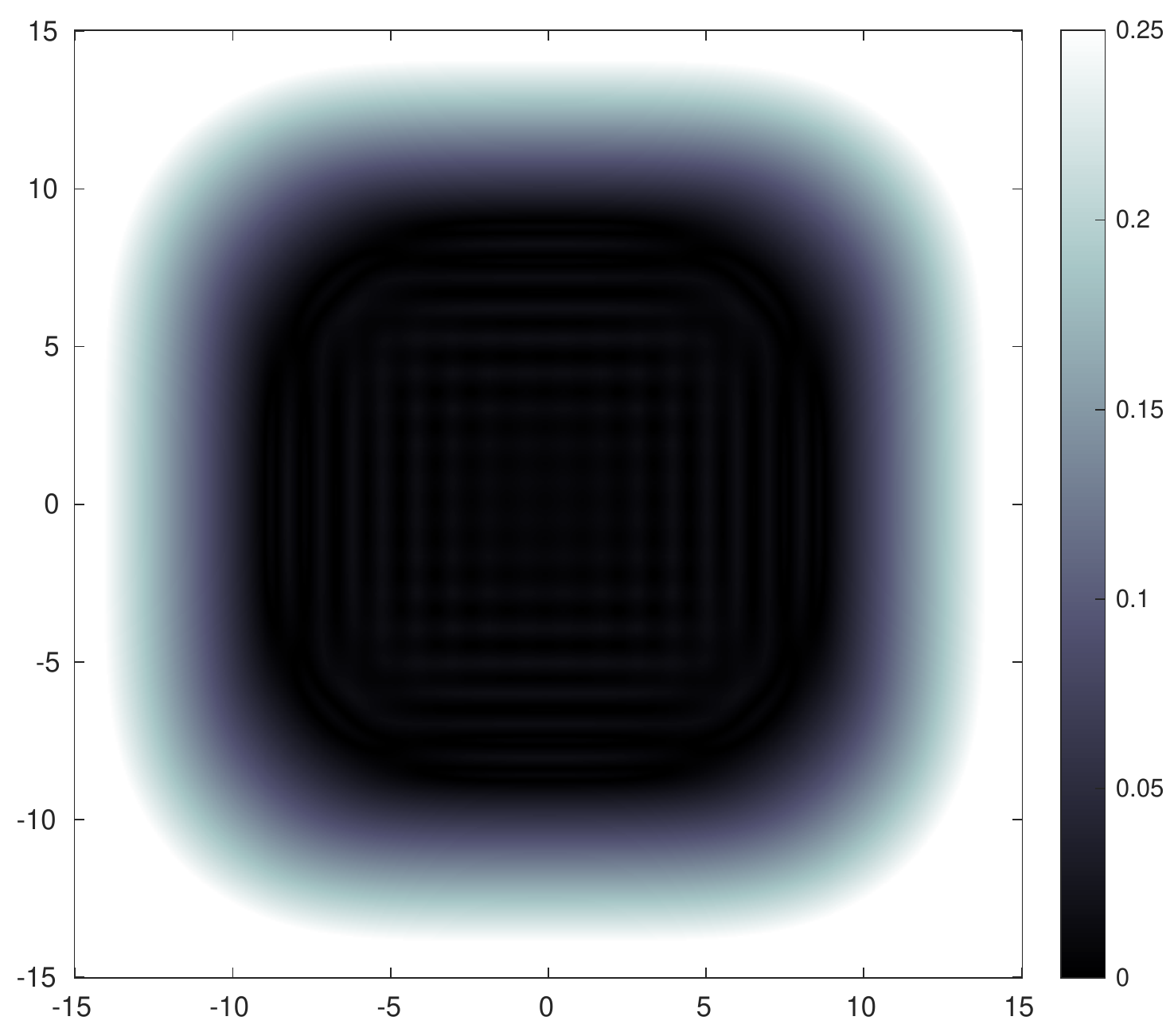}\quad{}\includegraphics[viewport=65bp 25bp 490bp 400bp,clip,scale=0.38]{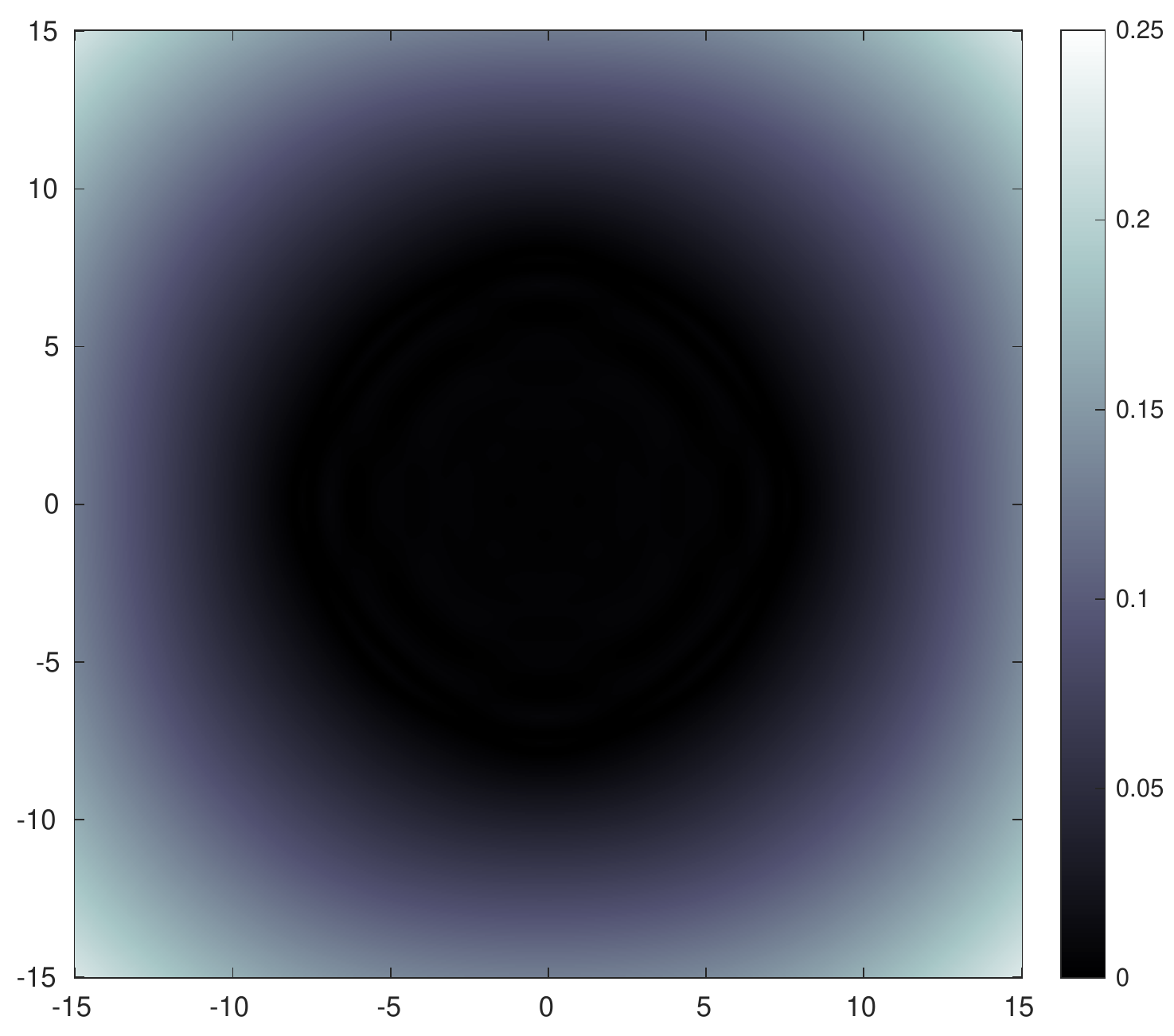}\caption{A slice of the pseudospectrum at $\lambda_{3}=E_{F}=-2.4$. No disorder,
various $\kappa$. \label{fig:slice_2.4_clean}}
\end{figure}

\begin{figure}
\noindent {\footnotesize{}\makebox[10cm][l]{\raisebox{-0.08cm}[0.0cm][0.1cm]{$\kappa=5$, $E=0$\hspace*{4.4cm}$\kappa=2$, $E=0$}}}\\
\includegraphics[viewport=65bp 25bp 490bp 400bp,clip,scale=0.38]{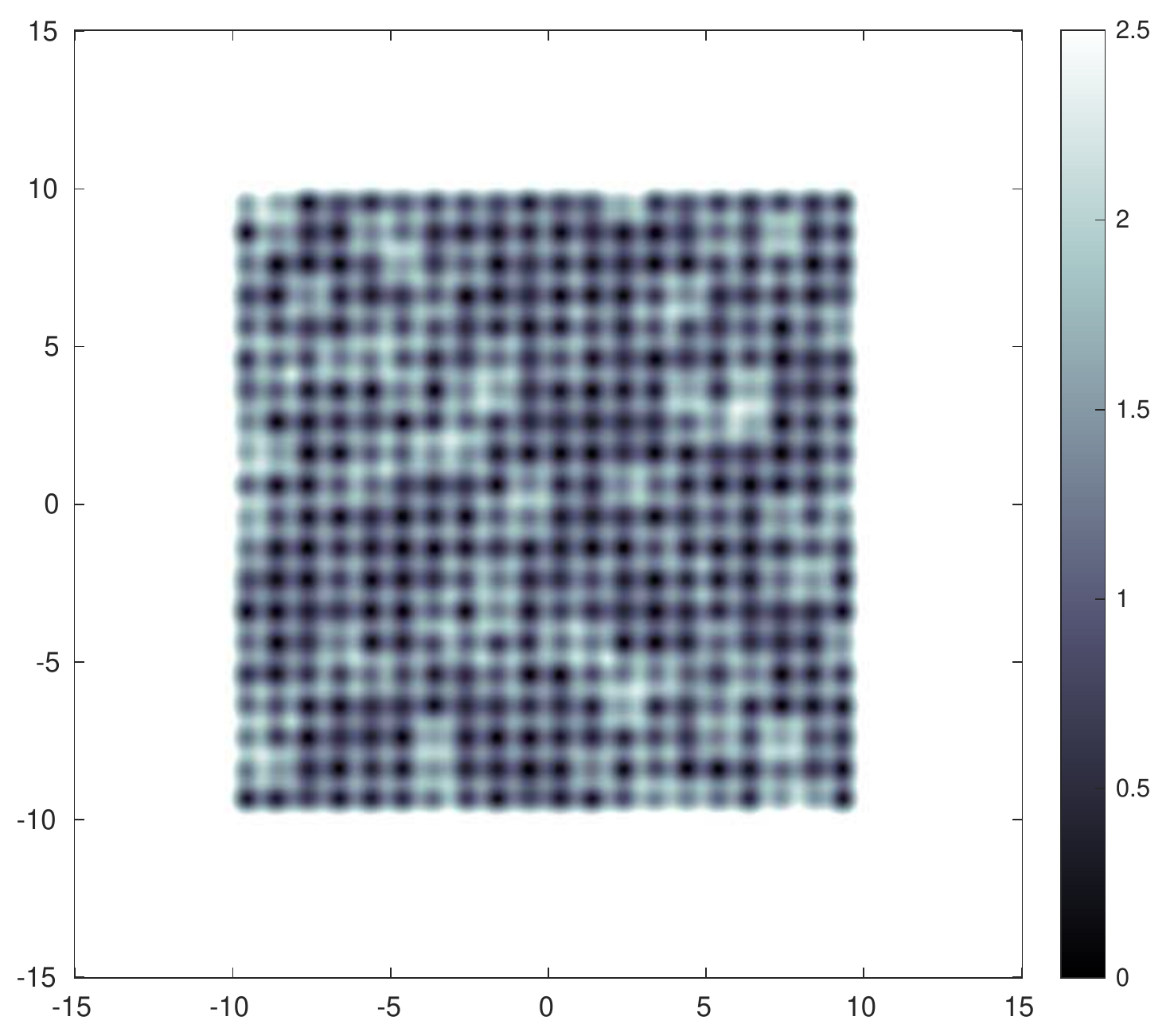}\quad{}\includegraphics[viewport=65bp 25bp 490bp 400bp,clip,scale=0.38]{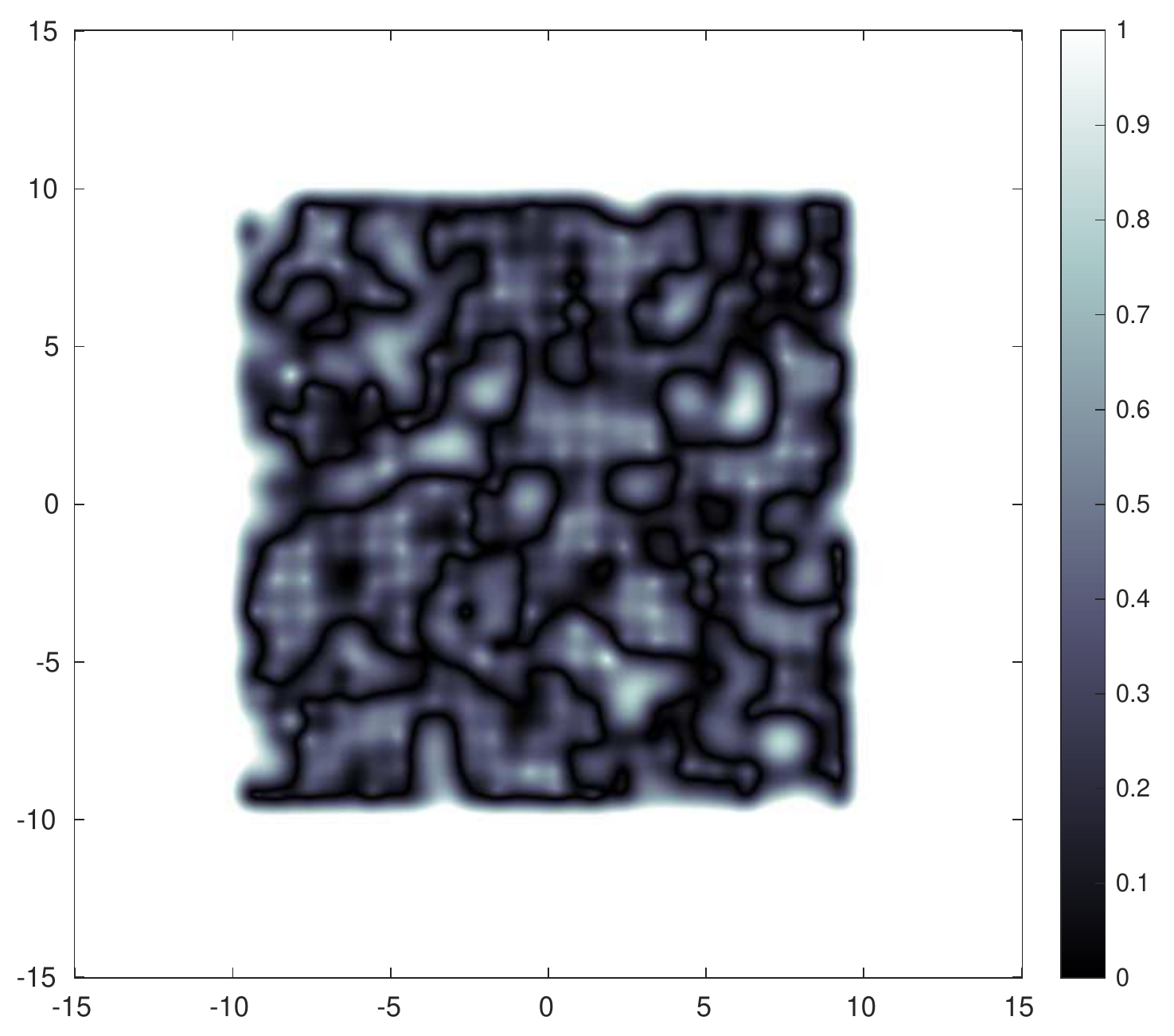}\vspace*{-0.3cm}
\\
{\footnotesize{}\makebox[10cm][l]{\raisebox{-0.08cm}[0.0cm][0.1cm]{$\kappa=1$, $E=0$\hspace*{4.4cm}$\kappa=0.5$, $E=0$}}}\\
\includegraphics[viewport=65bp 25bp 490bp 400bp,clip,scale=0.38]{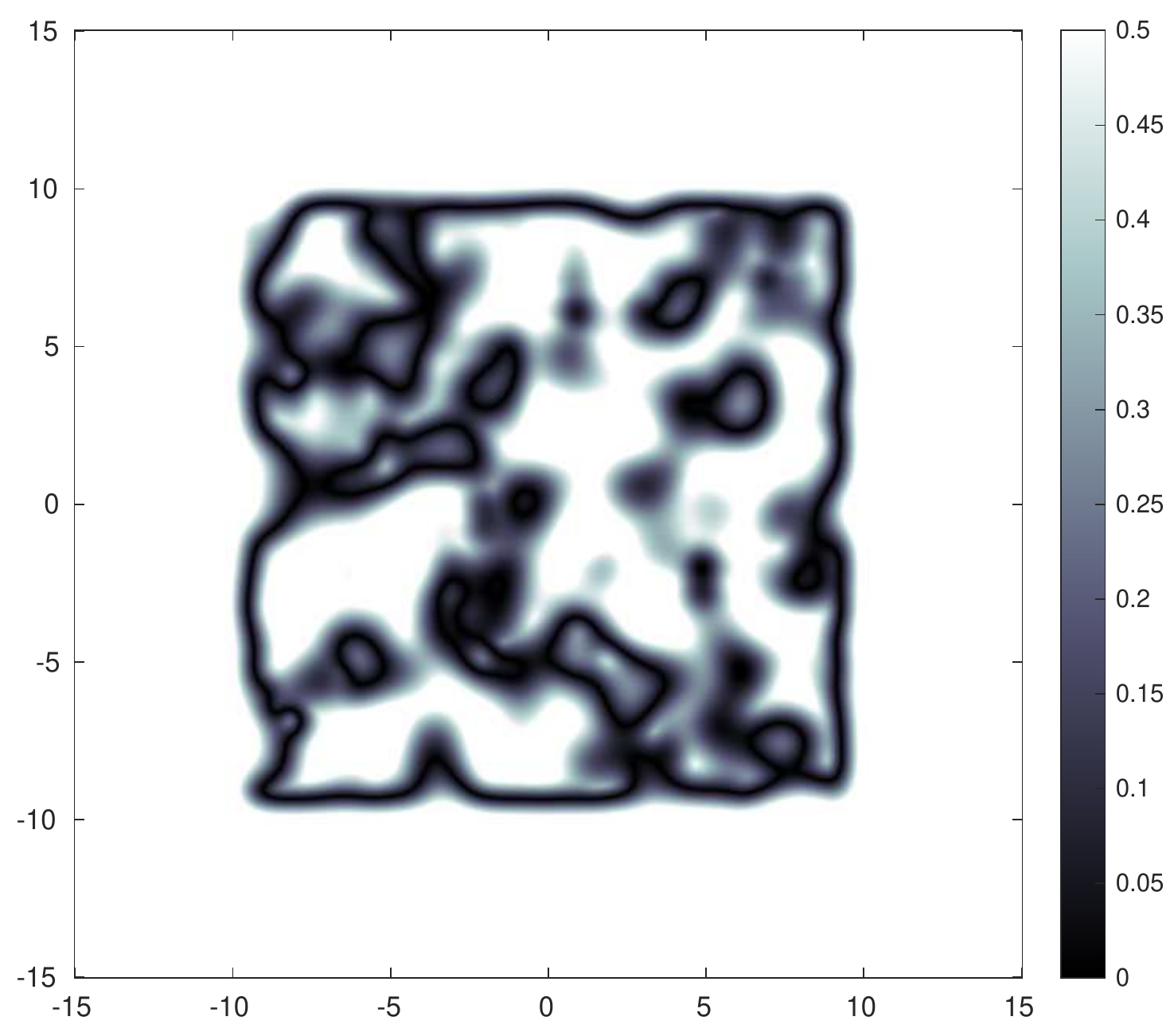}\quad{}\includegraphics[viewport=65bp 25bp 490bp 400bp,clip,scale=0.38]{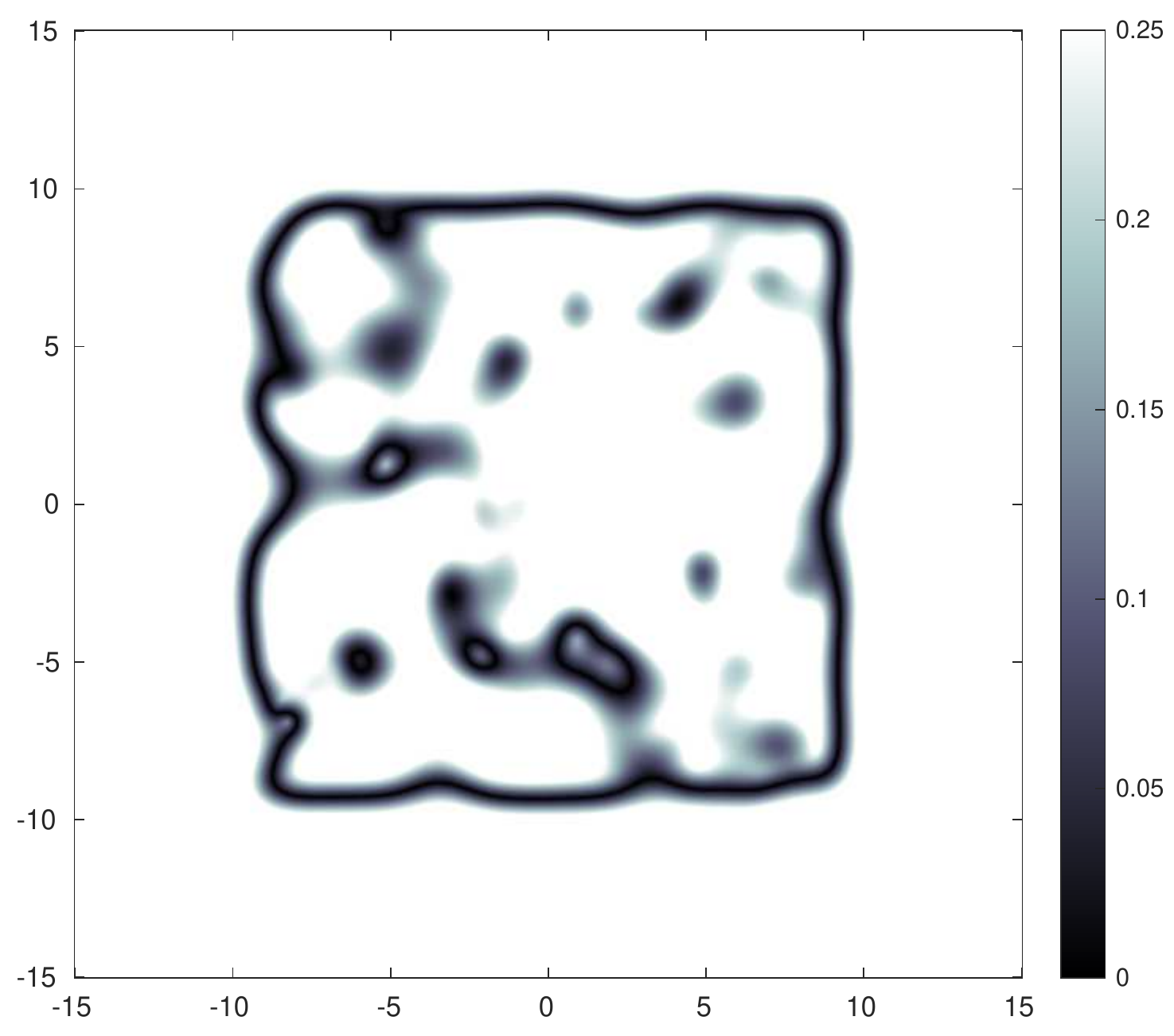}\vspace*{-0.3cm}
\\
{\footnotesize{}\makebox[10cm][l]{\raisebox{-0.08cm}[0.0cm][0.1cm]{$\kappa=0.2$, $E=0$\hspace*{4.4cm}$\kappa=0.1$, $E=0$}}}\\
\includegraphics[viewport=65bp 25bp 490bp 400bp,clip,scale=0.38]{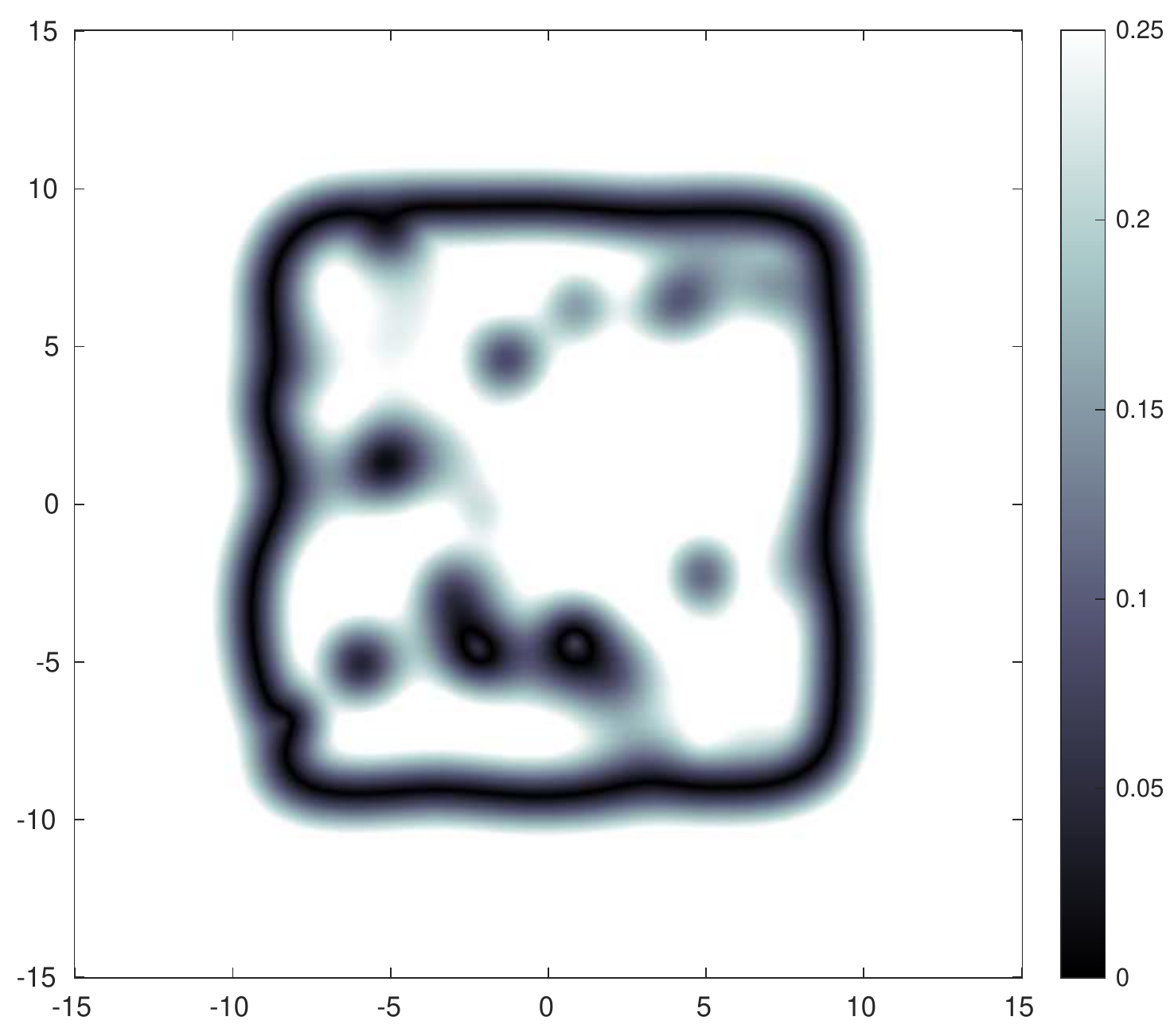}\quad{}\includegraphics[viewport=65bp 25bp 490bp 400bp,clip,scale=0.38]{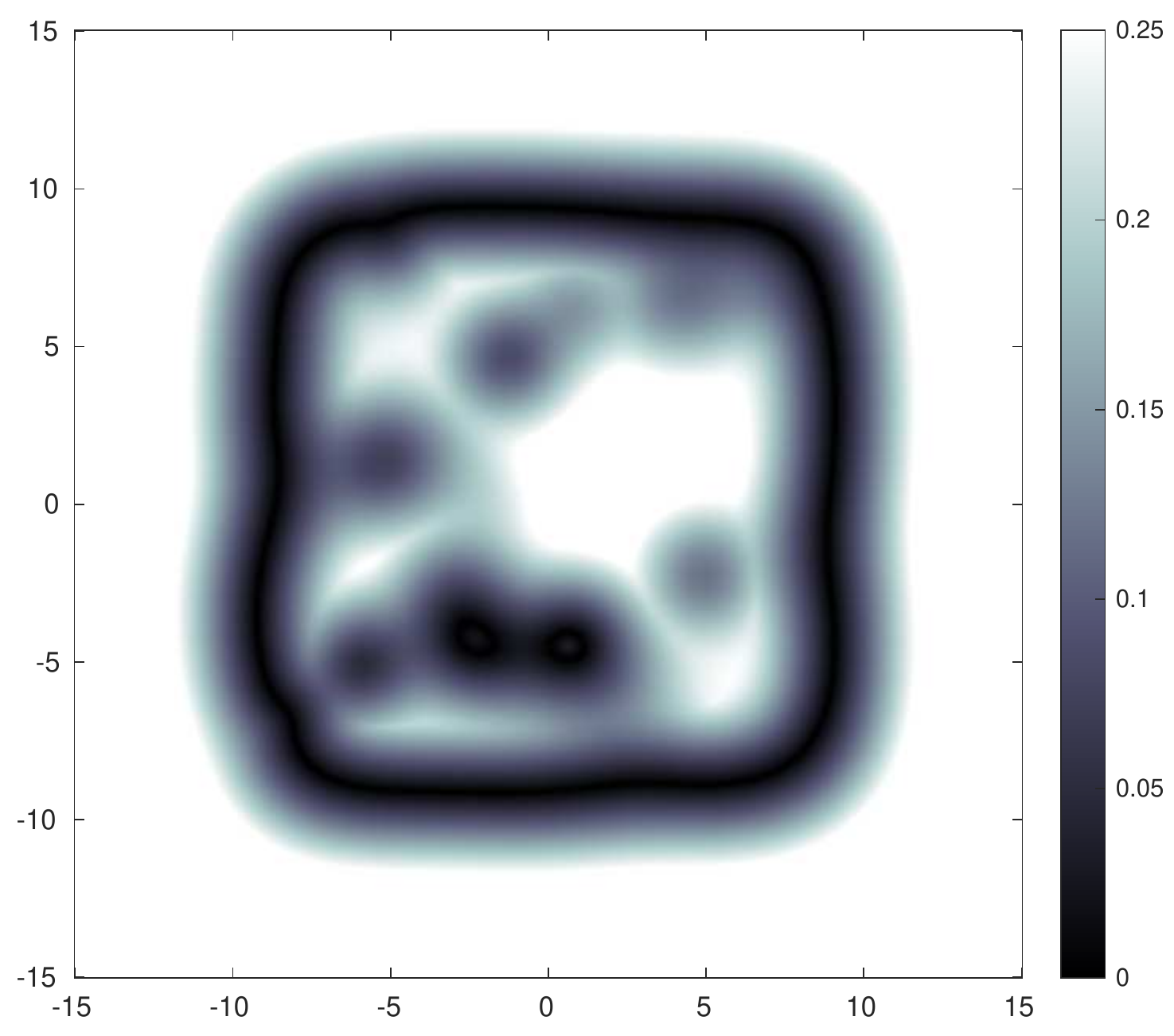}\vspace*{-0.3cm}
\\
{\footnotesize{}\makebox[10cm][l]{\raisebox{-0.08cm}[0.0cm][0.1cm]{$\kappa=0.05$, $E=0$\hspace*{4.4cm}$\kappa=0.0.02$, $E=0$}}}\\
\includegraphics[viewport=65bp 25bp 490bp 400bp,clip,scale=0.38]{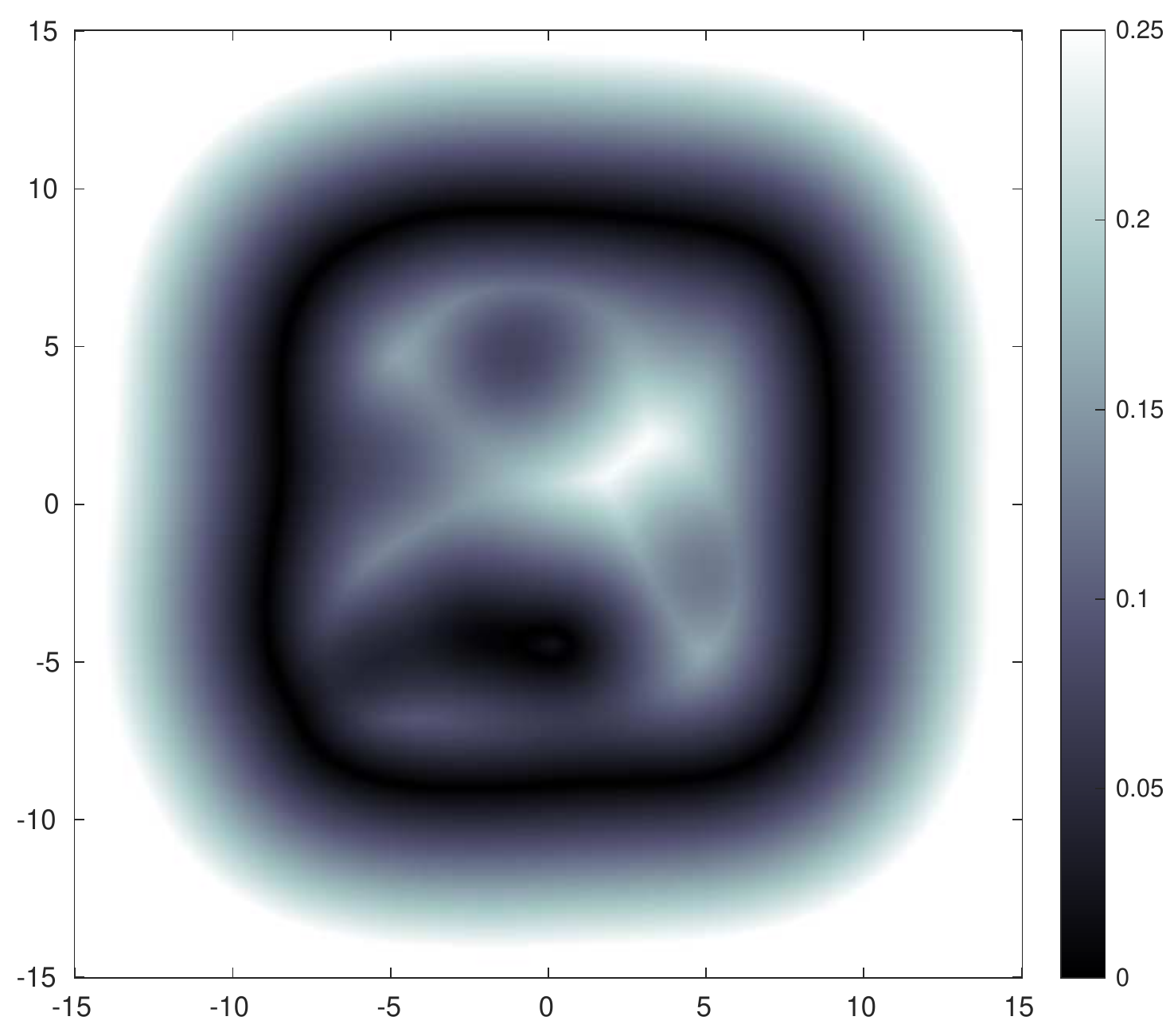}\quad{}\includegraphics[viewport=65bp 25bp 490bp 400bp,clip,scale=0.38]{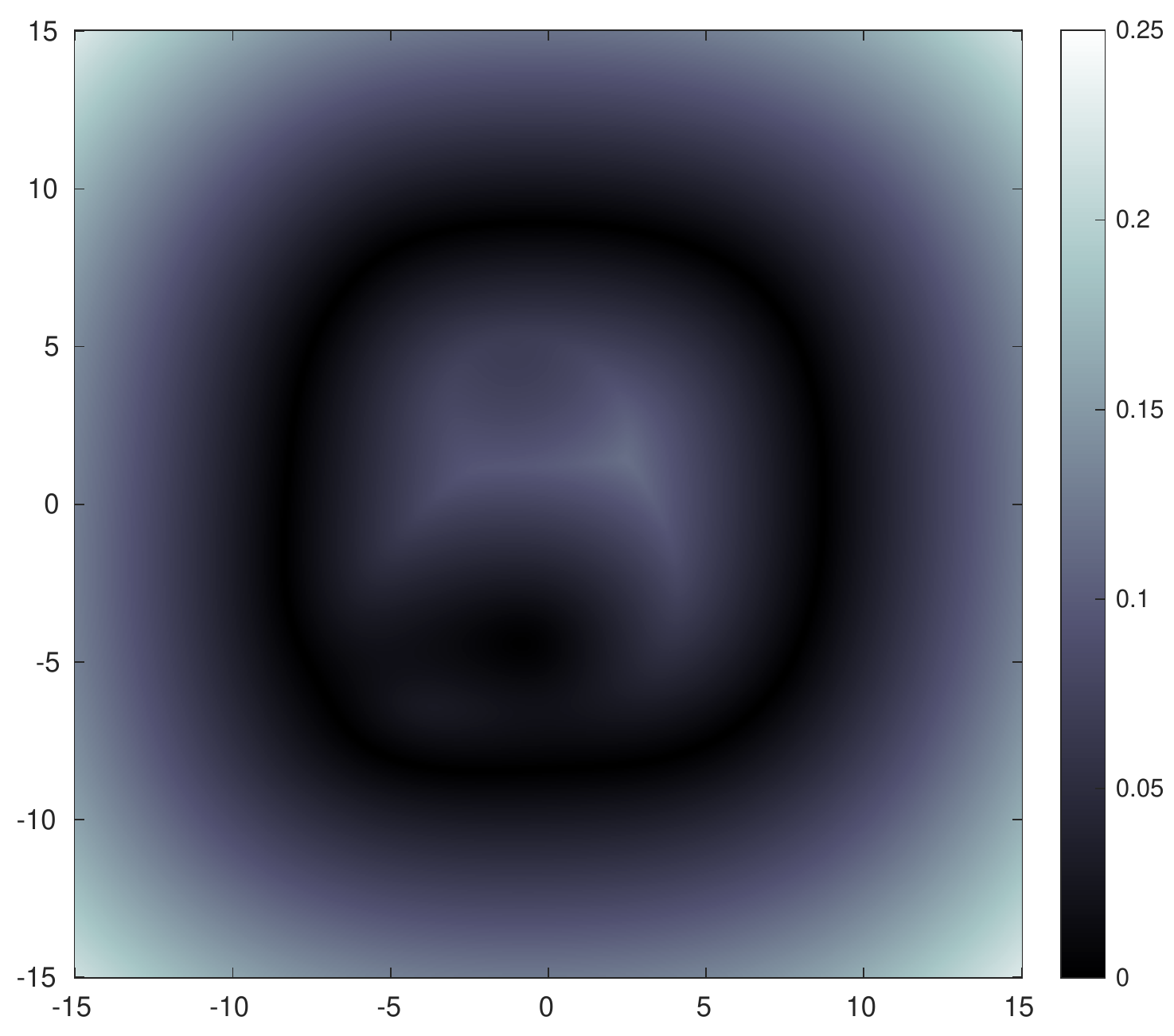}\caption{A slice of the pseudospectrum at $\lambda_{3}=E_{F}=0$. Strong disorder,
various $\kappa$. \label{fig:slice_zero_disorder}}
\end{figure}

\begin{figure}
\noindent {\footnotesize{}\makebox[10cm][l]{\raisebox{-0.08cm}[0.0cm][0.1cm]{$\kappa=5$, $E=0$\hspace*{4.4cm}$\kappa=2$, $E=0$}}}\\
\includegraphics[viewport=65bp 25bp 490bp 400bp,clip,scale=0.38]{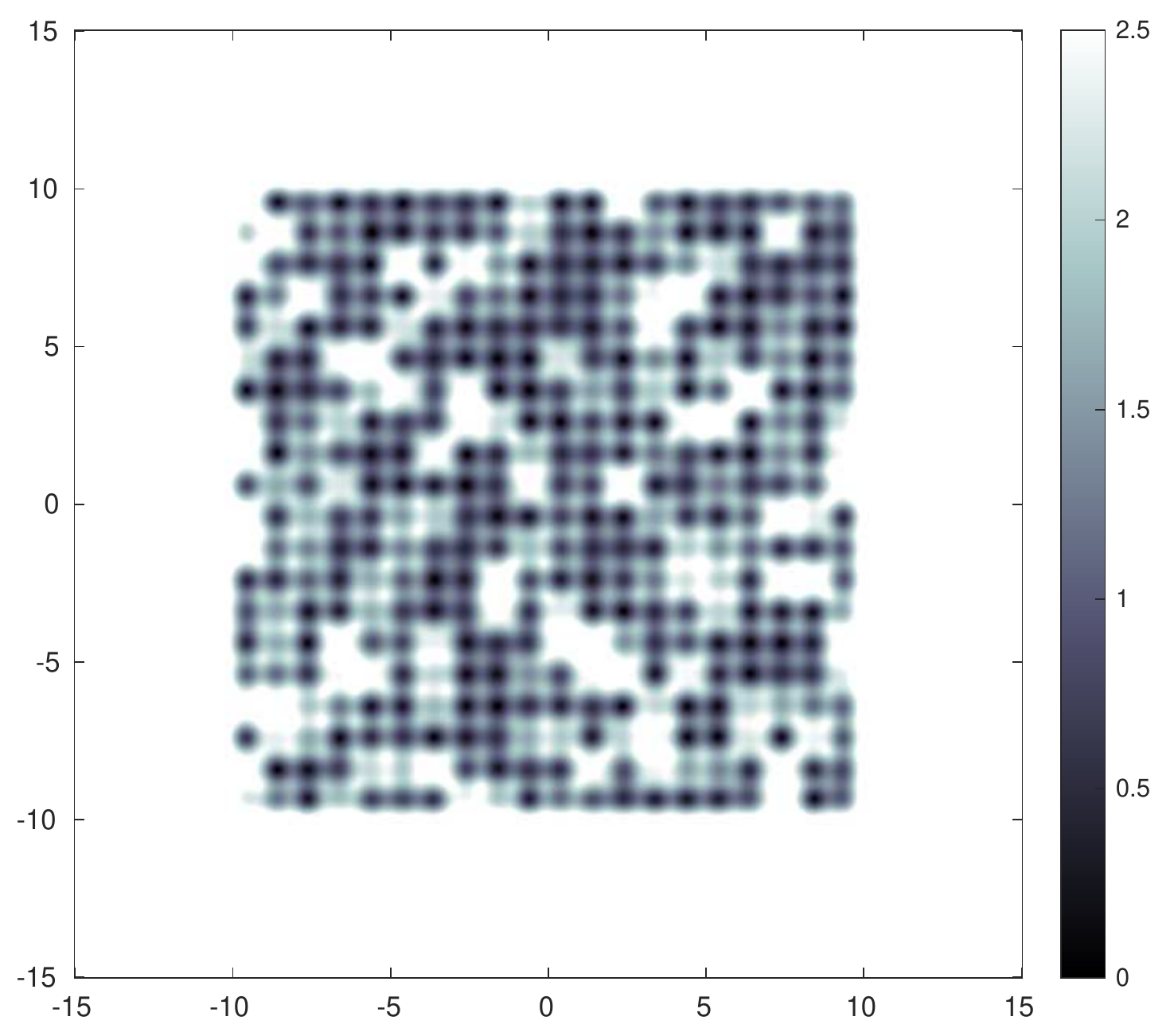}\quad{}\includegraphics[viewport=65bp 25bp 490bp 400bp,clip,scale=0.38]{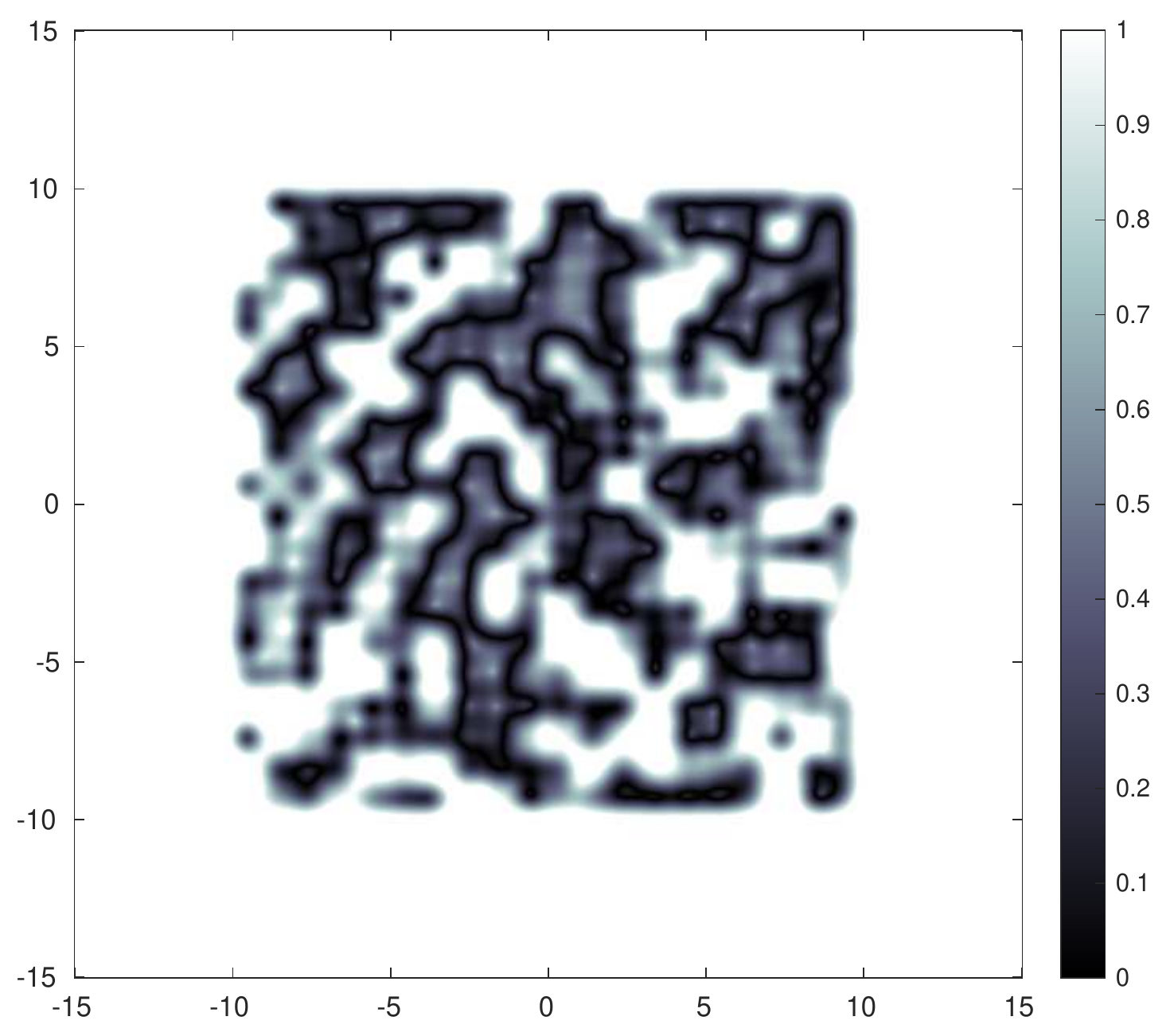}\vspace*{-0.3cm}
\\
{\footnotesize{}\makebox[10cm][l]{\raisebox{-0.08cm}[0.0cm][0.1cm]{$\kappa=1$, $E=0$\hspace*{4.4cm}$\kappa=0.5$, $E=0$}}}\\
\includegraphics[viewport=65bp 25bp 490bp 400bp,clip,scale=0.38]{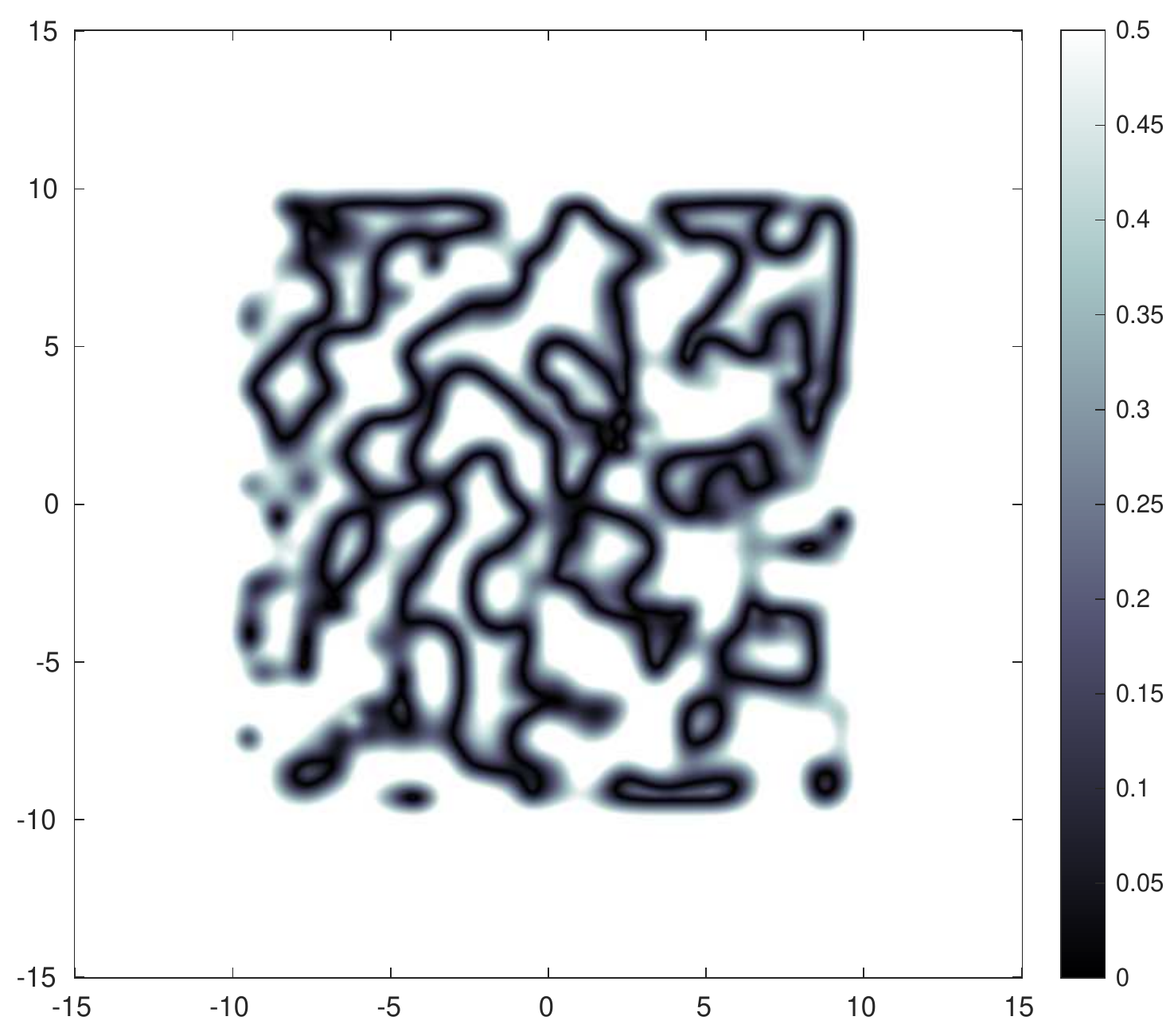}\quad{}\includegraphics[viewport=65bp 25bp 490bp 400bp,clip,scale=0.38]{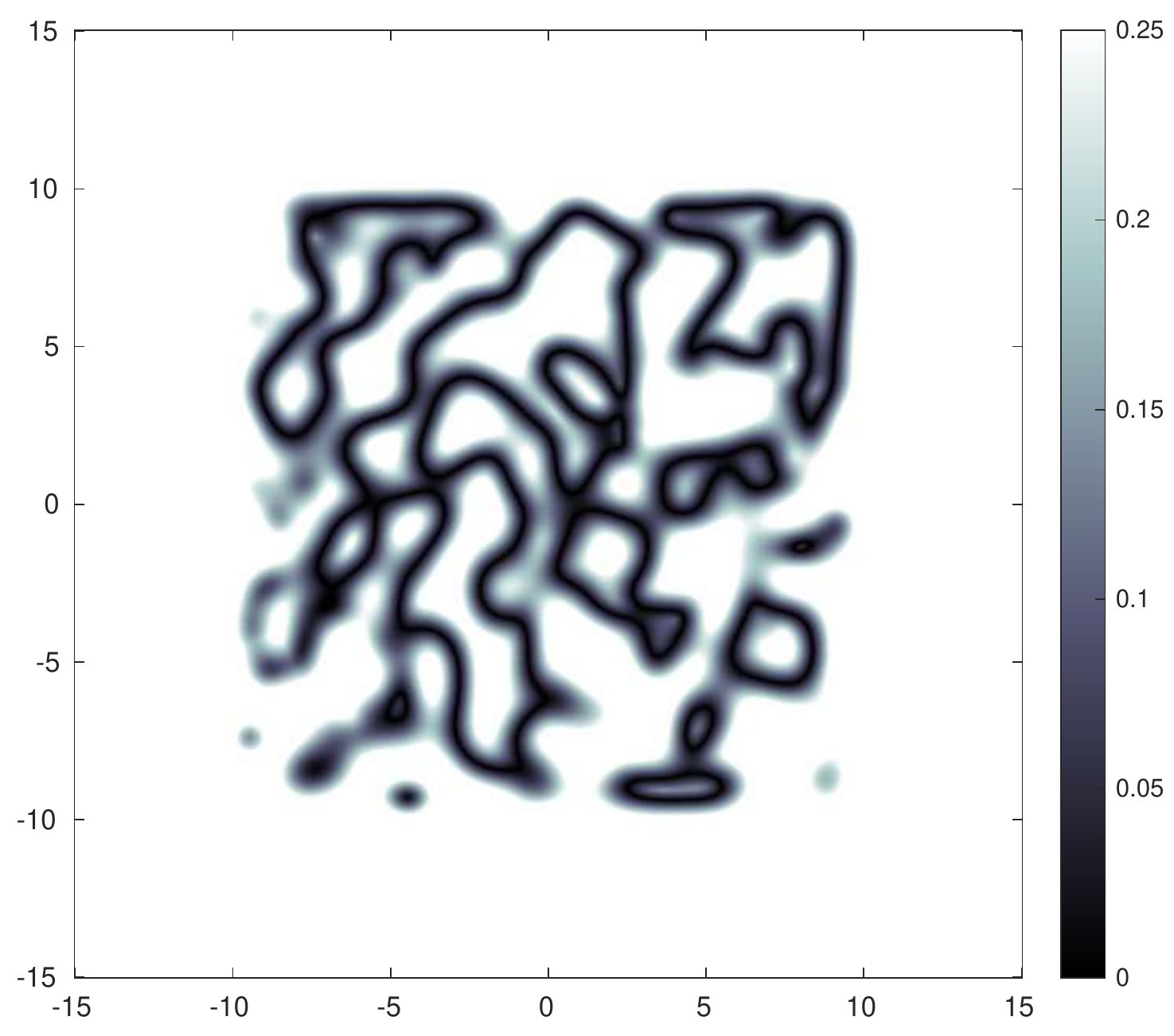}\vspace*{-0.3cm}
\\
{\footnotesize{}\makebox[10cm][l]{\raisebox{-0.08cm}[0.0cm][0.1cm]{$\kappa=0.2$, $E=0$\hspace*{4.4cm}$\kappa=0.1$, $E=0$}}}\\
\includegraphics[viewport=65bp 25bp 490bp 400bp,clip,scale=0.38]{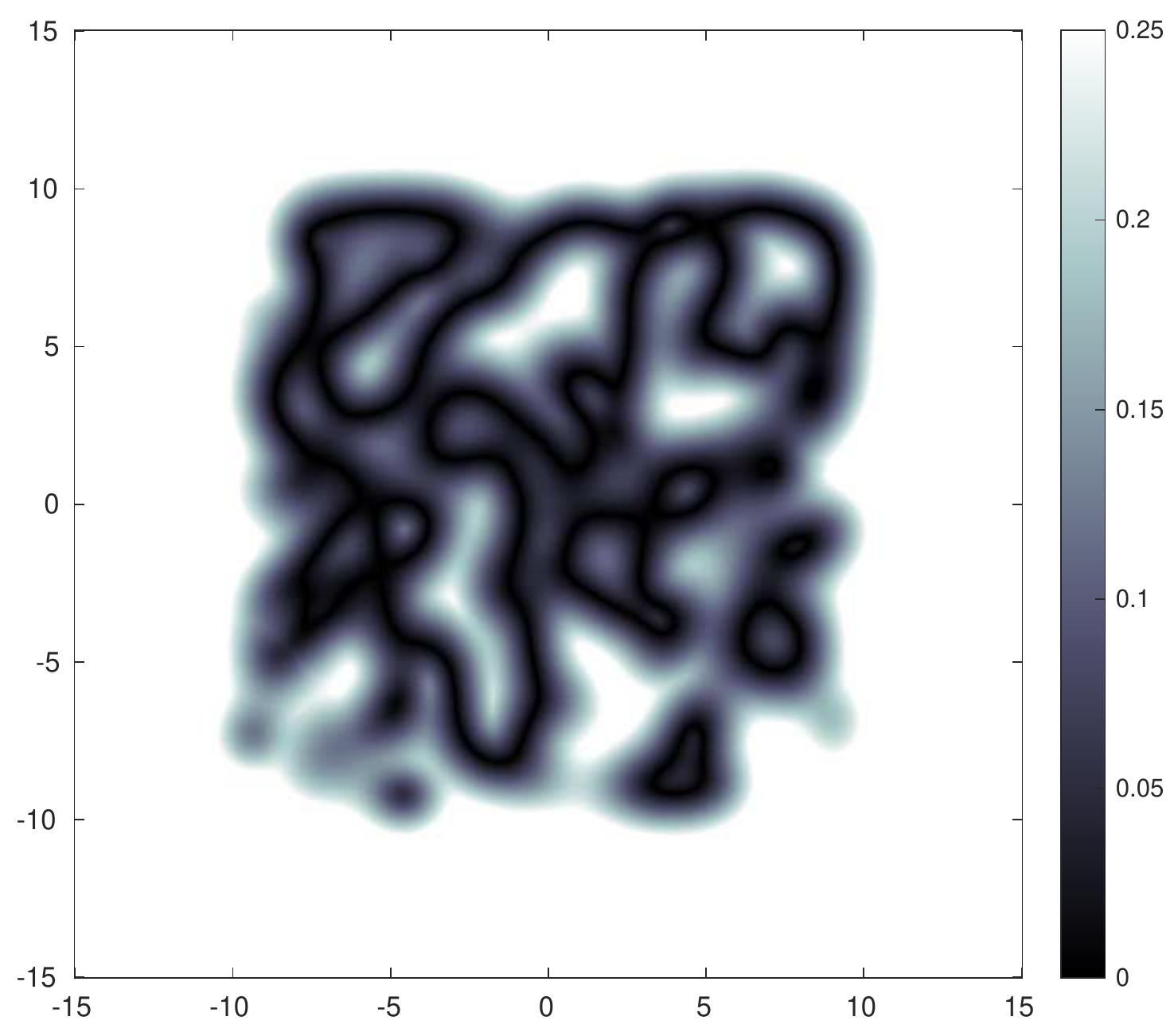}\quad{}\includegraphics[viewport=65bp 25bp 490bp 400bp,clip,scale=0.38]{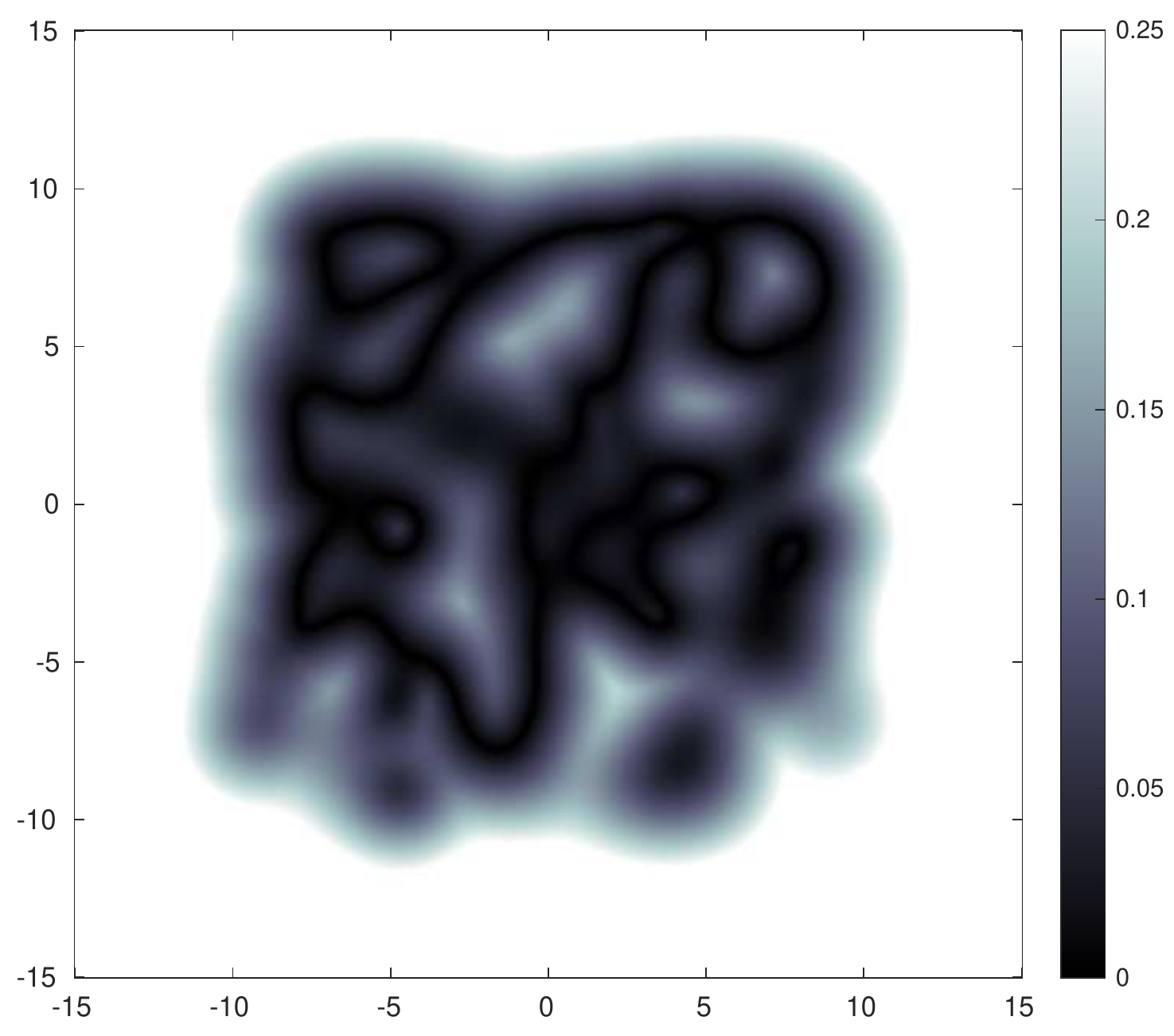}\vspace*{-0.3cm}
\\
{\footnotesize{}\makebox[10cm][l]{\raisebox{-0.08cm}[0.0cm][0.1cm]{$\kappa=0.05$, $E=0$\hspace*{4.4cm}$\kappa=0.0.02$, $E=0$}}}\\
\includegraphics[viewport=65bp 25bp 490bp 400bp,clip,scale=0.38]{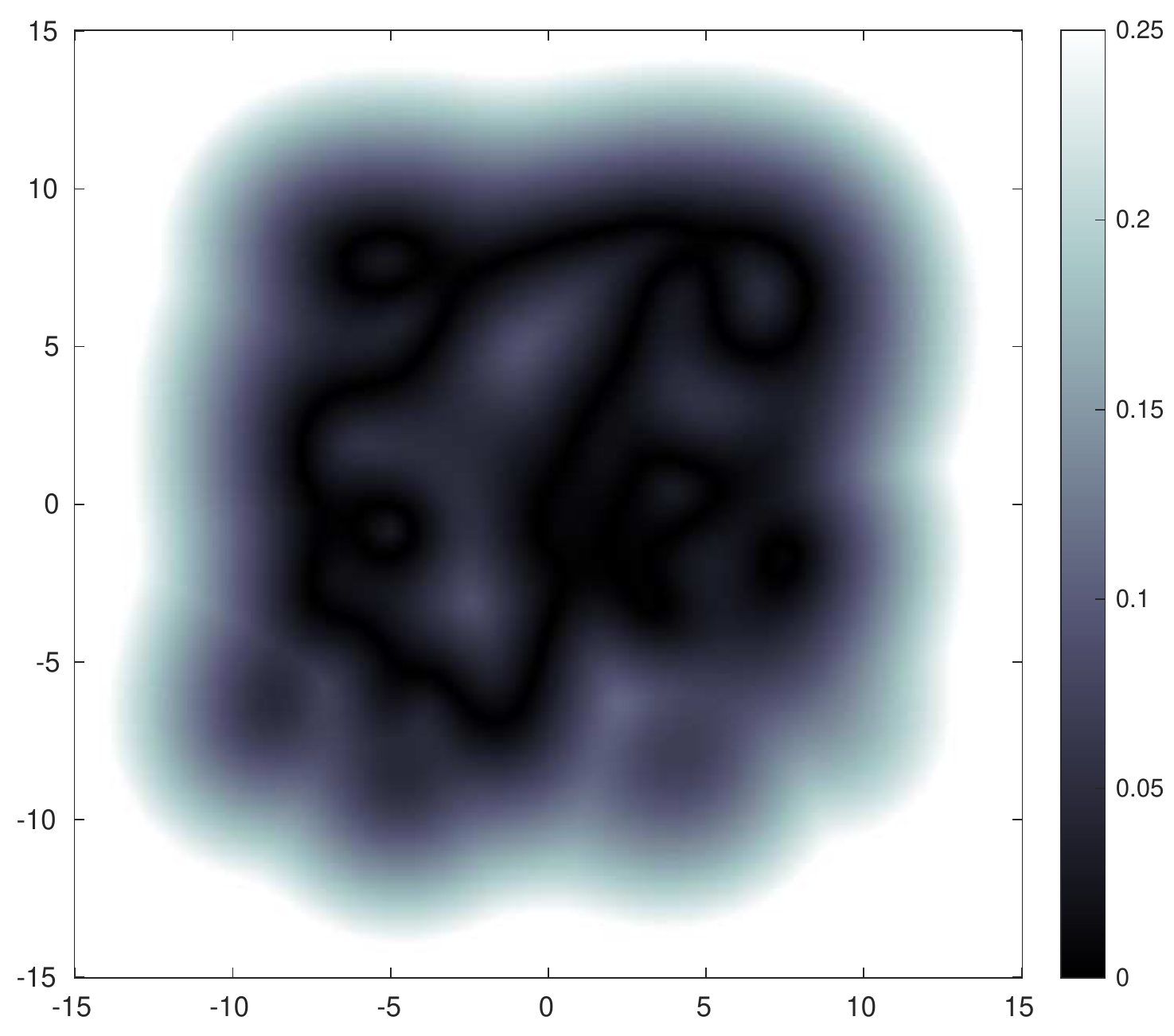}\quad{}\includegraphics[viewport=65bp 25bp 490bp 400bp,clip,scale=0.38]{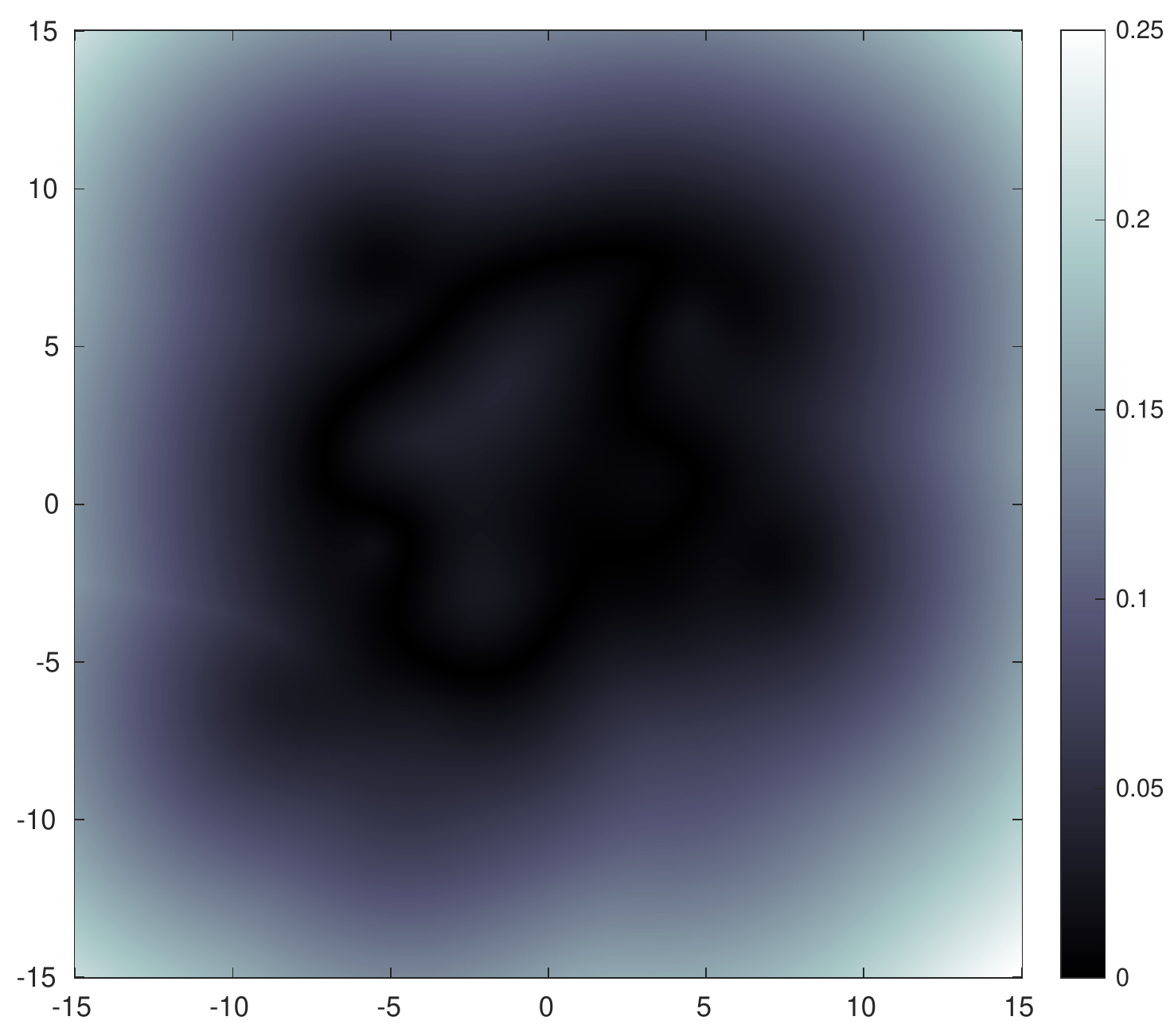}\caption{A slice of the pseudospectrum at $\lambda_{3}=E_{F}=-2.4$. Strong
disorder, various $\kappa$. \label{fig:slice_2.4_disorder}}
\end{figure}

Certainly the resulting algorithm seems fast. Figure~\ref{fig:Timing_localizer}
shows how the time to compute a single localizer index grows at a
smaller order than the growth we saw for the Bott index. It is not
so hard to compute a single localizer index for a system with $L\approx1000$
\cite{loring2018bulk}. However, this is not a thorough analysis if
we are looking at a disorder averaged index study. Open boundaries
and periodic boundary conditions will have different effects for the
same system size. Also, the localizer index is, at is heart, a local
index. To see the effect of growing system size, we ought to be working
with a global index. It is possible to make the localizer index behave
globally, but this conversion is not trivial. 

\begin{figure}[t]
\includegraphics[viewport=65bp 25bp 490bp 400bp,clip,scale=0.4]{slice_K1E0_ps}\quad{}\includegraphics[viewport=65bp 25bp 490bp 400bp,clip,scale=0.4]{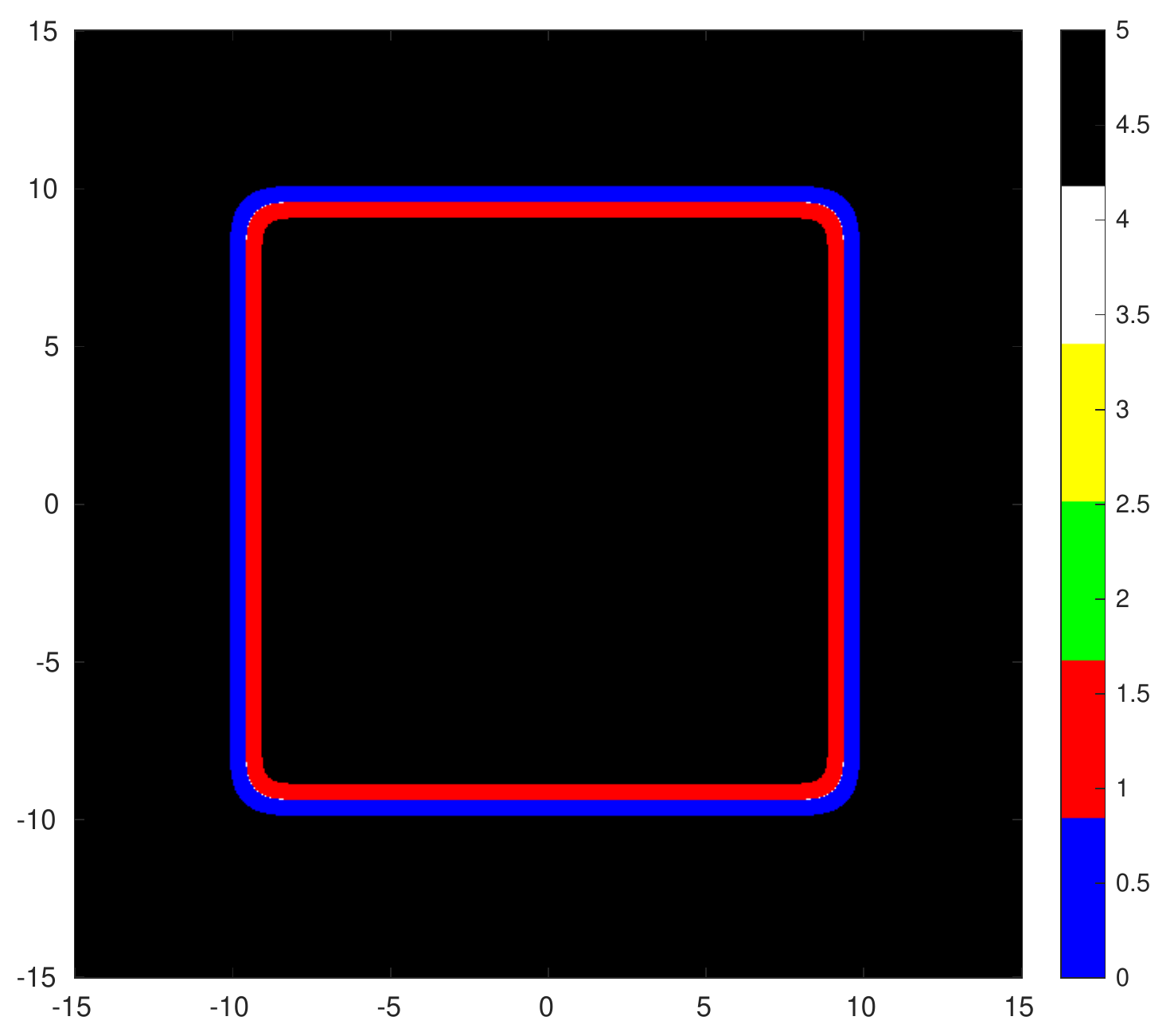}\vspace*{0.3cm}
\\
\includegraphics[viewport=65bp 25bp 490bp 400bp,clip,scale=0.4]{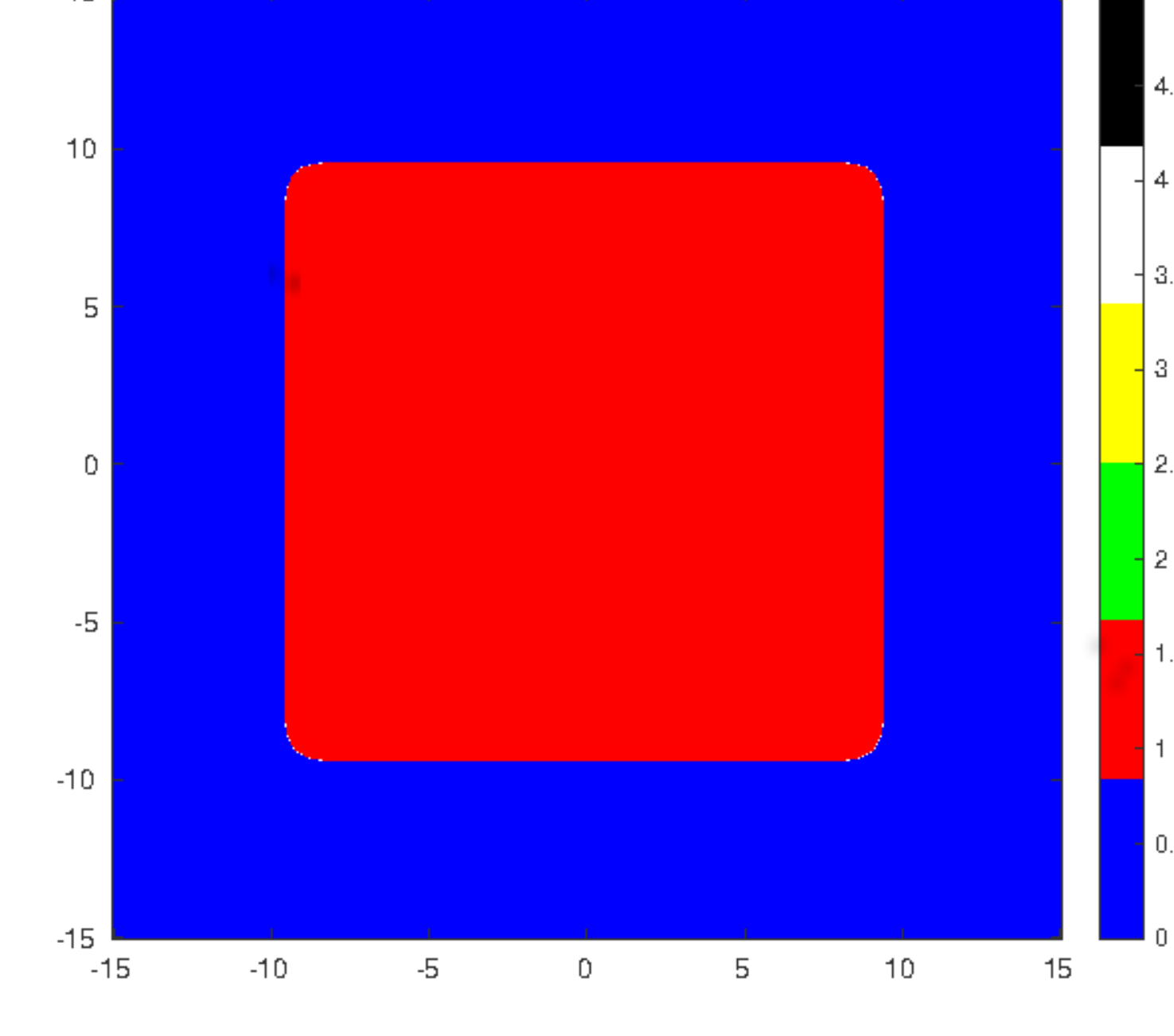}\quad{}\includegraphics[viewport=65bp 25bp 490bp 400bp,clip,scale=0.4]{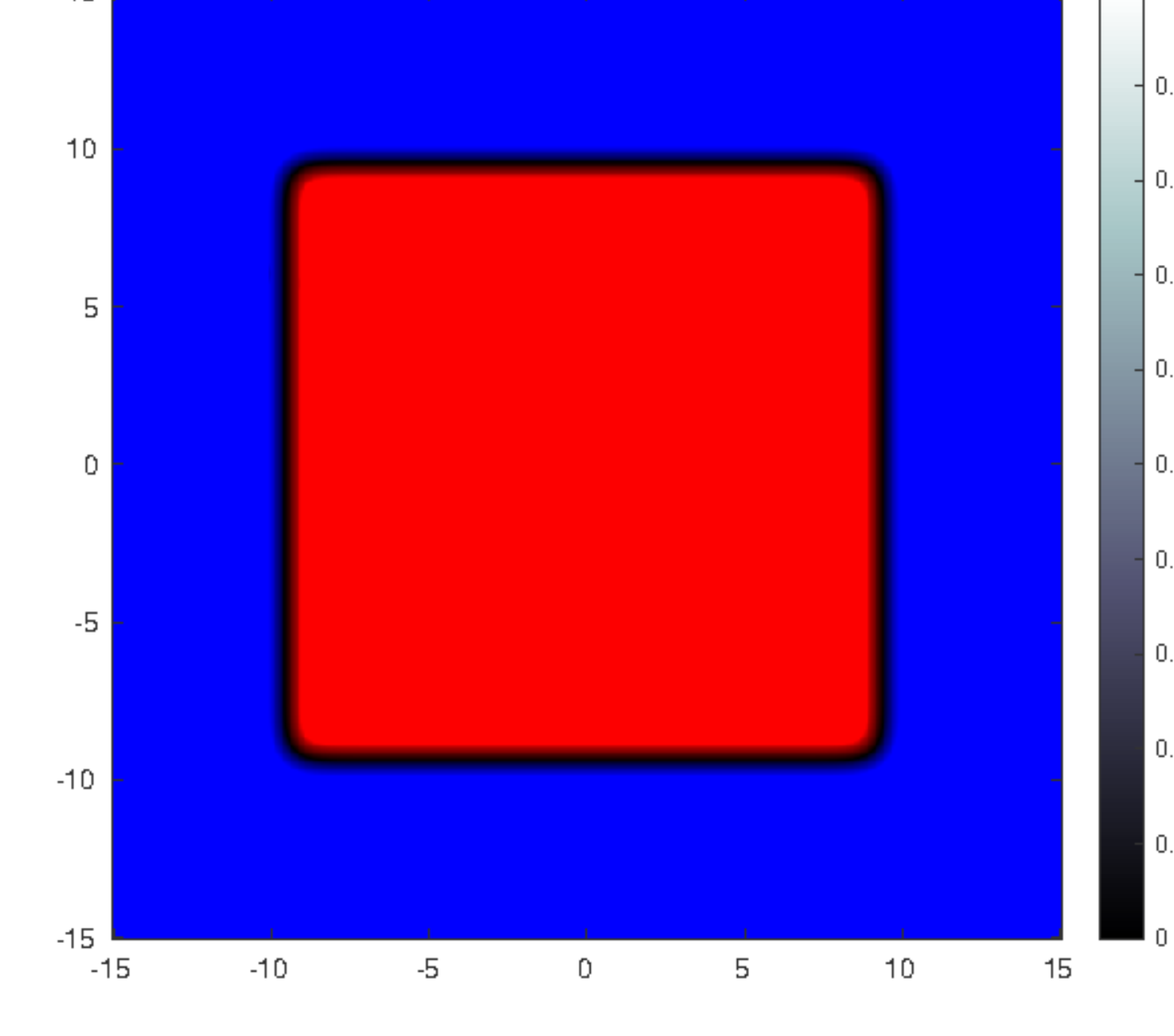}

\caption{Pseudospectrum. Computed index. Extended by using image editing software.
Index overlayed (darken only mode) on pseudospectrum. Blue is index
$0$. Red is index $1$. Total black is where the Clifford spectrum
and points where localizer gap is too small for the index to make
even a little sense. \label{fig:Painting_on_index}}
\end{figure}

Unless we are dealing with a defect (like the hole in the system used
in Figure~\ref{fig:Pseudospectrum_qc-hole}) we are going to set
$(\lambda_{1},\lambda_{2})$ to $(0,0)$, or whatever are the coordinates
of the center of the sample according to $X$ and $Y$. If we have
a disordered system and are sweeping in energy, then we may set $\lambda_{3}=E_{F}$
to a value where the localizer is, according to the computer, singular.
Technically, we should not calculate the index there, but it is often
simpler to treat $0$ as positive when counting eigenvalues. This
happens so infrequently as to be a negligible effect. (This is a guess.
In fact, the author knows of no study on the accuracy of calculating
the the signature of large, sparse Hermitian matrices. Most related
work in numerical linear algebra is on estimating the number of eigenvalues
in a region \cite{lin2016approximating}, which is of no use in the
context of numerical $K$-theory.)

In theory, we can avoid computing the pseudospectrum for a study of
disorder averaged index. However, we need a way to set $\kappa$,
and looking at the pseudospectrum is one way to figure values for
$\kappa$. The theory of how to set $\kappa$ is the subject of continuing
research \cite{L_S-B_finite_vol}.

Some may complain that if we need to tune the localizer index to use
it, then it is not useful. However, many measurement instruments need
tuning. In particular, when using a scanning tunneling microscope
(STM) to do spectroscopy, one must determine the sample to tip gap.
Let us focus on $E_{F}=0$ and ``probing'' at the center of the sample,
so $\boldsymbol{\lambda}=\boldsymbol{0}.$ For simplicity, let us
assume there is no lattice point at $(0,0)$. If $\kappa$ is very
small the localizer index is zero because then
\[
L_{\boldsymbol{0}}(\kappa X,\kappa Y,H)=\left[\begin{array}{cc}
H & \kappa X-i\kappa Y\\
\kappa X+i\kappa Y & -H
\end{array}\right]\approx\left[\begin{array}{cc}
H & 0\\
0 & -H
\end{array}\right]
\]
and the matrix on the right has spectrum that is symmetric across
$0$. Also the localizer index is zero if $\kappa$ is large because
then 
\[
L_{\boldsymbol{0}}(\kappa X,\kappa Y,H)=\left[\begin{array}{cc}
H & \kappa X-i\kappa Y\\
\kappa X+i\kappa Y & -H
\end{array}\right]\approx\kappa\left[\begin{array}{cc}
0 & X-iY\\
X+iY & 0
\end{array}\right]
\]
and, since $X+iY$ is normal we again find spectrum that is symmetric
across $0$.

\begin{figure}[tp]
{\footnotesize{}\makebox[10cm][l]{\raisebox{-0.08cm}[0.0cm][0.1cm]{$\kappa=1$, $E=0$ disorder $0$\hspace*{2.3cm}$\kappa=1$, $E=-2.4$ disorder $0$}}}\\
\includegraphics[viewport=65bp 25bp 490bp 400bp,clip,scale=0.4]{slice_K1E0_both}\quad{}\includegraphics[viewport=65bp 25bp 490bp 400bp,clip,scale=0.4]{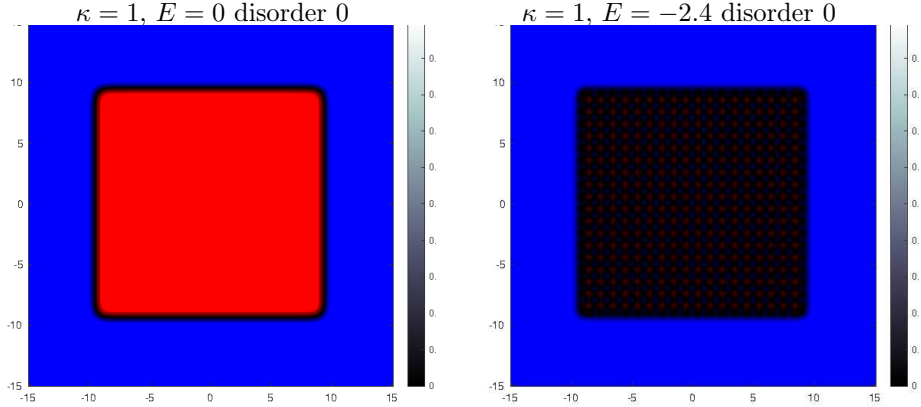}

\caption{Index as color overlay on pseudospectum, no disorder, at $\lambda_{3}=E_{F}=0$
and $\lambda_{3}=E_{F}=-2.4$. \label{fig:Index_clean_system} }
\end{figure}
Finding the intermediate values of $\kappa$ that work can be done
in several ways. One way is to look at the joint approximate eigenvectors
\cite[Lemma 1.2]{LoringPseudospectra} for $X$, $Y$ and $H$ that
can be extracted from the localizer, comparing the deviation of that
state in position to system size and the deviation of that state in
energy compared to the predicted bulk gap. Here we take a more visual
approach.

First we look at a slice of the pseudospectrum at two energy levels,
across a small sample, shown in Figures~\ref{fig:slice_zero_clean}
and \ref{fig:slice_2.4_clean}. We see that for large $\kappa$, the
image looks somewhat like STM microscopy imagery, just showing us
where the lattice points are. That is, $H$ is almost completely ignored.
For small $\kappa$, it is the position information that is getting
ignored. What we are looking for is a clear distinction between edge
effects at $E_{F}=0$ and bulk effects at $E_{F}=-2.4$. This is all
done with the Hamiltonian used in Section~\ref{sec:Extending-old_study},
and with $L=20$. Looking at Figures~\ref{fig:slice_zero_clean}
and \ref{fig:slice_2.4_clean} we might think $0.05\leq\kappa\leq1$
can work.

We also want a stability against disorder. With the rather strong
disorder as used in Section~\ref{sec:Extending-old_study}, and just
one sample, the same slices are shown in Figures~\ref{fig:slice_zero_disorder}
and \ref{fig:slice_2.4_disorder}. At $E_{F}=0$, as shown in Figure~\ref{fig:slice_zero_disorder},
we see still a clear boundary and lots of area inside the boundary
which it is safe to assume has index $0$ where defined, at least
for $0.1\leq\kappa\leq1$. The data in Figure~\ref{fig:slice_2.4_disorder}
is harder to interpret without the $K$-theory data. What we expect,
based on the average Bott index study as in Figure~\ref{fig:Bott-index-averaged},
is that the slice at $E_{F}=-2.4$ with this level of disorder should
have, approximately by area, equal parts index $0$ and index $1$.

\begin{figure}[tp]
{\footnotesize{}\makebox[10cm][l]{\raisebox{-0.08cm}[0.0cm][0.1cm]{$\kappa=1$, $E=0$ disorder $8$\hspace*{2.3cm}$\kappa=1$, $E=-2.4$ disorder $8$}}}\\
\includegraphics[viewport=65bp 25bp 490bp 400bp,clip,scale=0.4]{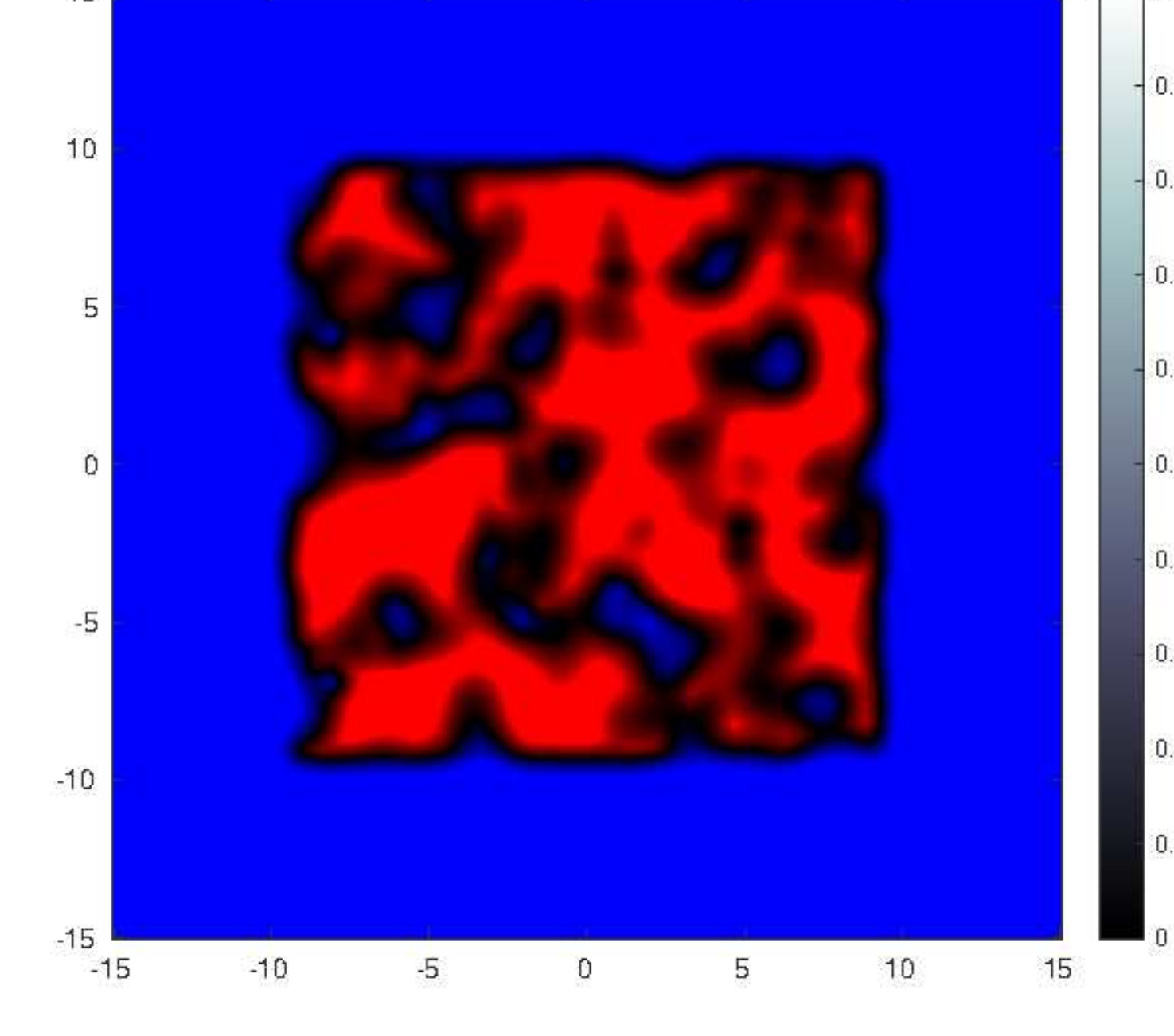}\quad{}\includegraphics[viewport=65bp 25bp 490bp 400bp,clip,scale=0.4]{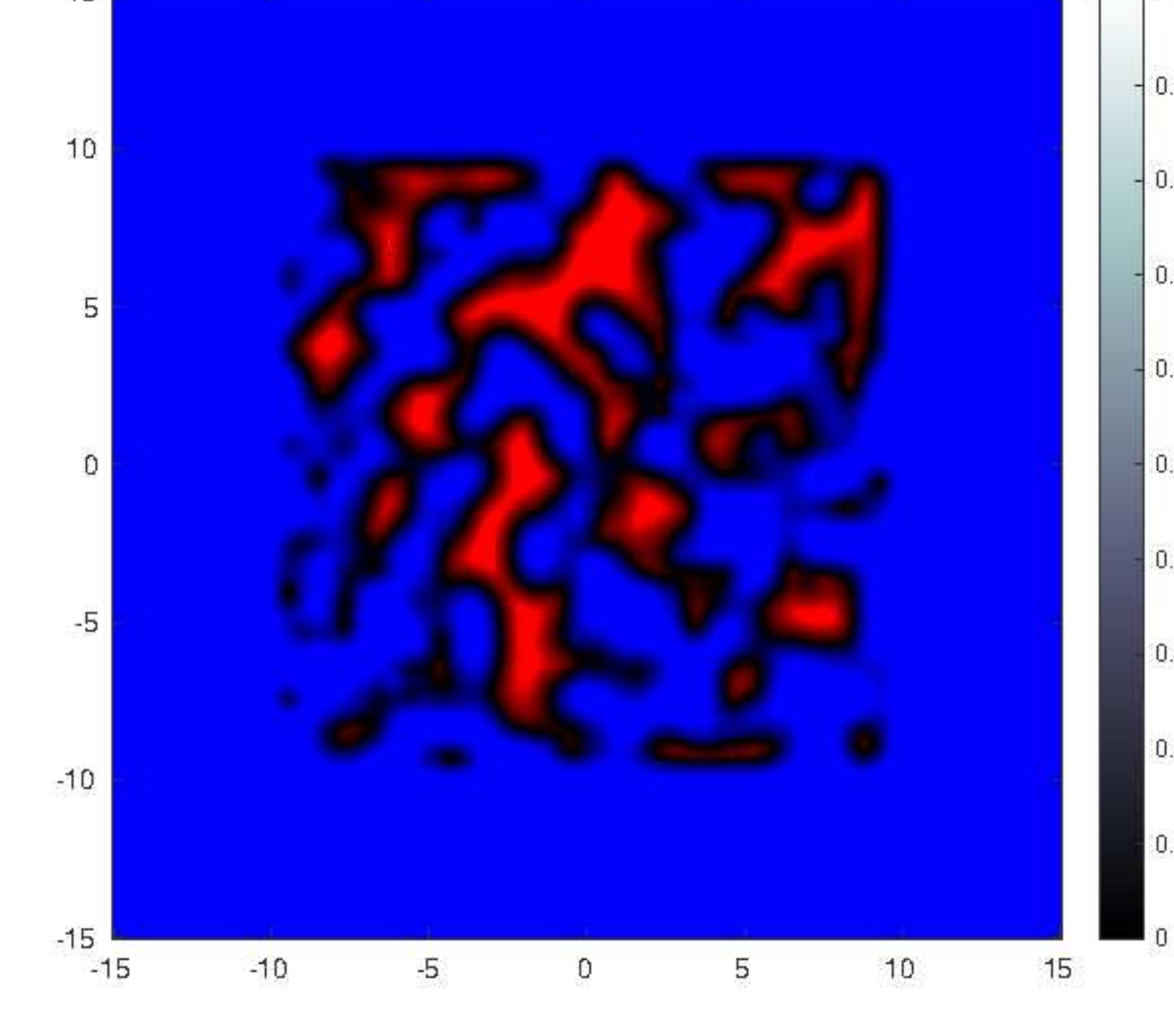}

\caption{Index as color overlay on pseudospectum, with disorder, at $\lambda_{3}=E_{F}=0$
and $\lambda_{3}=E_{F}=-2.4$. \label{fig:Index_disorder} }
\end{figure}

Computing multivariate pseudospectrum has not caught the attention
of numerical analysts. The method the author has developed is rather
crude. It does have one optimization, which is that when a large gap
(meaning
\[
\mathrm{gap}_{\lambda}(M_{1},M_{2},M_{3})=\left\Vert \left(L_{\boldsymbol{\lambda}}(\kappa X,\kappa Y,H)\right)^{-1}\right\Vert ^{-1}
\]
or the smallest absolute value of an eigenvalue of the localizer)
is found, there is no need to compute the gap a near locations. This
is based on \cite[Lemma 7.2]{LoringPseudospectra}. In regions where
the gap stays enough above zero (say $10^{-8})$ where we can trust
the numerical calculations, the index cannot change. As such, well-optomized
software would first compute the pseudospectrum, ignoring regions
where the gap is large, then compute the contiguous gapped regions,
then compute the index at one point in each region. For now, we don't
have such an optomization, and compute the index everywhere where
the gap is neither too small to be meaningful or too large to be necessary.
On hopes an applied mathematician will find this of interest in the
near future.

The index data computed is as shown in the top-left panel of Figure~\ref{fig:Painting_on_index}.
The author shamelessly used image processing software to extend the
index data to contiguous regions and then used a darking-only overlay
of the index color data over the pseudospectrum. Regions that are
too dark to distinguish hue are where the local index is nearly meaningless.
It is differing index in locations with a decent gap that are of interest,
following the logic of \cite[Lemma 7.5]{LoringPseudospectra}. This
seems like a good way to visualize the index and pseudospectrum together.

Recall that the pseudospectrum is a grey scale map on three space,
and it is regions of 3-space we should be coloring. Producing and
displaying such information is difficult to do with current algorithms.
An attempt to display the full Clifford pseudospectrum was made in
Figures~\ref{fig:Pseudospectrum_qc} and \ref{fig:Pseudospectrum_qc-hole},
and one sees the large interior space, in the shape of a distorted
cube in Figure~\ref{fig:Pseudospectrum_qc} and a distorted solid
torus in Figure~\ref{fig:Pseudospectrum_qc-hole}, where the index
is nonzero. There seem to be smaller regions that most likely have
nonzero index at higher and lower energies, but marking this by index
will be difficult to even display. 

For now, we are just trying explain what are good values of $\kappa$
for systems here of $L=20$. Figure~\ref{fig:Index_clean_system}
shows the index data with the pseudospectral data for $\kappa=1$
and a clean system. Figure~\ref{fig:Index_disorder} shows the same
but now with disorder, showing ``in the bulk'' mainly index $1$ at
$E_{F}=0$ and a blend of index $0$ and index $1$ at $E_{F}=-2.4$.

The fact that the left panel of Figure~\ref{fig:Index_disorder}
shows some index $0$ well in from the edge indicates a difficulty
we have in dealing with open boundaries. Changing $\kappa$ does not
make this go away. What works is moving to a larger system, say $L=60.$
This is above where we can easily make pictures. What was advocated
in \cite{LoringPseudospectra} was to keep $\kappa$ constant as $L$
increases to create a local invariant, and to set $\kappa=C/L$ for
some fixed $C$ to create a global index. Setting $\kappa=C/L$ may
well work in a gapped system. It seems that when dealing with a mobility
gap the situation can be more subtle.

\section{Localizer index goes global \label{sec:Localizer-as-global}}

\begin{figure}
\includegraphics[viewport=22bp 15bp 410bp 285bp,clip,scale=0.6]{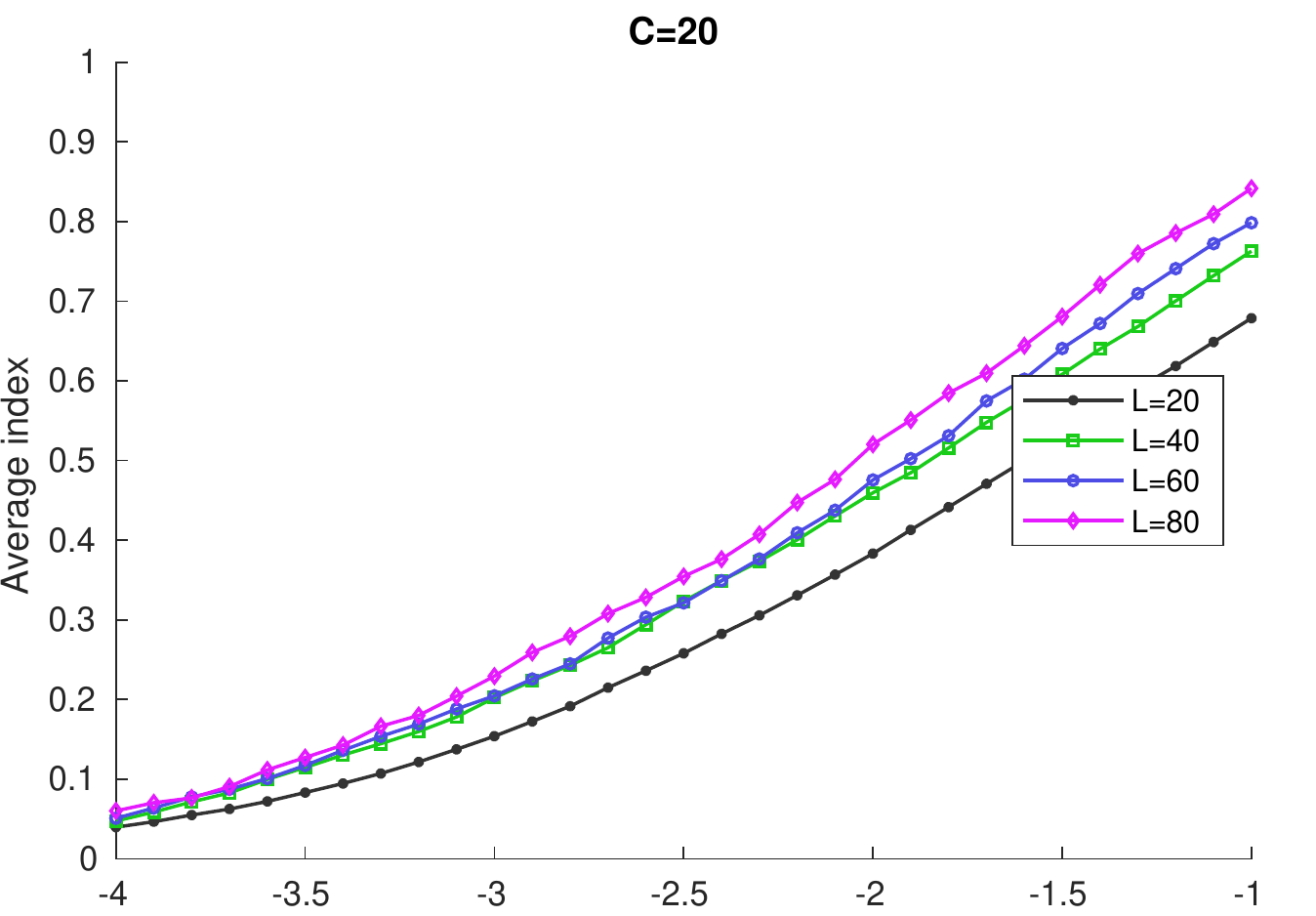}\includegraphics[viewport=22bp 15bp 410bp 285bp,clip,scale=0.6]{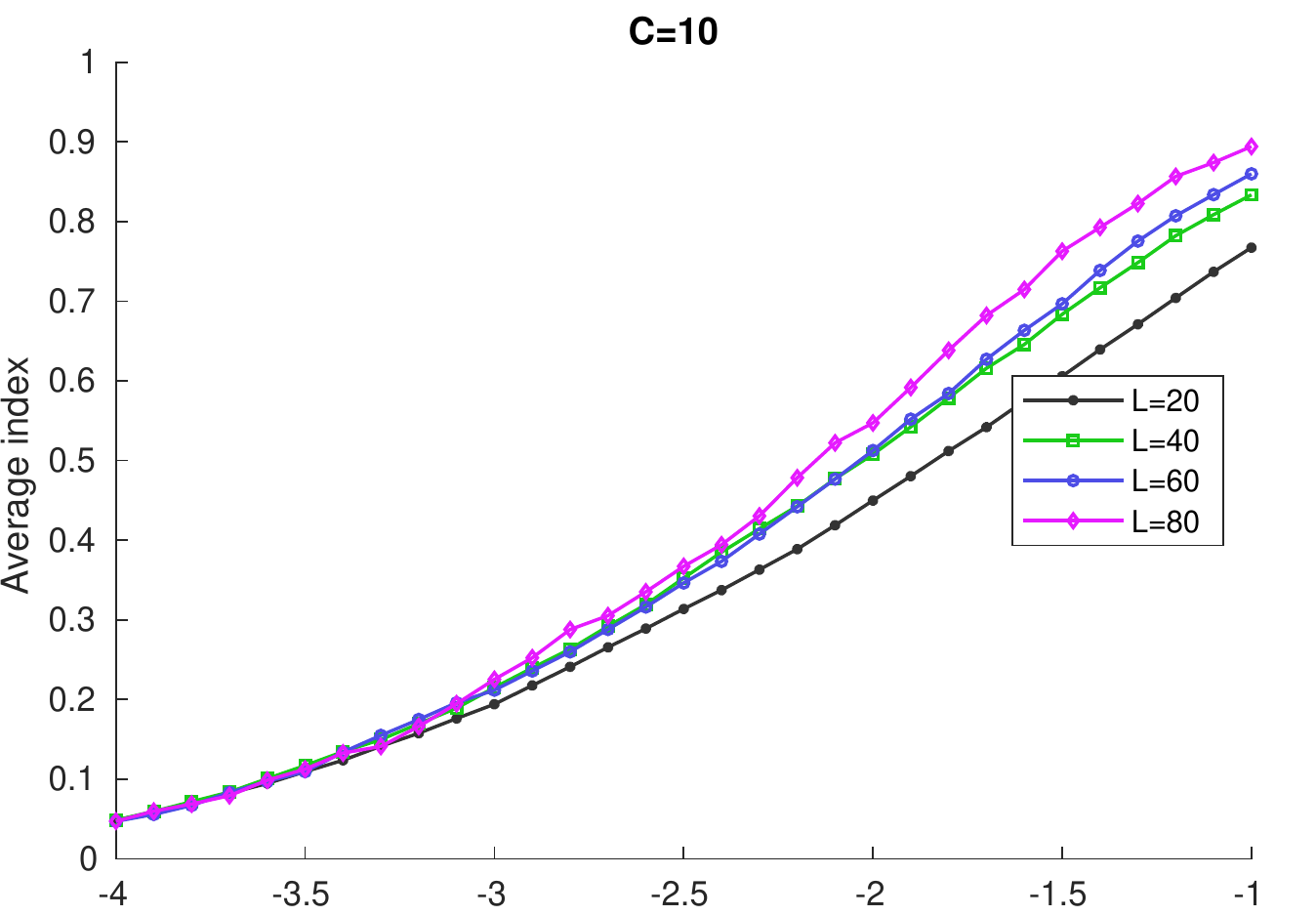}\\
\includegraphics[viewport=22bp 15bp 410bp 285bp,clip,scale=0.6]{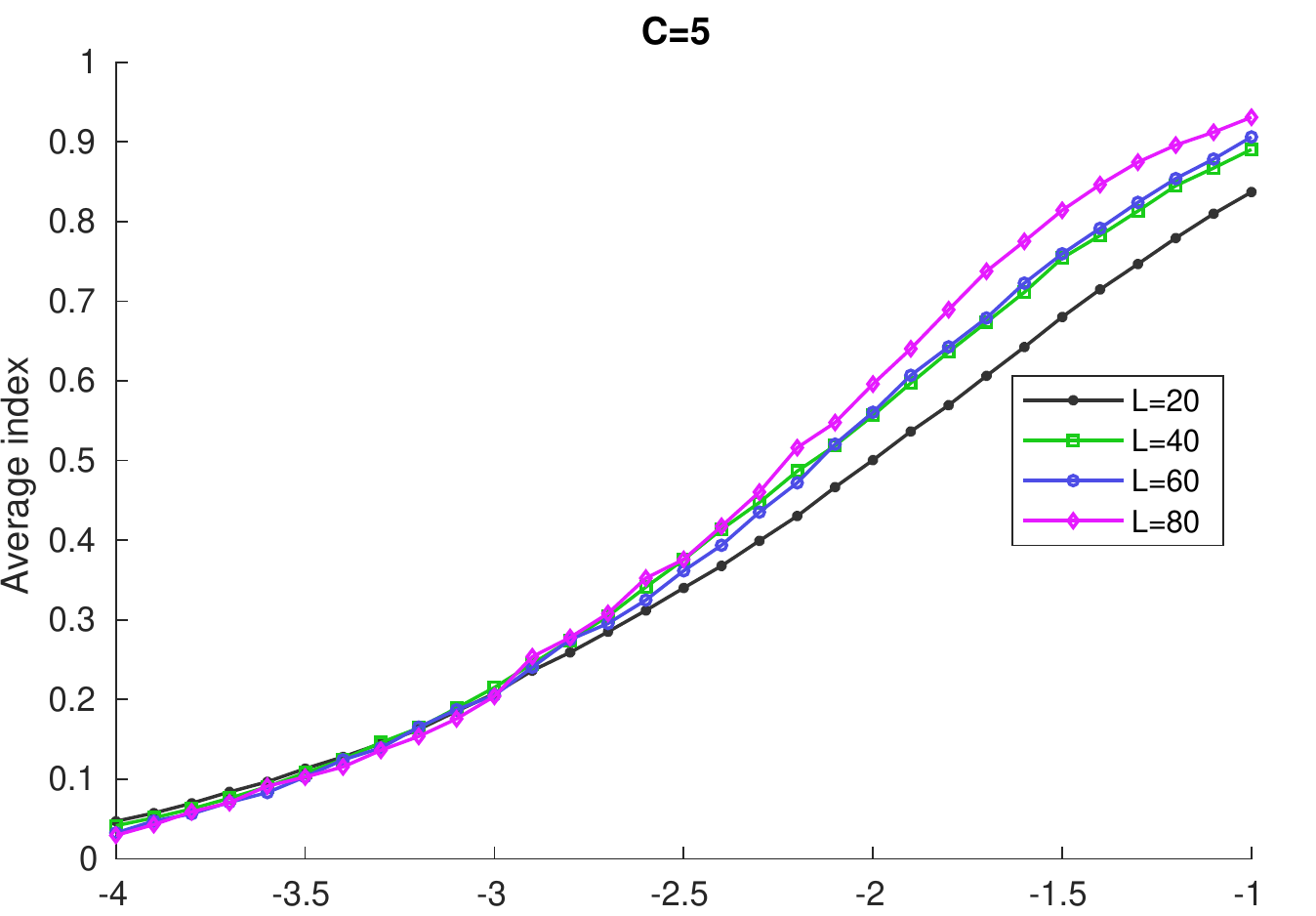}\includegraphics[viewport=22bp 15bp 410bp 285bp,clip,scale=0.6]{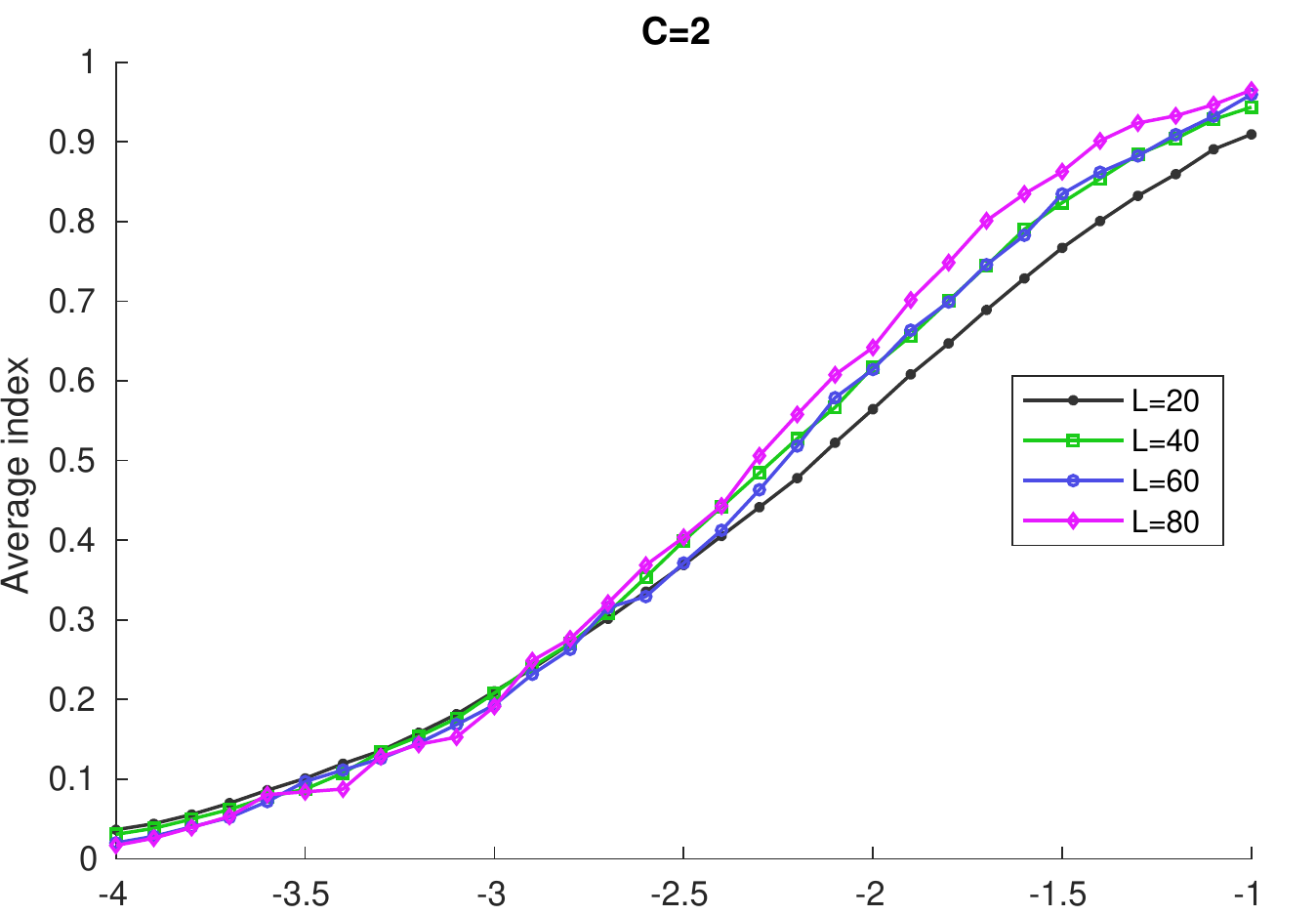}\\
\includegraphics[viewport=22bp 15bp 410bp 285bp,clip,scale=0.6]{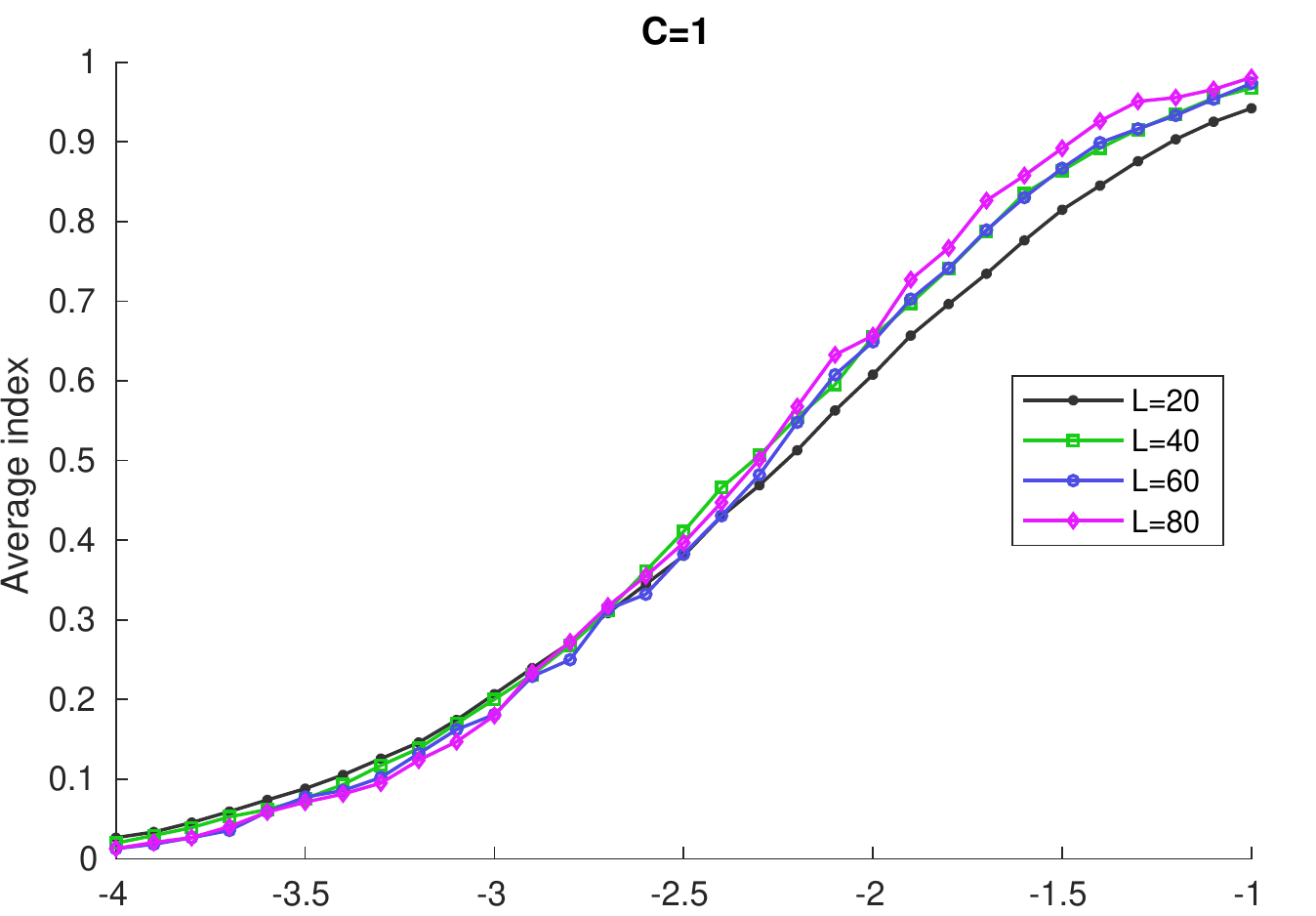}\includegraphics[viewport=22bp 15bp 410bp 285bp,clip,scale=0.6]{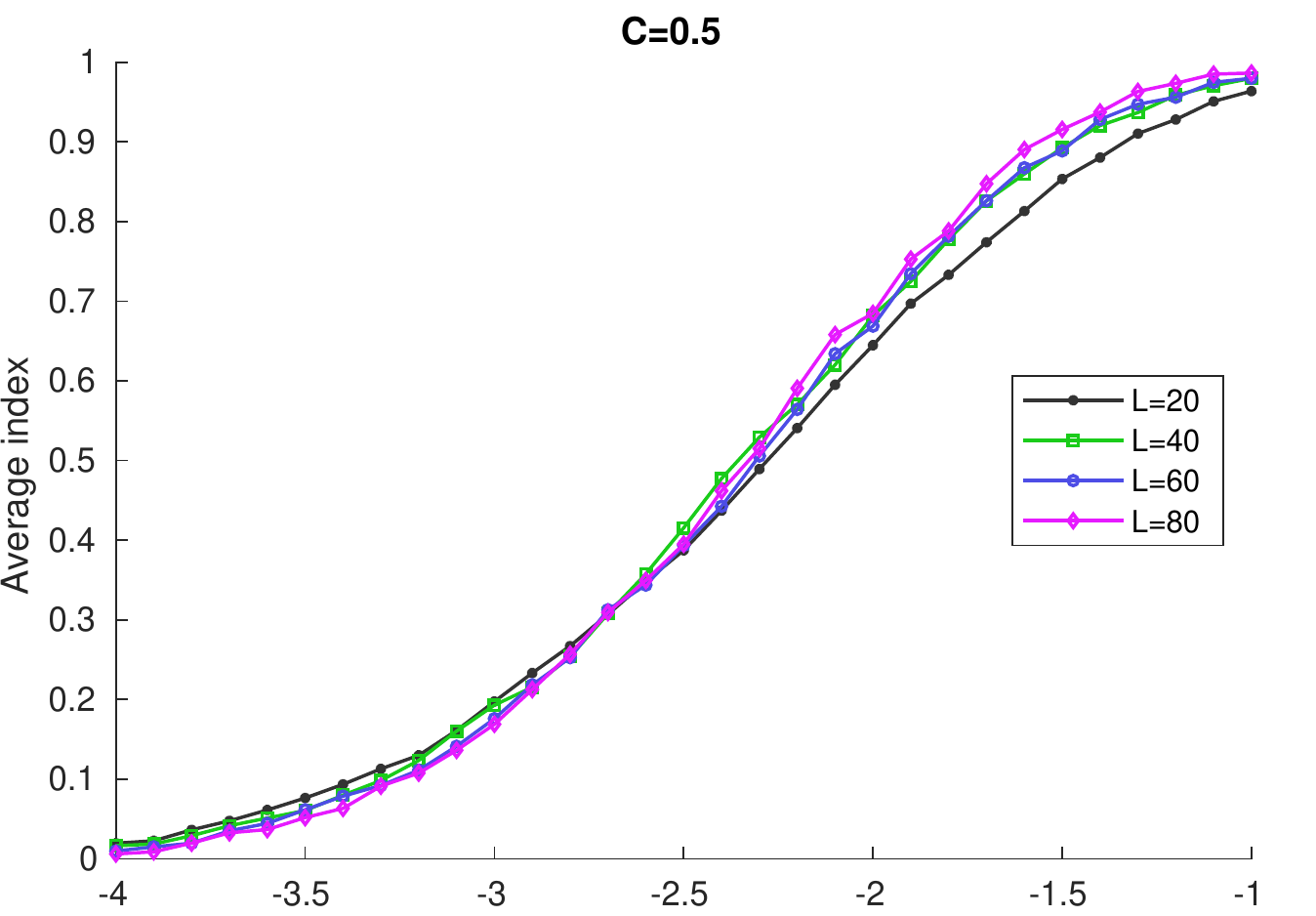}\\
\includegraphics[viewport=22bp 15bp 410bp 285bp,clip,scale=0.6]{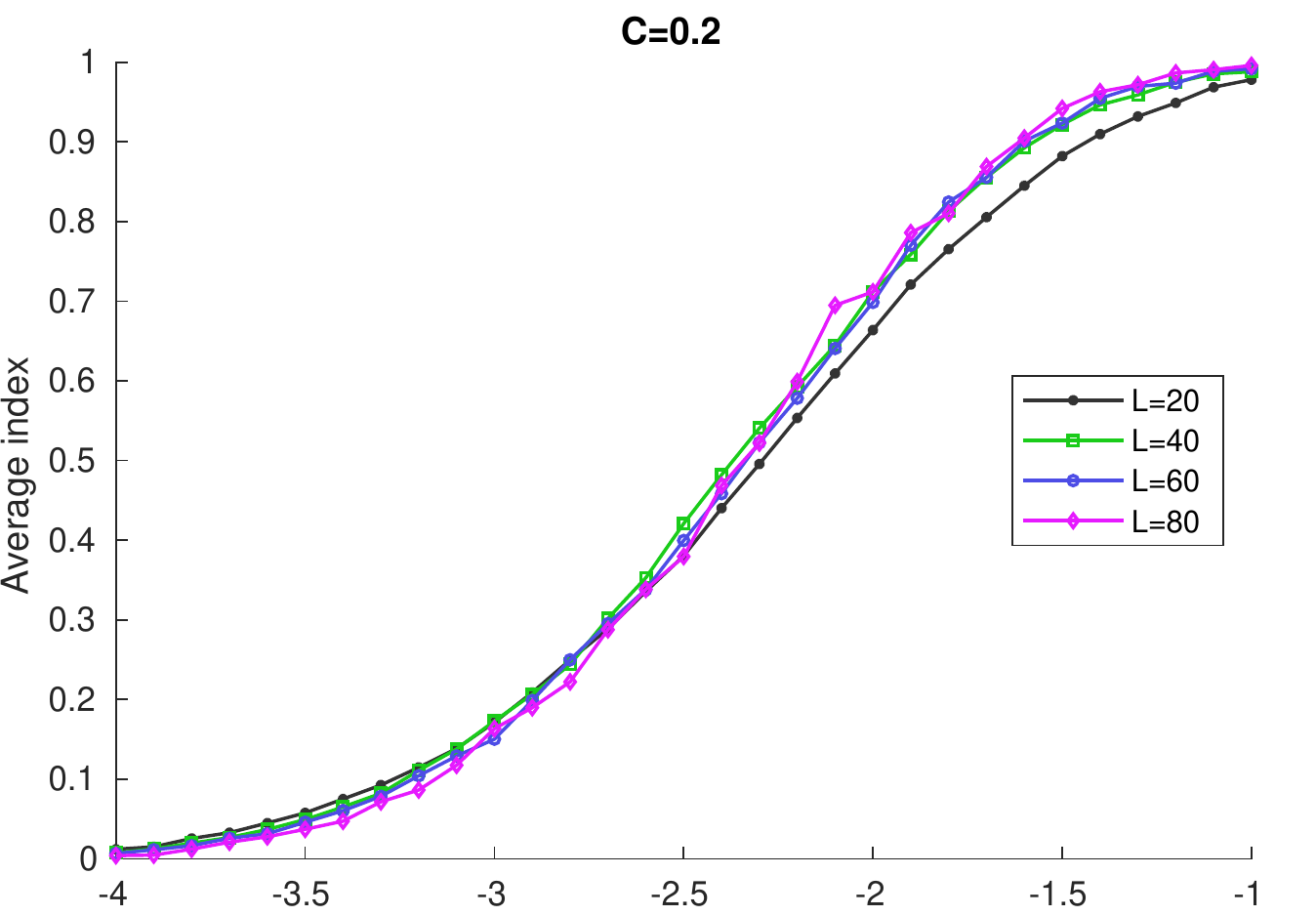}\includegraphics[viewport=22bp 15bp 410bp 285bp,clip,scale=0.6]{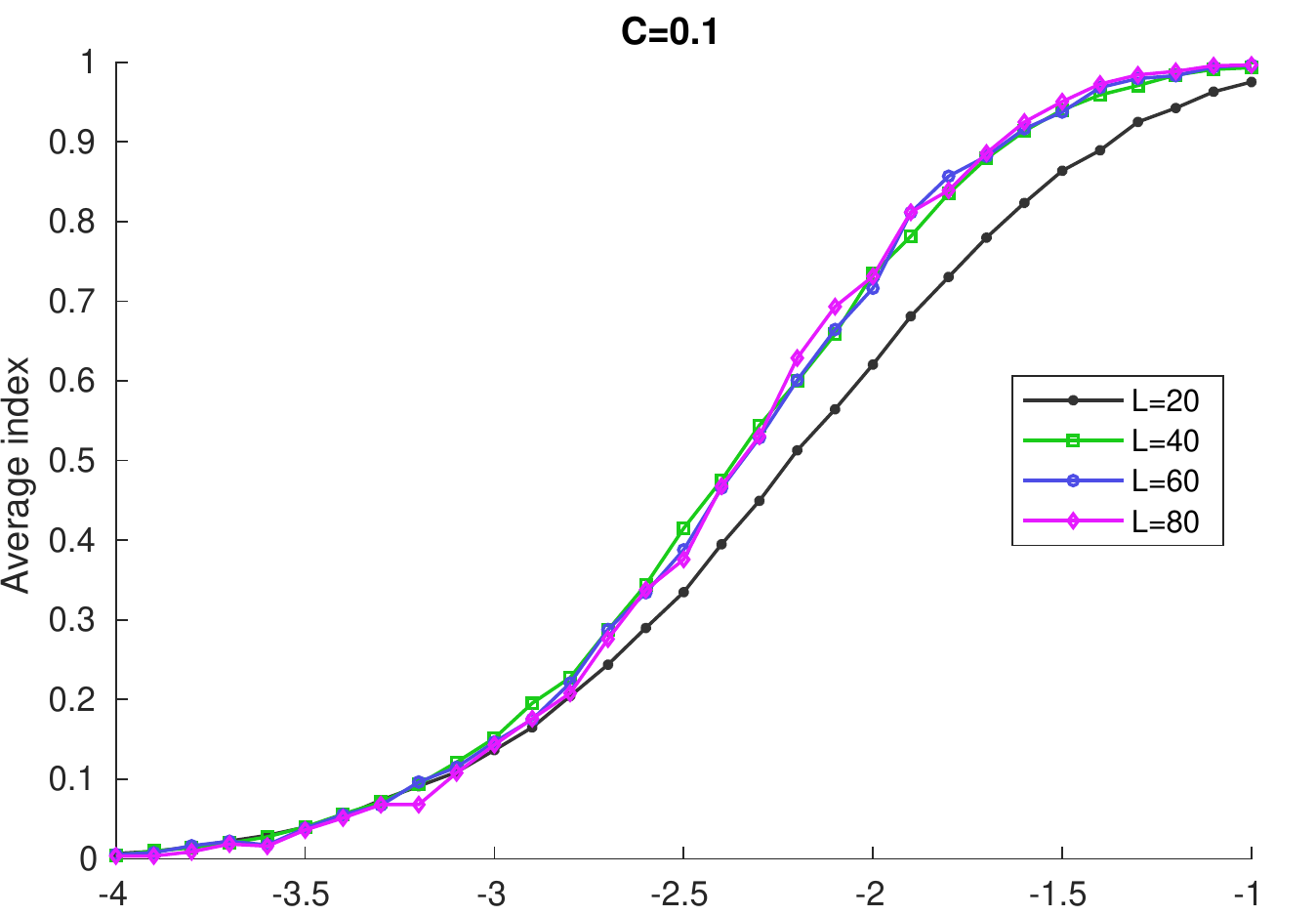}

\caption{Disorder averaged localizer index, with $\kappa=C/L$. \label{fig:kappa_C/L}}
\end{figure}

\begin{figure}
\noindent \begin{raggedright}
\includegraphics[viewport=22bp 15bp 410bp 285bp,clip,scale=0.6]{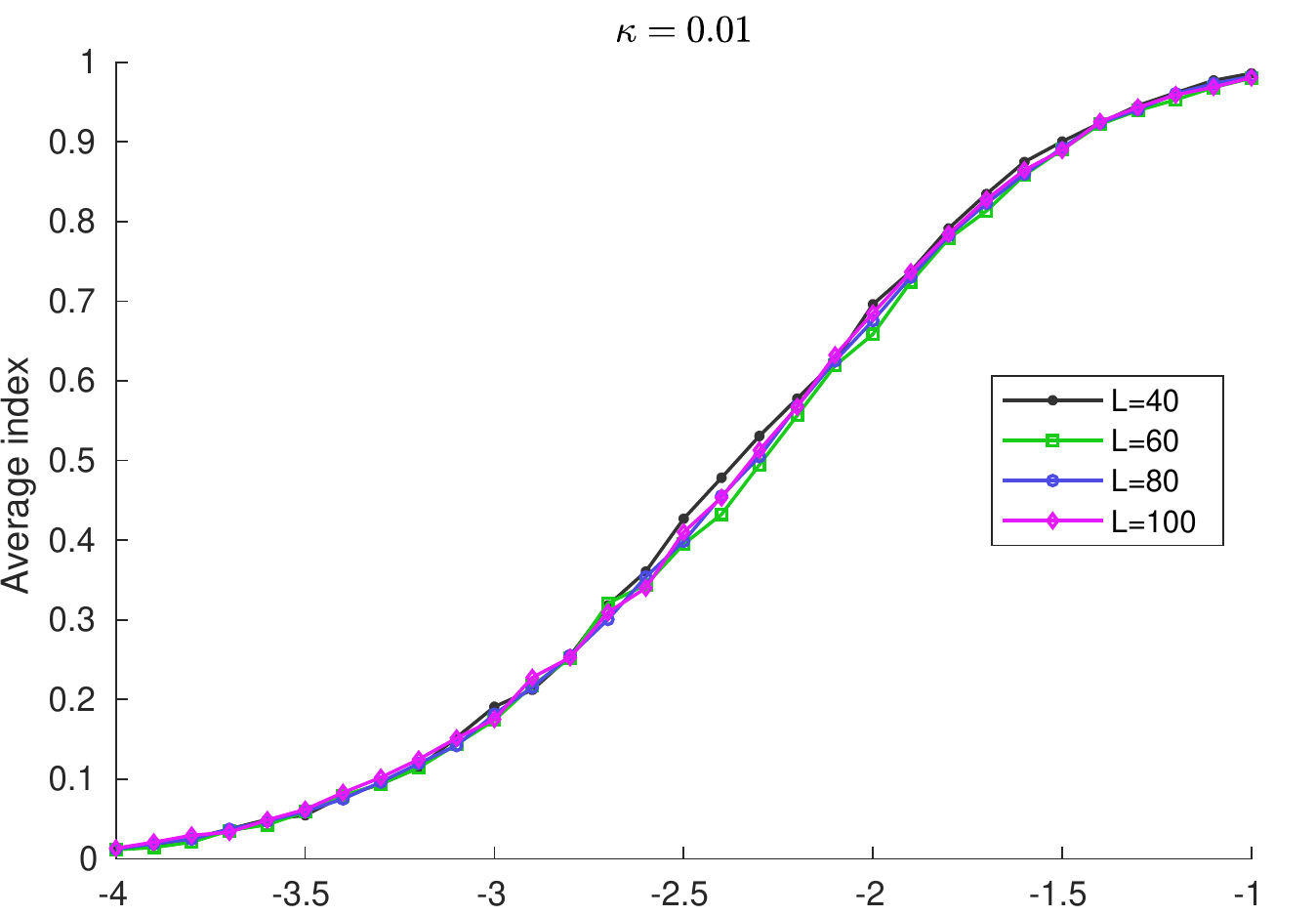}\includegraphics[viewport=22bp 15bp 410bp 285bp,clip,scale=0.6]{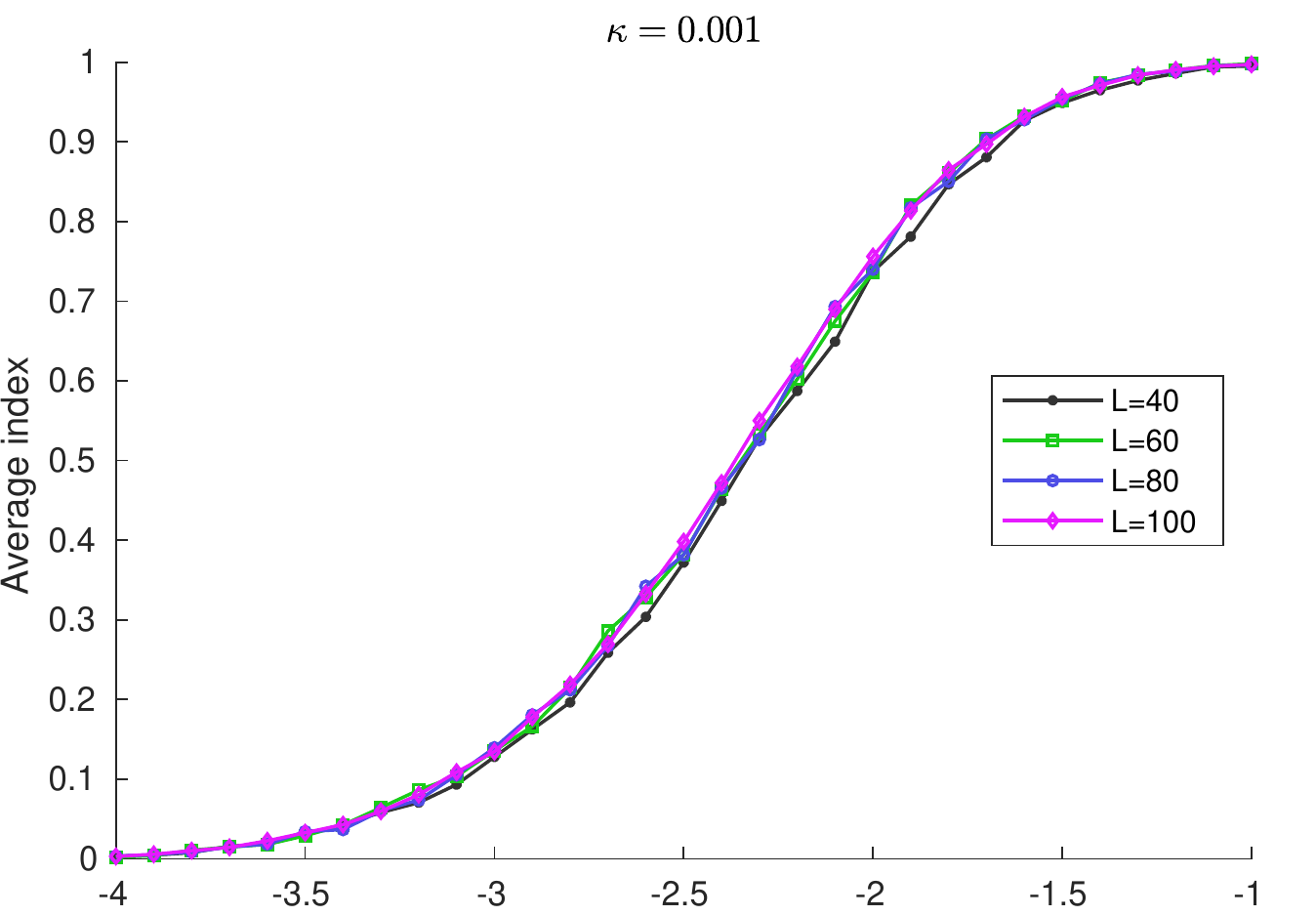}\\
\includegraphics[viewport=22bp 15bp 410bp 285bp,clip,scale=0.6]{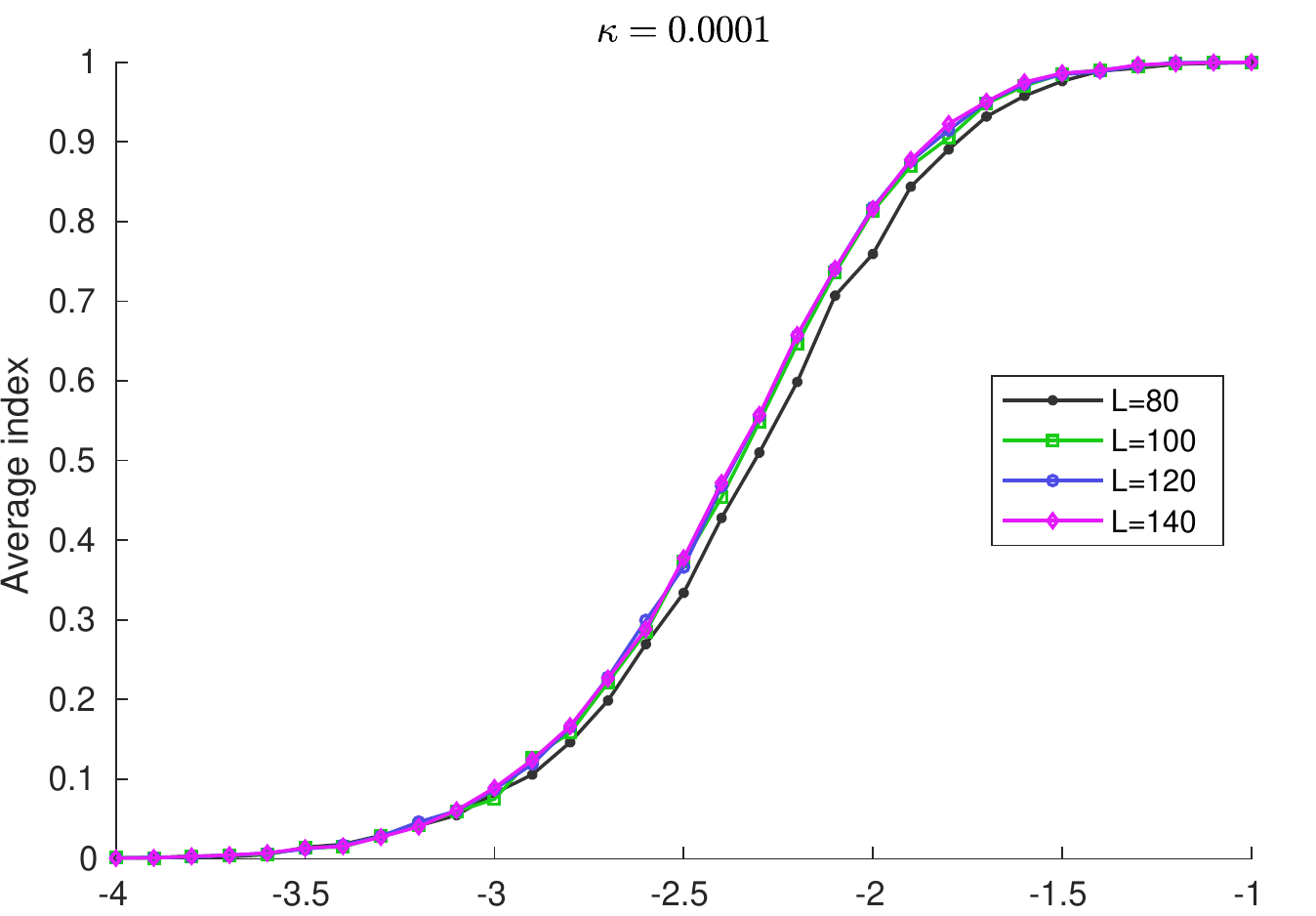}\includegraphics[viewport=22bp 15bp 410bp 285bp,clip,scale=0.6]{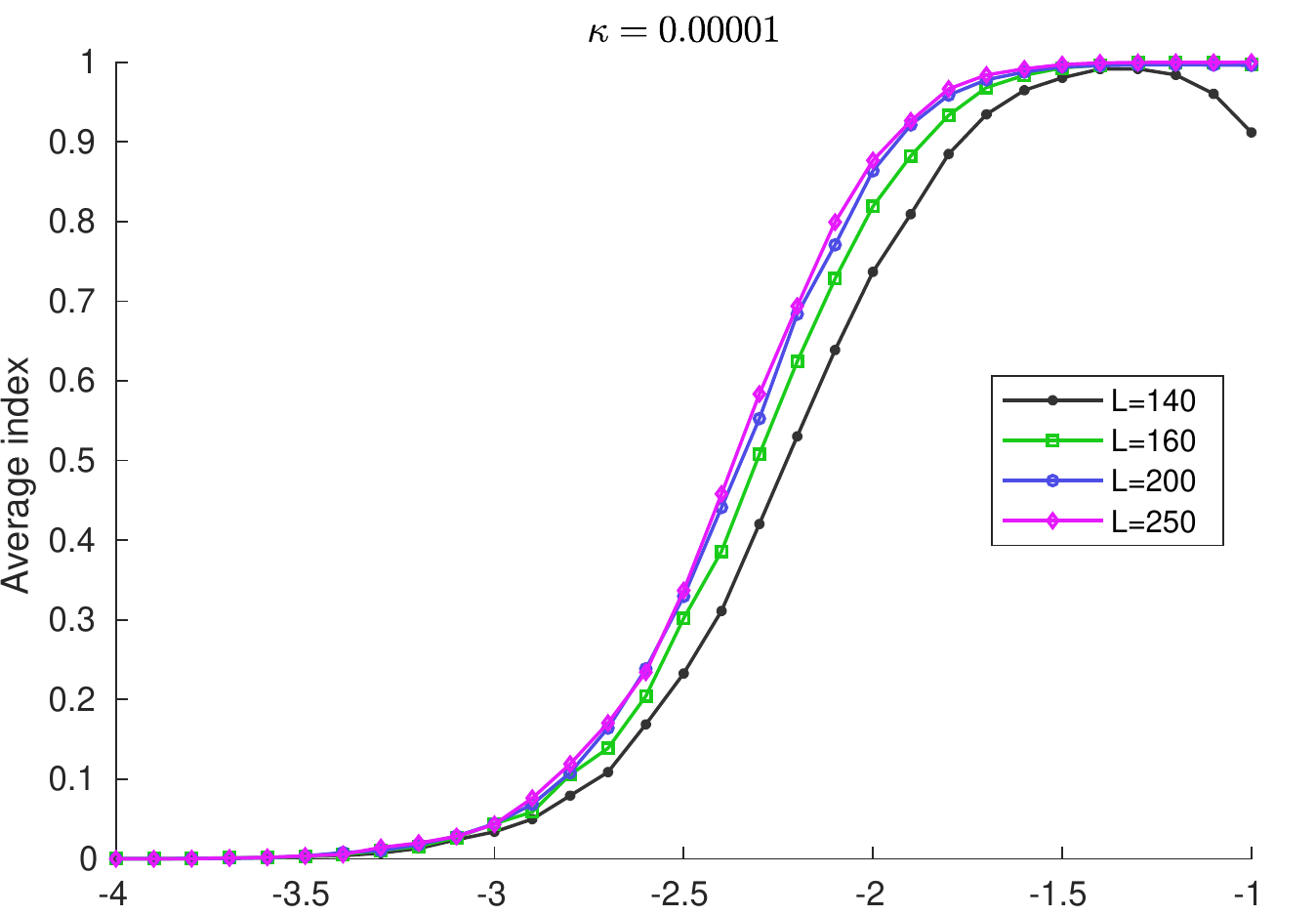}\\
\includegraphics[viewport=22bp 15bp 410bp 285bp,clip,scale=0.6]{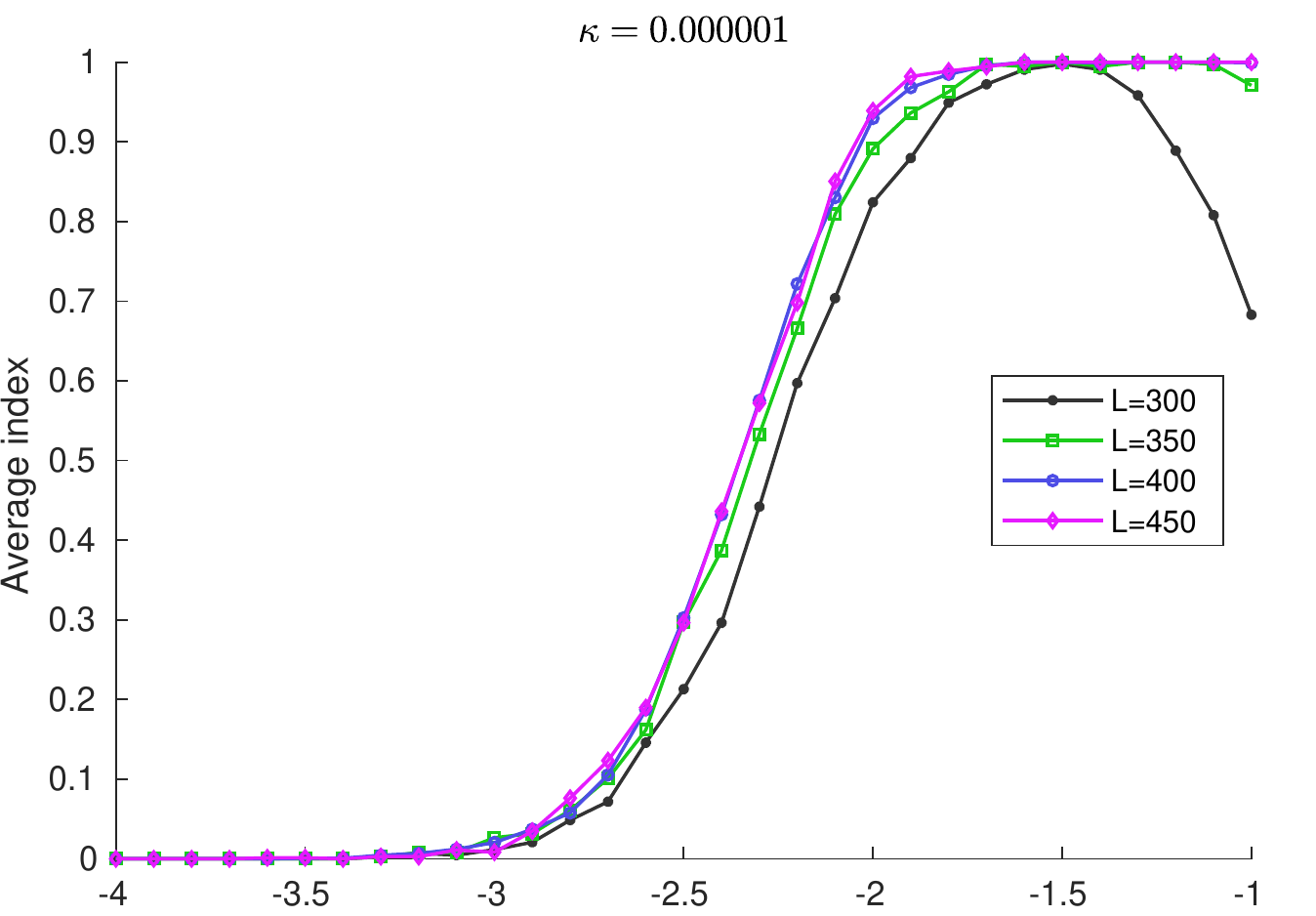}
\par\end{raggedright}
\caption{Various $\kappa$ values, searching for appropriate value for $L$.
\label{fig:Fix_kappa_increase_L}}
\end{figure}

We start with doing the disorder average study with $L=20,40,60,80$
using the localizer index, setting $\kappa=C/L$. We have an idea,
based on setting $\kappa$ for $L=20$, that we want $0.1\leq\frac{C}{20}\leq1$
and so $2\leq C\leq20$. Unfortunately, various values in and around
$C$ fail to lead to convincing results, as shown in Figure~\ref{fig:kappa_C/L}.
None of these look at all like Figure~\ref{fig:Bott-index-averaged}. 

Looking at the plots in Figure~\ref{fig:Fix_kappa_increase_L}, and
many others, it seemed the steeper plots correspond to smaller $\kappa$,
but not for all $L$. So perhaps it is $\kappa\rightarrow0$ here
that plays the role of $L\rightarrow\infty$ when working with the
Bott index. Various calculations, for example in \cite{loring2018bulk,L_S-B_finite_vol},
tell us that the localizer gap often converges when $L\rightarrow\infty$.
See Lemma~\ref{lem:error_Delta_H}.

Figure~\ref{fig:Fix_kappa_increase_L} shows what happens for various
values of $\kappa$, computing the disorder averaged localizer index
at a range of Fermi levels, as $L$ increase somewhat. We emphasize
that we stick with $\lambda_{1}=\lambda_{2}=0$ unless we have a good
reason, like a hole in our system, so that is the one choice we make
in this section and the last. These are too computationally intensive
to extend much further, but we have soft evidence that these plots
have essentially converged at the higher values for $L$ shown in
each instance. 

\begin{figure}
\includegraphics[viewport=20bp 0bp 410bp 275bp,clip,scale=0.6]{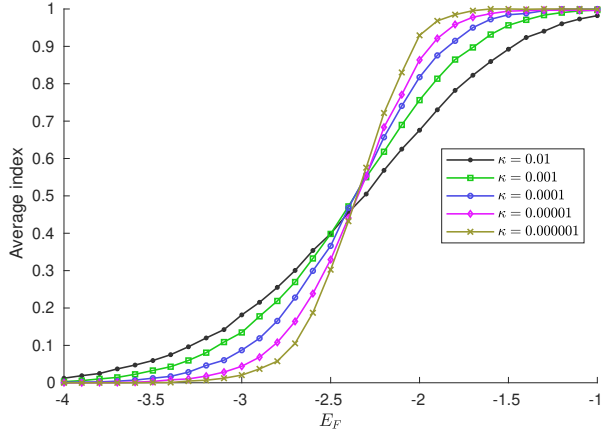}

\caption{Localizer index averaged over disorder as a function of the Fermi
level: $\kappa=0.01$ (16,955 samples with L=80); $\kappa=0.001$
(14,878 samples with L=100); $\kappa=0.0001$ (5,598 samples with
L=120); $\kappa=0.00001$ (8,031 samples with L=200); $\kappa=0.000001$
(870 samples with L=400). \label{fig:kappa-index-averaged}}
\end{figure}

There is some theory when one has a bulk gap of a known size, as in
\cite{LoringSchuBa_even_AI}. However, we have just a mobility gap.
Presumably, one could work in an appropriate $C^{*}$-algebra, find
a gapped Hamiltonian that lives in that $C^{*}$-algebra and can extend
the methods of \cite{LoringSchuBa_even_AI} to that setting. For now,
we are happy to gather numerical evidence.

In Figure~\ref{fig:kappa-index-averaged} shows the plots for various
decreasing $\kappa$. Notice they are decreasing by powers of ten.
The plots in Figure~\ref{fig:Fix_kappa_increase_L} were used to
deduce a reasonable values of $L$ to go with each value of $\kappa$.
It is perhaps the case that open boundaries really interact with disorder
and this is what forced us into such large systems to see plots to
create Figure~\ref{fig:kappa-index-averaged}.

\section{Local nature of the localizer \label{sec:Local-nature-localizer}}

Better algorithms are needed to make the spectral localizer effective
for three-dimensional systems, and indeed would help with two-dimensional
systems. Here is one possible improvement.

Suppose one is working on a two-dimensional system with $L\geq300$
for a specific $\boldsymbol{\lambda}$. It should be possible to compute
the gap and index of a further truncation of the system to smaller
size and use those results if that gap is large enough, and only if
this is small recomputing the gap and index on the full system. This
should be possible due to the local nature of the localizer, where
drastic alterations of the Hamiltonian away from the $(\lambda_{1},\lambda_{2})$
have little effect, as we will soon see. Also, it is the smaller values
of the localizer gap that need to be computed precisely.

A glance of Figures~\ref{fig:Pseudospectrum_qc} and \ref{fig:Pseudospectrum_qc-hole}
shows that, at least in one instance, something drastic like cutting
out a hole in the middle of a sample has little effect on the pseudospecturm
at the edge. An analytic expression of this fact is the following
lemma. This was proving jointly with Schulz-Baldes years ago, but
in this form it never found a place in our papers. Many similar, and
more technical results, are in our papers \cite{L_S-B_finite_vol,loringSchuba2019spectral,LoringSchuBa_even_AI}.

For simplicity, we state this regarding computing the spectral gap
for the localizer when $\boldsymbol{\lambda}=\boldsymbol{0}$, and
only for three matrices (so two physical dimensions). So let $L=L_{\boldsymbol{0}}$.
Notice also the lemma is stated in terms of changes to the square
of the localizer gap.
\begin{lem}
\label{lem:error_Delta_H} Suppose $X,$ $Y$ and $H$ are Hermitian
matrices, with $X$ and $Y$ commuting, and set 
\[
\gamma=\left\Vert L(X,Y,H)^{-1}\right\Vert ^{-2}.
\]
Let $Z=X+iY$. If $H_{0}$ is Hermitian with
\begin{equation}
\left\Vert |Z|^{-1}\left(HH_{0}+H_{0}H+H_{0}^{2}\right)|Z|^{-1}\right\Vert \leq C\label{eq:bound_on_H0}
\end{equation}
and
\[
\left\Vert \left[Z,H\right]\right\Vert \leq D
\]
and 
\begin{equation}
\left\Vert |Z|^{-1}\left[Z,H_{0}\right]|Z|^{-1}\right\Vert \leq E\label{eq:bound_on_commutator_H0}
\end{equation}
then
\[
|\gamma-\gamma_{0}|\leq\left(C+E\right)\gamma+\left(C+E\right)D
\]
where
\[
\gamma_{0}=\left\Vert L(X,Y,H+H_{0})^{-1}\right\Vert ^{-2}.
\]
If 
\[
\gamma>\frac{(C+E)D}{1-(C+E)}
\]
then the index at zero for $(X,Y,H+H_{0})$ will be the same as for
$(X,Y,H)$.
\end{lem}

\begin{proof}
Notice $Z$ is normal, and compute: 
\begin{align*}
L(X,Y,H) & =\left[\begin{array}{cc}
H & Z^{\dagger}\\
Z & -H
\end{array}\right],
\end{align*}
\begin{align*}
L(X,Y,H)^{2} & =\left[\begin{array}{cc}
H^{2}+|Z|^{2} & -\left[H,Z\right]^{\dagger}\\
-\left[H,Z\right] & H^{2}+|Z|^{2}
\end{array}\right],
\end{align*}
\begin{align*}
L(X,Y,H+H_{0})^{2} & =L(X,Y,H)^{2}+\left[\begin{array}{cc}
HH_{0}+H_{0}H+H_{0}^{2} & -\left[H_{0},Z\right]^{\dagger}\\
-\left[H_{0},Z\right] & HH_{0}+H_{0}H+H_{0}^{2}
\end{array}\right].
\end{align*}
We deduce from (\ref{eq:bound_on_H0}) that 
\[
-C|Z|^{2}\leq HH_{0}+H_{0}H+H_{0}^{2}\leq C|Z|^{2}.
\]
From (\ref{eq:bound_on_commutator_H0}) we deduce
\[
\left[\begin{array}{cc}
-E|Z| & 0\\
0 & -E|Z|
\end{array}\right]\leq\left[\begin{array}{cc}
0 & -\left[H_{0},Z\right]^{*}\\
-\left[H_{0},Z\right] & 0
\end{array}\right]\leq\left[\begin{array}{cc}
E|Z| & 0\\
0 & E|Z|
\end{array}\right].
\]
We find 

\begin{align*}
 & L(X,Y,H+H_{0})^{2}\\
 & =\left[\begin{array}{cc}
H^{2}+|Z|^{2} & -\left[H,Z\right]^{*}\\
-\left[H,Z\right] & H^{2}+|Z|^{2}
\end{array}\right]+\left[\begin{array}{cc}
HH_{0}+H_{0}H+H_{0}^{2} & -\left[H_{0},Z\right]^{*}\\
-\left[H_{0},Z\right] & HH_{0}+H_{0}H+H_{0}^{2}
\end{array}\right]\\
 & \geq\left[\begin{array}{cc}
(1-C-E)H^{2}+|Z|^{2} & -\left[H,Z\right]^{*}\\
-\left[H,Z\right] & (1-C-E)H^{2}+|Z|^{2}
\end{array}\right]+\left[\begin{array}{cc}
\left(-C-E\right)|Z|^{2} & 0\\
0 & \left(-C-E\right)|Z|^{2}
\end{array}\right]\\
 & =(1-C-E)B(X,Y,H)^{2}+\left[\begin{array}{cc}
0 & \left(C+E\right)\left[H,Z\right]^{*}\\
\left(C+E\right)\left[H,Z\right] & 0
\end{array}\right]\\
 & \geq(1-C-E)B(X,Y,H)^{2}-\left(C+E\right)D
\end{align*}
and
\begin{align*}
 & L(X,Y,H+H_{0})^{2}\\
 & =\left[\begin{array}{cc}
H^{2}+|Z|^{2} & -\left[H,Z\right]^{*}\\
-\left[H,Z\right] & H^{2}+|Z|^{2}
\end{array}\right]+\left[\begin{array}{cc}
HH_{0}+H_{0}H+H_{0}^{2} & -\left[H_{0},Z\right]^{*}\\
-\left[H_{0},Z\right] & HH_{0}+H_{0}H+H_{0}^{2}
\end{array}\right]\\
 & \leq\left[\begin{array}{cc}
(1+C+E)H^{2}+|Z|^{2} & -\left[H,Z\right]^{*}\\
-\left[H,Z\right] & (1+C+E)H^{2}+|Z|^{2}
\end{array}\right]+\left[\begin{array}{cc}
\left(C+E\right)|Z|^{2} & 0\\
0 & \left(C+E\right)|Z|^{2}
\end{array}\right]\\
 & =(1+C+E)B(X,Y,H)^{2}+\left[\begin{array}{cc}
0 & \left(-C-E\right)\left[H,Z\right]^{*}\\
\left(-C-E\right)\left[H,Z\right] & 0
\end{array}\right]\\
 & \leq(1+C)B(X,Y,H)^{2}+\left(C+E\right)D.
\end{align*}
\end{proof}

\section*{Acknowledgments}

This material is based upon work supported by the National Science
Foundation under DMS 1700102. The author thanks Liang Du for pointing
out the differing use of the constants in the BHZ model. Most of the
computing was done on machines at the Center for Advance Research
Computing at the University of New Mexico. The author is most grateful
to Hermann Schulz-Baldes for allowing Lemma~\ref{lem:error_Delta_H}
to appear here.

\end{document}